\documentclass[12pt]{article}
\usepackage[utf8]{inputenc}

\usepackage[margin=1in]{geometry}
\usepackage{setspace}
\usepackage{graphicx} 
\usepackage[parfill]{parskip} 

\usepackage{tabularx}

\usepackage{adjustbox}
\usepackage{booktabs} 
\usepackage{array} 
\usepackage{bm} 
\usepackage{paralist}
\usepackage{verbatim}  
\usepackage{titling}
\usepackage{upgreek}
\usepackage{cite}
\usepackage{enumitem}
\usepackage{amsmath}
\usepackage{amsthm}
\usepackage{amsfonts}
\usepackage{amssymb}
\usepackage{mathtools}
\usepackage{hhline}
\usepackage{breqn}
\usepackage{accents}
\usepackage{graphicx}
\usepackage{caption}
\usepackage{subcaption}
\usepackage{xfrac}
\usepackage{bbm}
\usepackage{physics}
\usepackage{csvsimple}
\usepackage{hyperref}
\usepackage{multirow}
\usepackage{float}
\usepackage{tikz}
\usepackage{marvosym}
\usepackage[round]{natbib}
\usepackage{pdflscape}
\usepackage[scr=boondox]{mathalfa}
\usepackage{authblk}
\usepackage{xr}
\usepackage{setspace} 
\usepackage{xcolor}
\usepackage{amssymb}

\usepackage{sectsty}
\allsectionsfont{\sffamily\mdseries\upshape} 
\newcommand{\subtitle}[1]{%
  \posttitle{%
    \par\end{center}
    \begin{center}\large#1\end{center}
    \vskip0.5em}%
}

\usepackage[nottoc,notlof,notlot]{tocbibind} 
\usepackage[titles,subfigure]{tocloft} 

\newtheorem{lemma}{Lemma}
\newtheorem{theorem}{Theorem}
\newtheorem*{theorem*}{Theorem}
\newtheorem{corollary}{Corollary}

\theoremstyle{definition}
\newtheorem{definition}{Definition}
\newtheorem{remark}{Remark}

\newtheorem*{remark*}{Remark}
\newtheorem{assumption}{Assumption}
\newtheorem{example}{Example}
\newtheorem*{example*}{Example}

\newtheorem{condition2}{Condition}[section]

\makeatletter
\def\thm@space@setup{%
  \thm@preskip=\parskip \thm@postskip=0pt
}
\makeatother

\newcommand{\R}{\mathbb{R}}

\newcommand{\Z}{\mathcal{Z}}

\newcommand{\E}{\mathbb{E}}

\newcommand{\F}{\mathcal{F}}

\newcommand{\G}{\mathcal{G}}

\title{ Quasi-Bayes in Latent Variable Models \thanks{Email address: \texttt{sid.kankanala@chicagobooth.edu}. I thank Xiaohong Chen, Yuichi Kitamura, and Donald Andrews for their guidance on this project. I am also grateful to St{\'e}phane Bonhomme, Victor Chernozhukov, Tim Christensen, Max Cytrynbaum, Jiaying Gu, Han Hong, Matt Masten, Ulrich M\"{u}ller, Jean-Marc Robin, Veronika Ro\v{c}kov\'{a}, Elie Tamer, Victoria Zinde-Walsh, and several seminar participants for their constructive comments and suggestions.
}}
\author{Sid Kankanala}
\affil{The University of Chicago}
\date{\today} 
\begin{document}
\maketitle

\begin{abstract}

	\noindent
	\normalsize
Latent variable models are widely used to account for unobserved determinants of economic behavior. This paper develops a quasi-Bayes framework to nonparametrically estimate a large class of latent variable models. As an application, we model U.S. earnings from the Panel Study of Income Dynamics (PSID) as the sum of latent permanent and transitory shocks. Simulations illustrate the favorable performance of quasi-Bayes estimators relative to common alternatives.

\vspace{0.5cm}

	\bigskip

\end{abstract}

\newpage

\section{Introduction}
A large class of models explicitly account for the influence of latent unobservables on observed economic behavior. These models are constructed to be informative at the population in that the distribution of observables uniquely identifies the latent distribution of interest. However, when the link between observables and latent heterogeneity is tenuous, the finite sample identifying strength can be very weak. As a consequence, classical nonparametric estimators often display several undesirable properties, such as large finite-sample variability, irregularity and extreme sensitivity to small perturbations of the data and user-selected tuning parameters.

In empirical settings, it is common practice to impose strong distributional restrictions on the unobservables. Intuitively, these restrictions limit the model complexity, and thereby require weaker information content from the data. These approaches are attractive in that they significantly reduce finite sample uncertainty and provide a researcher with a tractable distribution (e.g. Gaussian) to use for counterfactual analysis. On the downside, they often introduce significant misspecification bias if the restrictions are incorrect. This bias is particularly evident in the descriptive evidence and estimates of
non-Gaussian models (\citealp{bonhomme2010generalized}; \citealp*{guvenen2014nature};  \citealp*{arellano2017earnings}).

Motivated by these concerns, this paper proposes a class of nonparametric estimators that arise as solutions to quasi-Bayes decision rules. Traditional quasi-Bayes (\citealp{chernozhukov2003mcmc}) combines a GMM objective function $Q_n(\theta)$   and a prior $\nu(\theta)$ for a finite dimensional parameter $\theta \in \R^k$ to create a quasi-posterior distribution:
\begin{align}
    \label{cherhong} \nu(\theta\:|\:\mathcal{D}_n) \propto e^{- Q_n(\theta)} d \nu(\theta).
\end{align}
In this paper, we generalize the representation in (\ref{cherhong}) to settings where the parameter is an infinite dimensional latent distribution of interest. We maintain the classical quasi-Bayes interpretation in that  $Q_n(.)$ only depends on the observables through a collection of identifying population moments. In this setup, $Q_n(.)$ can be viewed as a generalized GMM objective and $\nu(.)$ is a nonparametric prior on distributions. As is standard, the effect of the prior is negligible in settings with strong identification. However, with weakly informative data and a nearly flat objective function $Q_n(.)$, the prior provides significant regularization and allows researchers to incorporate additional information (e.g. moment and smoothness restrictions) that improve the finite-sample information content.

We propose a theoretically motivated class of priors that are supported on nonparametric Gaussian mixtures. Our focus on this class is partially motivated from the observation that finite Gaussian mixtures are widely used to model rich forms of heterogeneity. As Gaussian mixtures can be efficiently sampled from, they provide a researcher with a convenient framework to perform counterfactual analysis. One consequence of building a framework around this structure is that we also obtain decision rules with particularly desirable theoretical guarantees in settings where the true data generating process closely resembles a (possibly infinite) Gaussian mixture, while remaining viable in the broader fully nonparametric setting. As we illustrate in our analysis, this choice is also computationally attractive in that it greatly facilitates the interplay analysis between the identifying moments and the prior. 

The main theoretical contributions of this paper are as follows. We provide a unified treatment of quasi-Bayes for a large class of nonparametric latent variable models. This includes all models in which characteristic function based moment restrictions are employed to identify the latent distribution. As we illustrate in Section \ref{sec2}, this encompasses a large class of commonly used latent variable models in economics. In each model, we take these identifying restrictions as given and use them to build a moment-based likelihood. We then derive posterior contraction rates for the quasi-Bayes posterior in a variety of metrics (e.g. Wasserstein and $L^2$). As our work is the first quasi-Bayes approach to nonparametric latent heterogeneity, we expect the analysis to be of independent interest and useful towards a wide range of related extensions.

As an application, we use our methodology to model individual earnings data from the Panel Study of Income Dynamics (PSID), revisiting the seminal study in \citet{bonhomme2010generalized}. Consistent with the existing literature (e.g.  \citealp{lilliard1978dynamic}; \citealp{arellano2017earnings}), we find that the distribution of U.S. wage shocks is highly non-Gaussian and leptokurtic. Overall, we find that our approach accurately reproduces the higher-order moments observed in U.S. wage growth data. Finally, we use our framework to examine the influence of permanent and transitory shocks on wage mobility.

The paper is organized as follows. Section \ref{sec2} introduces the quasi-Bayes framework, provides a detailed description of our procedures and outlines the class of models to which our methodology applies. Section \ref{sec3} provides the main assumptions and states the main results. In Section \ref{sec6}, we apply our methodology to study earnings dynamics in U.S wages, using data from the PSID. Section \ref{sec4} verifies the main assumptions for several latent variable models and further develops the main results in a model specific context. Section \ref{sec5} provides simulation evidence on quasi-Bayes estimators relative to common alternatives. Section~\ref{conclus} concludes. 

\subsection{Related Literature} \label{lit}
Latent variable models have a long history in econometrics and statistics. In Section \ref{sec2}, we provide an overview of several commonly used models in the literature. For a more comprehensive overview, see  \citet*{chen2011nonlinear}; \citet{schennach2022measurement}.

\citet{chernozhukov2003mcmc} developed the foundational quasi-Bayes limit theory for parametric models that are strongly identified through a finite set of moments. In this article, we focus on models with latent heterogeneity. Importantly, in such models, the parameter is infinite dimensional, and its recovery is statistically ill-posed. Recently, quasi-Bayes procedures have also been applied in parametric models with weak or non-standard identification (\citealp*{chen2018monte}; \citealp{andrews2022optimal}). As the finite sample information content in latent variable models is usually limited, we view our work as complementary to this line of research and particularly informative for the potential of quasi-Bayes decision rules in nonparametric moment models.

 Our work contributes to a growing literature on Bayesian inference in density estimation and deconvolution. For Bayesian density estimation, see \citet{ghosal2001entropies, ghosal2007posterior, shen2013adaptive}, among many others. For Bayesian estimation in deconvolution with a known error distribution, see \citet{donnet2018posterior};  \citet{rousseau2024wasserstein}. To the best of our knowledge, our paper is the first to develop a quasi-Bayes framework to estimate a density and latent distribution. Within this framework, many of our results also represent the first nonparametric Bayes guarantees for the models under consideration.

\subsection{Notation} \label{notationsec}
The set of Borel probability measures on $\R^d$ is denoted by $\mathcal{B}(\R^d)$. Convolution of a function $g: \R^d \rightarrow \R$ with a measure $\mu \in \mathcal{B}(\R^d) $ is denoted by  $g \star \mu (x) = \int_{\R^d} g(x-z) d \mu (z)$. Given two measures $(\lambda,\mu) $,  we denote the product measure by $\nu = \lambda \otimes \mu$. Let $ \mathbf{H}^{p} = ( \mathbf{H}^p, \| . \|_{\mathbf{H}^p}) $ denote the usual $p$-Sobolev space of functions on $\R^d$.

\section{Models and Procedures} \label{sec2}


\subsection{Framework} \label{model-setup}
Let $Y \in \R^q$ denote a vector of observables. We are interested in the distribution of a latent unobservable $X \in \R^d$ that influences $Y$ through a known economic model. In many natural settings, identification of the latent distribution is established by verifying that the latent characteristic moments $\varphi_{X}(t) = \E[e^{\mathbf{i}t'X}]$, or its derivatives, can be identified through suitable population moments of $Y$. To fix ideas, suppose $\mathcal{G}(.) $ is a known function valued operator. Then, the latent distribution is assumed to satisfy a restriction of the form
 \begin{equation}
    \label{gen-identity} \mathcal{G}( \mathbb{P}_{Y}  ,\varphi_{X})(t) = \mathbf{0} \; \; \; \; \; \; \; \; \; \forall \; t \in \R^v.
\end{equation}
The argument of $\mathcal{G}(.)$ consists of two inputs, a probability distribution $P$ on $Y$ and a characteristic function $\varphi_F$. In most cases, we can view (\ref{gen-identity}) as a generalized moment restriction in that the operator $\mathcal{G}(.)$ only depends on $P$ through a collection of moments on $Y$. We say the restriction is an identifying restriction if \begin{equation*} F \neq F_X \implies \mathcal{G}( \mathbb{P}_{Y}  ,\varphi_{F}) \neq \mathbf{0} .\end{equation*}
In this paper, we focus on a general class of latent variable models for which a generalized moment restriction of the form in (\ref{gen-identity}) is known to hold. As the following examples illustrate, this encompasses a large class of commonly used models in economics. Furthermore, the map $ \varphi_F \rightarrow  \mathcal{G}(\mathbb{P}_Y,\varphi_F) $ is usually an elementary transformation of $\varphi_F$ or its derivatives.


\begin{example}[Classical Measurement Error] \label{ex1}
Consider a classical measurement error model where we observe a random sample from
 \begin{align}
    \label{mdeconv-orig}   W = X + \epsilon \; \; \; \; , \; \; \E[\epsilon] = 0.
\end{align}
Here, \( W \in \mathbb{R}^d \) is an observed vector, \( X \in \mathbb{R}^d \) is a latent vector whose distribution is of interest, and \( \epsilon \) is a nuisance error that is independent of \( X \). Since the individual contributions of \( X \) and \( \epsilon \) cannot be separately identified from an observation of \( W \), identification of  \( F_X \) typically requires some auxiliary information about the distribution \( F_{\epsilon} \). In statistics, it is common to assume that \( F_{\epsilon} \) is fully known. For instance, in the classical empirical Bayes framework \citep*{efron2016empirical,gu2023invidious}, \( \epsilon \) is modeled as a known Gaussian distribution. In settings with differential privacy, \( \epsilon \) is a known Laplace distribution (\citealp{rousseau2024wasserstein}).

As is common in the literature (e.g. \citealp{kato2018uniform};  \citealp{arellano2023recovering}), we consider the weaker setting where a researcher has auxiliary information in the form of a random sample \( \epsilon_1^{\star}, \dots, \epsilon_n^{\star} \stackrel{i.i.d}{\sim} F_{\epsilon} \). If the nuisance error $\epsilon$ is independent of $X$, it follows that  \begin{equation*}
    \label{func-deconv} \E[e^{\mathbf{i}t' W}] = \varphi_{X}(t) \E[e^{\mathbf{i}t' \epsilon}] \; \; \; \; \forall \; t \in \R^d.
\end{equation*}
We view \( Y = (W, \epsilon^{\star}) \) as the observable, and the restriction map is defined as:  
\[
\mathcal{G}(\mathbb{P}_Y, \varphi_X)(t) = \mathbb{E}[e^{\mathbf{i}t' W}] - \varphi_X(t) \mathbb{E}[e^{\mathbf{i}t' \epsilon}].
\]
Similar restrictions arise in a variety of generalizations, including models with heteroskedastic errors \citep{chetty2018impacts}.

\end{example}

\begin{example}[Repeated Measurements] \label{ex2}
Suppose we observe measurements $(Y_1,Y_2)$ from the model:
\begin{align}
    \label{repm-model} &  Y_1 = X + \epsilon_1 \;  ,  \;  \E[\epsilon_1] = 0 ,  \\  & Y_2 = X + \epsilon_2 \;  ,  \;  \E[\epsilon_2] = 0 . \nonumber
\end{align}
Here, \( X \in \mathbb{R} \) is a latent variable of interest, and \( (\epsilon_1, \epsilon_2) \) represent unobserved nuisance errors. In this setup, \( (Y_1, Y_2) \) are interpreted as imperfect or noisy proxies for \( X \). For example, in \citet*{cunha2010estimating}, \( X \) represents latent child ability, while \( (Y_1, Y_2) \) denote test scores in different subjects.

If the nuisance errors satisfy $\E[\epsilon_1|X,\epsilon_2] = 0 $ and $\epsilon_2 \perp X$, it is well known (\citealp{li1998nonparametric}; \citealp*{cunha2010estimating}) that the latent density admits the representation \begin{align}
    \label{koltarski}   f_{X}(x) =  \frac{1}{2 \pi} \int_{\mathbb{R}}  e^{- \mathbf{i} tx} \exp \bigg(  \int_{0}^t  \frac{\E[\mathbf{i} Y_1 e^{\mathbf{i} \zeta Y_2}  ]}{\E[ e^{\mathbf{i} \zeta Y_2}  ]} d \zeta    \bigg) dt.
\end{align}
From the representation in (\ref{koltarski}), it follows that the latent distribution satisfies \begin{equation}
    \label{func-repmeas}  \E[  e^{\mathbf{i} t Y_2}  ] \partial_{t} \log \varphi_{X}(t) = \E[\mathbf{i} Y_1 e^{\mathbf{i} t Y_2}  ]  \; \; \; \; \; \; \; \; \forall\: t \in \R.
\end{equation}
In this case, we view $Y = (Y_1,Y_2)$ as the observable and the restriction map is \begin{equation*}
    \label{func-repmeas2}  \mathcal{G}(\mathbb{P}_Y,\varphi_X)(t) =  \E[  e^{\mathbf{i} t Y_2}  ] \partial_{t} \log \varphi_{X}(t) - \E[\mathbf{i} Y_1 e^{\mathbf{i} t Y_2}  ] .
\end{equation*}
The identifying restrictions in (\ref{func-repmeas}) remain valid even when \(X\) does not admit a density. Beyond applications in labor economics, repeated measurements and variants of the identity in (\ref{koltarski}) have been used to estimate unobserved heterogeneity in auctions (e.g. \citealp{asker2010study}; \citealp{krasnokutskaya2011identification}).

\end{example}

\begin{example}[Random Coefficient Models] \label{ex3}
Consider the classical random coefficient model \begin{align}
    W = Z' \beta
\end{align}
where $W \in \R$ is an observed scalar response, $Z \in \R^d$ is a vector of observed covariates and $\beta \in \R^d$ is a vector of random unobserved coefficients that are independent from $Z$. We are interested in recovering the latent distribution of $\beta$. It is shown in \cite{hoderlein2010analyzing} that the distribution of $\beta$ satisfies  \begin{align*}
    \E[ e ^{ \mathbf{i} t W} | Z=z] = \varphi_{\beta} (tz) \; \; \; \; \; \; \; \; \; \; \; \; \; \; \forall \; \; t \in \R.
\end{align*}
We view $Y = (W,Z)$ as the observable and the restriction map is given by  \begin{equation*}
    \label{func-randc}  \mathcal{G}(\mathbb{P}_Y,\varphi_{\beta})(t,z) =   \E[ e ^{ \mathbf{i} t W} | Z=z] - \varphi_{\beta} (tz). 
\end{equation*}
Similar restrictions have been derived to identify and estimate the distribution of random coefficients in; discrete choice \citep*{gautier2013nonparametric} and simultaneous equation models \citep{masten2018random,hoderlein2017triangular}.

\end{example}

\begin{example}[Multi-Factor Models] \label{ex4}
Consider the linear multi-factor model \begin{align}
    \label{mult-fac}\mathbf{Y}  = \mathbf{A} \mathbf{X}
\end{align}
where $\mathbf{Y} = (Y_1,\dots,Y_L)' \in \R^L$ is a vector of $L$ measurements, $\mathbf{X} = (X_1,\dots,X_K)'$ is a vector of $K$ latent and and mutually independent factors and $\mathbf{A}$ is a known (or estimable) $L \times K$ matrix of factor loadings. Without loss of generality, we assume the factors are demeaned, i.e $\E[\mathbf{X}] = \mathbf{0}$. In \cite{bonhomme2010generalized}, it is shown that  the latent factors  $\{F_{X_k}\}_{k=1}^K$ satisfy \begin{align}
    \label{mult-fac-id1}  \nabla \nabla'  \log \varphi_{\mathbf{Y}} (t)      = \sum_{k=1}^K \mathbf{A}_k \mathbf{A}_k'   (\log \varphi_{X_k})''(t' \mathbf{A}_k)     \; \; \; \; \; \;  \forall \; t \in \R^L \; ,
\end{align}
where $\nabla \nabla' \log \varphi_{\mathbf{Y}} (t) $ denotes the Hessian of the map $ t \rightarrow \log \varphi_{\mathbf{Y}} (t)$. Let $\mathcal{V}_{ \mathbf{Y}   } (t)$ denote the vector of upper triangular elements of $\nabla \nabla' \log \varphi_{\mathbf{Y}} (t) $. Let $\mathcal{V}_{\mathbf{X}}(t) $ denote the vector with elements $  \{ (\log \varphi_{X_i})''(t' \mathbf{A}_i) \}_{i=1}^K  $. Let $V(\mathbf{A}_k)$  denote the vector of upper triangular elements of
$\mathbf{A}_k \mathbf{A}_k'$ and $\mathbf{Q} = [ V(\mathbf{A}_1) , \dots , V(\mathbf{A}_K) ] $ the matrix with columns composed of those vectors. As the matrices in (\ref{mult-fac-id1}) are symmetric, the identifying restrictions may be expressed as \begin{align*} \label{QVX}
    \mathcal{V}_{\mathbf{Y}}(t) = \mathbf{Q} \mathcal{V}_{\mathbf{X}}(t) 
\end{align*}
This restriction is appealing because it provides separate moment conditions for each latent factor $F_{X_k}$. To see this, fix any $k \in \{ 1,2,\dots,K \} $, and let $X_k$ denote the random variable whose distribution $F_{X_k}$ is of interest. If $\mathbf{Q}^* = (\mathbf{Q}'\mathbf{Q})^{-1} \mathbf{Q}'$ and $\mathbf{Q}_k^*$ denotes the $k^{th}$ row of $\mathbf{Q}^*$, we have \begin{align*}
    (\log \varphi_{X_k})''(t' \mathbf{A}_k) =  \mathbf{Q}_k^*\mathcal{V}_{\mathbf{Y}}(t) \; \; \; \;  \forall \; t \in \R^L.
\end{align*}
Since feasible versions of $\mathcal{V}_{\mathbf{Y}}$ typically contain estimated reciprocals of $\varphi_{\mathbf{Y}}^2$, this quantity can exhibit high variability in finite samples. To attenuate the finite sample uncertainty, it is usually more straightforward to work with the restriction: \begin{align*}
    \G(\mathbb{P}_Y,\varphi_{X_k})(t) = \varphi_{\mathbf{Y}}^2 (t) \big[\mathbf{Q}_k^*\mathcal{V}_{\mathbf{Y}}(t) - (\log \varphi_{X_k})''(t' \mathbf{A}_k) ].
\end{align*}

Beyond consumption dynamics, these identifying restrictions have also been used to estimate the distribution of latent beliefs in models with survey elicitation \citep{herbst2024opportunity}.

\end{example}

\begin{remark}[On Nuisance Factors]
\label{remark-nuis} As illustrated in previous examples, the latent distribution of interest can often be constructively identified through a generalized moment restriction on the observables. Such restrictions are particularly attractive for estimation, as they avoid the need to explicitly model and estimate the distribution of nuisance unobservables, which may be very complex.
\end{remark}
The preceding examples and their extensions cover a broad class of latent variable models commonly used in economics, though the list is far from exhaustive. Similar characteristic moment  restrictions arise in a variety of other settings, including models with mismeasured regressors (\citealp{schennach2004estimation,schennach2007instrumental}; \citealp*{ben2017identification}), nonparametric panel data \citep{arellano2012identifying,wilhelm2015identification}, and workhorse models for consumption dynamics \citep{arellano2017earnings,arellano2024heterogeneity}.

In nearly all the settings described above, the standard approach to estimating the latent distribution involves inverting the identifying restrictions. This typically entails solving a finite-sample analog of 
\(\mathcal{G}(\mathbb{P}_Y, \varphi_X) = \mathbf{0}\) 
to obtain an estimate \(\widehat{\varphi}_X\). The estimated latent density \(\widehat{f}_X\) is then constructed using an inverse Fourier transform.\footnote{For applications of this approach to specific models, see \citet{horowitz1996semiparametric}; \citet{bonhomme2010generalized}; \citet{krasnokutskaya2011identification}; \citet{gautier2013nonparametric}; \citet{masten2018random}; \citet{herbst2024opportunity}, among many others.} As noted in the literature (e.g. \citealp*{efron2016empirical}; \citealp*{arellano2023recovering}), these solutions often exhibit significant finite sample variability and are highly sensitive to small perturbations in the data and user-selected tuning parameters. Intuitively, the large finite sample variability of these estimators arises from their representation as an inverse of an ill-posed objective function over an excessively large parameter space.\footnote{These estimators allow for any square integrable function- the effective parameter space is $L^2(\R^d)$. As a consequence, the estimator $\widehat{f}_X$ is often negative on a subset of the domain  and $\int \widehat{f}_X(t) dt \neq 1$.} 

Two natural questions arise from the preceding discussion: $(i)$ Can the information content in the moments be utilized more efficiently for the class of nonparametric latent variable models described above? $(ii)$ Is the identifying strength in the moments significantly amplified in settings where some regularity is imposed on the parameter space? To address these questions, this paper proposes and examines a class of nonparametric estimators that arise as solutions to quasi-Bayesian decision rules.

In parametric models, where a finite-dimensional parameter $\theta$ is identified through finite moment restrictions, the classical quasi-Bayes approach (e.g., \citealp{chernozhukov2003mcmc}) treats a monotonic transformation of the GMM objective function as a likelihood for $\theta$. In our setting, we view the latent distribution $F_X$ as an infinite-dimensional parameter of interest. The identifying restriction $ \G(\mathbb{P}_Y,\varphi_{\mathbf{X}}) = \mathbf{0} $ then motivates a quasi-Bayes likelihood of the form \begin{equation*}
     L^{\star}(F) = \exp \bigg( - \frac{n}{2} \| \mathcal{G}(\mathbb{P}_Y, \varphi_{F})   \|_{L^2}      \bigg).
\end{equation*}
While ideal, this choice is infeasible as the true distribution $\mathbb{P}_{Y}$ is unknown and must be estimated from the data. In this paper, we proceed by replacing $\mathbb{P}_Y$ with the empirical measure $\mathbb{P}_{n,Y} = n^{-1} \sum_{i=1}^n \delta_{Y_i}$. As the preceding examples illustrate, the explicit dependence of $\mathcal{G}(\mathbb{P}_Y, \varphi_F)$ on $\mathbb{P}_Y$ is usually through an elementary transformation of $\varphi_{Y}(t) = \E[e^{\mathbf{i}t' Y}]$. In particular, the quantity $\mathcal{G}(\mathbb{P}_{n,Y}, \varphi_F)(t)$ will typically depend on  $\widehat{\varphi}_{Y}(t) = \E_n[ e^{\mathbf{i}t'Y}]$ and as a consequence, it cannot be consistently estimated uniformly over $t \in \mathbb{R}^q$.\footnote{Under standard regularity conditions, we have $  \sup_{\| t \|_{\infty} \leq T} \left|  \widehat{\varphi}_Y(t) - \varphi_Y(t) \right| = O_{\mathbb{P}}( n^{-1/2} \sqrt{\log T}    )$. This is Lemma 1 in the Appendix. } To control the finite sample uncertainty, we follow the usual approach and consider a bounded set of restrictions that grows with the sample size $n$. Specifically, given any fixed $T > 0$, we denote the ball of radius $T$ and the $L^2$ magnitude (norm) of a function $f : \R^q \rightarrow \mathbb{C}$ over the ball by \begin{align} \label{ball}
   &  \mathbb{B}(T) = \{ t \in \R^q :  \| t \|_{\infty} \leq T  \} \; \; , \; \; \;  \|f \|_{\mathbb{B}(T)}^2 =  \int_{\mathbb{B}(T)}  \| f(t) \|^2 dt .
\end{align}
Denote the observed data by  $\mathcal{D}_n = \{ Y_1,\dots,Y_n \}$. Given a deterministic, or possibly data driven, sequence of constants $T_n \uparrow \infty$  that grows with the sample size $n$, we define the characteristic likelihood by
 \begin{equation} \label{feas-quas}
    L(F) = \exp \bigg( - \frac{n}{2} \| \mathcal{G}( \mathbb{P}_{n,Y}, \varphi_{F})    \|_{\mathbb{B}(T_n)}^2 \bigg).
\end{equation}
When combined with a (possibly data dependent) prior $\nu $ that is supported on a set of distributions $\mathcal{F}$, we obtain the quasi-Bayes posterior: \begin{align}
    \label{qb-general}  \nu(  F  \: | \:  \mathcal{D}_n) =  \frac{  \exp \big( - \frac{n}{2} \| \mathcal{G}(\mathbb{P}_{n,Y}, \varphi_{F})    \|_{\mathbb{B}(T_n)}^2 \big) d \nu(F) }{\int_{\F} \exp \big( - \frac{n}{2} \| \mathcal{G}(\mathbb{P}_{n,Y}, \varphi_{F'})    \|_{\mathbb{B}(T_n)}^2 \big) d \nu(F') }.
\end{align}

%
\begin{remark}[Weighting]  \label{opt-weight}The preceding definition corresponds to identity weighting of the moment conditions. For simplicity of exposition, we focus on this case, although the analysis extends to accommodate optimally weighted moment conditions (see Section~\ref{conclus} for further discussion). In practice, either weighting scheme can be used without difficulty, as the integral in~(\ref{feas-quas}) is typically approximated via Monte Carlo integration using a discrete measure supported on grid points \( t_1, \dots, t_M \in \mathbb{B}(T) \). The resulting discretized objective can be interpreted as a scaled GMM criterion based on \( M \) moment conditions, with either identity or optimal weighting. In the theoretical analysis that follows, we ignore this computational discretization bias and instead focus on guarantees for the full quasi-Bayes posterior.

\end{remark}

\begin{remark}[Tuning]  \label{tchoice}
In frequentist analysis, the choice of \( T = T_n \) is a key tuning parameter that must be carefully calibrated to each specific setting, as it entirely controls the bias-variance tradeoff.
 In particular, classical results are highly sensitive to \( T \), with the variance increasing rapidly as \( T \) grows. In our setting, \( T \) determines the set of identifying restrictions and thus influences the bias. However, the variance is shaped by the interplay between the prior and the identifying restrictions. Intuitively, nonparametric priors significantly attenuate the variance, thereby providing significant protection against poor choices of \( T \). Moreover, this bias-variance tradeoff is always handled in a fully data-driven and adaptive manner through the interplay between the prior and the information content in the moments. In Section~\ref{sec5}, we illustrate this by examing results with a fixed choice of  \( T \) across a wide range of commonly used data-generating processes.
\end{remark}

\subsection{Classes of Priors} \label{priors}
In principle, many distinct families of priors are suitable to construct the quasi-Bayes posterior in (\ref{qb-general}). In Section \ref{sec3}, we provide general results that are applicable to a wide class of priors. However, in the interest of sampling efficiently from the quasi-posterior and using the samples in a tractable form, it is convenient to focus on priors that satisfy two basic properties. First, they have a known and tractable characteristic function. This ensures that the quasi-likelihood in (\ref{feas-quas}) and suitable derivatives of it can be easily evaluated. Second, the prior is supported on distributions that are tractable to sample from. This property is convenient for any post estimation analysis, especially if the overall goal is to estimate complex  functionals of the latent distribution. In the remainder of this section, we describe a class of priors that satisfy both of these requirements.


\subsubsection{Dirichlet Process Gaussian Mixtures (DPGM)} \label{dpgm}
Given a covariance matrix $\Sigma \in \R^{d \times d}$, denote by $\phi_{\Sigma}$ the $N(\mathbf{0},\Sigma)$ density on $\R^d$. Given a probability distribution $P \in \mathcal{B}(\R^d)$, we denote the Gaussian mixture with mixing distribution $P$ and covariance matrix $\Sigma $ by \begin{align}
    \label{gmm}  \phi_{P,\Sigma} (x) = \phi_{\Sigma} \star P (x) = \int_{\R^d} \phi_{\Sigma}(x-z) d P(z).
\end{align}
If $P$ is  a discrete distribution, i.e $P = \sum_{j=1}^{\infty} p_j \delta_{\mu_j} $ which assigns probability mass $p_j$ to a point $\mu_j \in \R^d$, the preceding definition reduces to \begin{align*}
    \phi_{P,\Sigma}(x) = \sum_{j=1}^{\infty} p_j \phi_{\Sigma}(x- \mu_j).
\end{align*}
Let $\varphi_F(t) = \int e^{\mathbf{i}t' x} d F(x)$ denotes the characteristic function (CF) of a distribution $F$. If $F = \phi_{P,\Sigma}$, we denote it by $\varphi_{P,\Sigma}$. By elementary properties of the CF, we have  \begin{align} \label{gmmcf}  \varphi_{P,\Sigma}(t) = \varphi_{\phi_{\Sigma} \star P}(t) =  \varphi_{\phi_{\Sigma}}(t) \times \varphi_{P}(t) = e^{- t' \Sigma t /2} \varphi_{F}(t).   \end{align}
For a discrete distribution $P = \sum_{j=1}^{\infty} p_j \delta_{\mu_j} $, the preceding expression simplifies to
\begin{align} \label{simplecf}   \varphi_{P,\Sigma}(t)  =    e^{- t' \Sigma t /2} \sum_{j=1}^{\infty} p_j e^{\mathbf{i} t' \mu_j}     .     \end{align}
A Gaussian mixture is completely characterized by its mixing distribution $P$ and covariance matrix $\Sigma$. In particular, a prior on $(P,\Sigma)$ is equivalent to placing a prior on infinite Gaussian mixtures. For a prior on $\Sigma$, a common choice is the Inverse-Wishart distribution. In Section \ref{sec4}, we show that a large class of covariance priors can be accommodated for the models considered in this paper. For the mixing distribution $P$, we focus on the canonical choice of a prior for a distribution, the Dirichlet process.

\begin{definition} \label{dp0}
A random distribution $P$ on $\R^d$ is a Dirichlet process distribution with base measure $\alpha$ if for every finite measurable partition of $\R^d = \bigcup_{i=1}^k A_i$, we have \begin{align}
    \label{dp01} (P(A_1),P(A_2), \dots , P(A_k)) \sim  \text{Dir}(k,\alpha(A_1),\dots,\alpha(A_k))\:,
\end{align}
where $\text{Dir}(.)$ is the Dirichlet distribution with parameters $k$ and $\alpha_1(A_1), \dots , \alpha_k(A_k) > 0$ on the $k$ dimensional unit simplex $\Delta = \{ x \in \R^k : x_i \geq 0 \; , \sum_{i=1}^k x_i =1 \}$.\footnote{The probability density of $\text{Dir}(k,\alpha_1,\dots,\alpha_k)$, with respect to the $k-1$ dimensional Lebesgue measure on the unit simplex $\Delta$, is proportional to $ \prod_{i=1}^k x_i^{\alpha_i-1}  $.} 
\end{definition}
An alternative definition, particularly useful for simulation, is given below.
\begin{definition}[Dirichlet process] \label{dp1}
Fix $\beta > 0$ and a base probability distribution $\alpha$ on $\R^d$. Let $V_1,V_2 , \dots \stackrel{i.i.d}{\sim} \text{Beta}(1, \beta) $  and $\mu_1,\mu_2 , \dots \stackrel{i.i.d}{\sim} \alpha $. Define the weights \begin{align*} & p_1 = V_1 \; \; , \;  p_j = V_j \prod_{i=1}^{j-1} (1 - V_i) \; \; \; \; \; \; \;  j \geq 2. \end{align*}
It can be shown \citep[Lemma 3.4]{ghosal2017fundamentals} that $\sum_{i=1}^{\infty} p_i = 1$ almost surely. The DP with concentration parameter $\beta$ and base measure $\alpha$ is the law of the discrete measure: $$ P = \sum_{i=1}^{\infty} p_i \delta_{\mu_i} \:, $$
where $\delta_{\mu_i}$ is the point mass at $\mu_i$. That is, $P$ takes value $\mu_i$ with probability $p_i$. The parameter  $\beta$ controls the spread around the base measure $\alpha$, with larger values of $\beta$ leading to tighter concentration. For computation, it is common to truncate the infinite sum at a large integer $K$.\footnote{All our main results remain valid if $K \asymp \sqrt{n}$. In settings with moderate regularity, $K$ can be taken much smaller; see, e.g., \citep[Proposition 4.20]{ghosal2017fundamentals}.} For notational simplicity, we suppress the dependence on $\beta$ and write $P \sim \mathrm{DP}_{\alpha}$.
\end{definition}
 A Dirichlet process prior can be used to build a prior on nonparametric Gaussian mixtures. Specifically, given a Dirichlet process prior $\text{DP}_{\alpha}$ and an independent prior $G$ on the set  of positive definite matrices, the induced DP-GM prior on Gaussian mixtures is \begin{align}
    \label{dpgmm} \phi_{P,\Sigma} (x)     =\int_{\R^d} \phi_{\Sigma}(x-z) d P(z) \;  , \; \; \;      (P, \Sigma) \sim \text{DP}_{\alpha} \otimes G.
\end{align}
An advantage of focusing on nonparametric Gaussian mixture priors is that they greatly facilitate the interplay analysis between the prior and the moments. As shown in (\ref{simplecf}), Gaussian mixtures admit a simple characteristic function, which in turn leads to a tractable quasi-likelihood. Importantly, priors on Gaussian mixtures translate seamlessly to priors on characteristic moments and vice versa.

An informative prior may set $\alpha$ to be a reasonable guess for the latent distribution, for instance, via a preliminary estimation based on a parametric model. A theoretical discussion of this is provided in Remark \ref{data-dep-prior}.
In the remaining sections, we will frequently refer to the prior in (\ref{dpgmm}), which is also the prior used in our application and simulations. For notational convenience, we denote this product prior by $\nu_{\alpha,G} = \text{DP}_{\alpha} \otimes G$.

\section{Theory} \label{sec3}
We begin by outlining the main regularity conditions in Section \ref{main-ass}. Section \ref{main-res} presents and discusses the main results for the general quasi-Bayes posterior in (\ref{qb-general}). As our framework accommodates a wide range of models, we state the conditions in a suitably general form. In Section \ref{sec4}, we verify these assumptions for the models in Section \ref{sec2} using low-level conditions and discuss their main implications in model specific detail.
\subsection{Assumptions} \label{main-ass}
We state and  discuss the assumptions that we impose on the model and prior. Throughout this section, let $\nu_n$ denote a, possibly data dependent, prior that is supported on a set of distributions $\F_n \subseteq \mathcal{B}(\R^d)$. Let $T_n \uparrow \infty$ denote a deterministic sequence of positive constants. Let $(\epsilon_n)_{n=1}^{\infty} $ denote a deterministic sequence of positive constants that converge to zero at a slower than parametric rate : $\epsilon_n \downarrow 0 $ and $n \epsilon_n^2 \uparrow \infty$. 

\begin{assumption}[Sampling Uncertainty] \label{sampling-uncert} 
For some set $\mathcal{S}_n \subseteq \F_n$ and a universal constant $D > 0 $, we have \begin{align*}
     \mathbb{P} \bigg( \sup_{F \in \mathcal{S}_n} \| \mathcal{G}(\mathbb{P}_{n,Y}, \varphi_{F})  - \mathcal{G}(\mathbb{P}_{Y}, \varphi_{F})    \|_{\mathbb{B}(T_n)} > D \epsilon_n  \bigg) \rightarrow 0.
\end{align*}
\end{assumption}
Assumption \ref{sampling-uncert} provides bounds on the sampling uncertainty that arises from the true population distribution  $\mathbb{P}_{Y}$ being unknown. Typically, $\mathcal{S}_n$ represents a ball (in an suitable metric) centered around a fixed distribution $F_n$. In certain settings, such as Examples \ref{ex2} and \ref{ex4}, $\mathcal{S}_n$ may also include additional constraints (e.g. moment bounds) on $F$, which help in controling the sampling uncertainty.

\begin{assumption}[Weak Bias] \label{weak-bias} 
For some probability measure $F_n \in \F_n$ and a universal constant $D > 0 $, we have  \begin{align*}
    \| \mathcal{G}(\mathbb{P}_{Y}, \varphi_{F_n})  - \mathcal{G}(\mathbb{P}_{Y}, \varphi_{X})    \|_{\mathbb{B}(T_n)} \leq D \epsilon_n
\end{align*}
Assumption \ref{weak-bias} imposes bounds on the bias between the latent distribution $F_X$ and a distribution $F_n$ within the support of the prior. The distribution $F_n$ can be interpreted as one that minimizes the distance to $F_X$ with respect to the objective function induced weak metric $ d_{\mathcal{G}}(F,F_X) =  \| \mathcal{G}(\mathbb{P}_{Y}, \varphi_{F}) - \mathcal{G}(\mathbb{P}_{Y}, \varphi_{X}) \|_{\mathbb{B}(T_n)}$.  In some settings, it is natural to set $F_n = F_X$ when the latent distribution $F_X$ shares structural similarities with the prior family. An important setting where this arises is when a researcher models the latent distribution as a nonparametric Gaussian mixture. In this case, there exists a positive definite matrix $\Sigma_0 \in \R^{d \times d}$ and a mixing distribution $F_0$ such that the density of $X$ is: \begin{align}
  \label{fxgaussm}  f_X(x) =   \phi_{\Sigma_0} \star F_0 (x)  = \int_{\R^d} \phi_{\Sigma_0}(x-z) d F_0(z).
\end{align}
If the prior is supported on a sufficiently rich subset of Gaussian mixtures (e.g., the DP-GM prior discussed in Section \ref{dpgm}), the bias is typically negligible at $f_X$.
\end{assumption}

\begin{assumption}[Local Concentration] \label{loc-conc}
Let $\mathcal{S}_n$ and $F_n$ be as in Assumption \ref{sampling-uncert} and \ref{weak-bias}. For some set $\mathcal{R}_n \supseteq \mathcal{S}_n$, we have  \begin{align*}
    & (i) \; \; \; \; \; \; \nu_n(F \in \mathcal{R}_n :  \| \mathcal{G}(\mathbb{P}_{Y}, \varphi_{F})  - \mathcal{G}(\mathbb{P}_{Y}, \varphi_{F_n})    \|_{\mathbb{B}(T_n)} \leq D \epsilon_n  ) \geq c \exp(- C' n \epsilon_n^2) \\ & (ii) \; \; \; \; \; \nu_n( \mathcal{R}_n \setminus \mathcal{S}_n ) \leq C \exp( - B_n  )
\end{align*}
where $c,C,C',D > 0$ are universal constants and $B_n \uparrow \infty$ is any sequence of constants which satisfies $n \epsilon_n^2 = o(B_n)$.
\end{assumption}
The set \( \mathcal{R}_n \) in Assumption \ref{loc-conc} is introduced to provide some flexibility when direct verification of a local concentration bound is challenging for the $\mathcal{S}_n$ in Assumption \ref{sampling-uncert}. In such cases, $\mathcal{R}_n$ relaxes certain restrictions (e.g. moment bounds) imposed on  \( \mathcal{S}_n \). Assumption \ref{loc-conc}\((ii)\) further requires that the subset of \( \mathcal{R}_n \) where these restrictions fail to hold is sufficiently negligible. Typically, \( \mathcal{R}_n \) is a small ball (in a suitable metric) around $F_n$. Assumption \ref{loc-conc}\((i)\) then imposes a standard small ball local concentration condition on the prior.
\subsection{Contraction} \label{main-res}
In this section, we verify that the quasi-Bayes posterior in (\ref{qb-general}) asymptotically concentrates on local neighborhoods of the latent distribution. 
\begin{theorem}[Weak Contraction]
\label{main-contract} Suppose Assumptions \ref{sampling-uncert}-\ref{loc-conc} hold with a sequence $\epsilon_n \rightarrow 0$. Then, there exists a universal constant $L > 0$ such that \begin{equation}
    \label{contraction-main}  \nu_n \big ( F : \| \mathcal{G}(\mathbb{P}_{n,Y}, \varphi_{F})  - \mathcal{G}(\mathbb{P}_{Y}, \varphi_{X})    \|_{\mathbb{B}(T_n)}  > L \epsilon_n  \: \big| \: \mathcal{D}_n  \big ) \xrightarrow{\mathbb{P}} 0.
\end{equation}
\end{theorem}
Theorem \ref{main-contract} establishes contraction with the respect to the weak  metric $$d_{\G} (F,F_X) = \| \mathcal{G}(\mathbb{P}_{n,Y}, \varphi_{F})  - \mathcal{G}(\mathbb{P}_{Y}, \varphi_{X})    \|_{\mathbb{B}(T_n)}  .$$The interpretation of this convergence varies across models. Typically, $\epsilon_n \downarrow 0$ at a fast rate, and the contraction can be used to deduce fast rates for a class of well-behaved functionals of the latent distribution (see Remark \ref{functionals} and the discussion in Section \ref{conclus}). More generally, it can be interpreted as a preliminary contraction for a suitable “direct problem”, which can then be leveraged to obtain convergence in a stronger distance. In particular, if (\ref{contraction-main}) holds and the bulk of the posterior mass is contained in a well-behaved subset, it is often possible to deduce results in a stronger metric like   $ d_2(F,F_X) = \|  f- f_X  \|_{L^2} $. To fix ideas, given a distance metric $d(.)$ and a class of distributions $\mathcal{H}_n$, we define the modulus of continuity by \begin{equation*}  \omega_n( d, \mathcal{H}_n, \epsilon) = \sup \{ d(F,F_X) :  F \in \mathcal{H}_n, \: \| \mathcal{G}(\mathbb{P}_{n,Y}, \varphi_{F})  - \mathcal{G}(\mathbb{P}_{Y}, \varphi_{X})    \|_{\mathbb{B}(T_n)} \leq \epsilon          \}  .    \end{equation*}
The modulus of continuity is frequently used to characterize the convergence rate in inverse problems (e.g. \citealp{chen2012estimation}). The following result is a straightforward consequence of Theorem \ref{main-contract}.
\begin{corollary}[Contraction]
\label{contract-strong}
    Assume the hypotheses of Theorem \ref{main-contract} hold.  Let $\mathcal{H}_n$ be any subset of distributions such that, for some constants $C>0$ and sufficiently large $D'>0$, $$ \nu_n(F \notin \mathcal{H}_n : \| \mathcal{G}(\mathbb{P}_{n,Y}, \varphi_{F})  - \mathcal{G}(\mathbb{P}_{Y}, \varphi_{X})    \|_{\mathbb{B}(T_n)} \leq L \epsilon_n )  \leq  C \exp(-D' n \epsilon_n^2) $$ 
 holds with $\mathbb{P}$ probability approaching $1$. Then  \begin{equation}
       \nu_n \big( F : d(F,F_X) > \omega_n(d,\mathcal{H}_n,L\epsilon_n  ) \: \big| \:  \mathcal{D}_n   \big) \xrightarrow{\mathbb{P}} 0.
    \end{equation}
\end{corollary}
The constant $D'$, which regulates the decay of mass on $\mathcal{H}_n^c$, is required to be larger than some of the preceding constants that appear in Assumption \ref{sampling-uncert} - \ref{loc-conc}. Usually, the set $\mathcal{H}_n$ is chosen as a function of $D'$ so as to ensure the desired bound holds trivially.

Corollary \ref{contract-strong} provides contraction rates in terms of the modulus $\omega_n(d,\mathcal{H}_n,L \epsilon_n)$. Intuitively, as $\omega_n \rightarrow 0$, the posterior concentrates on local neighborhoods of the latent distribution. The posterior mean naturally serves as an estimator, while the variability in the samples provides a measure of uncertainty. 

In Section~\ref{sec4}, we verify the main assumptions for several latent variable models using the DP-GM prior introduced in Section~\ref{dpgm}, and we further develop the main results in a model-specific context. In particular, we establish upper bounds on the modulus \( \omega_n(\cdot) \) and derive posterior contraction rates under a broad class of metrics, including the Wasserstein and \( L^2 \) distances. As we demonstrate in Section~\ref{sec4}, many of our results provide the first nonparametric Bayes guarantees for the models under consideration.

\section{Application: Earnings Dynamics} \label{sec6}
In this section, we apply our methodology to study the latent structure of permanent and transitory components in the Panel Study of Income Dynamics (PSID) data. Our analysis closely follows the seminal work in \citet{bonhomme2010generalized}, in which deconvolution-based methods were used to nonparametrically estimate the latent structure. 

Our objective in this section is to demonstrate the viability of quasi-Bayes procedures in a non-trivial empirical setting. The dataset consists of a panel of 624 individuals observed over 10 time periods. The model contains 17 latent distributions. Given the relatively small sample size and the presence of several nonparametric latent distributions, it is inevitable that results will exhibit some dependence on tuning parameters. Indeed, as the empirical analysis in \citet{bonhomme2010generalized}; \citet{arellano2017earnings} illustrate, even with carefully chosen tuning parameters, obtaining reasonable estimates through direct frequentist estimation is a very challenging task. 
\subsection{Data and Model} \label{6.1}
We begin with a brief summary of the data and model. The data is from the PSID, between 1978 and 1987. Let $y_{i,t}$ denotes annual log earnings for individual $i$ at time period $t$ and $x_{i,t}$ an associated set of regressors. The regressors are a quadratic polynomial in age and indicators for education, race, geography and year. The OLS residuals of $ y_{i,t} $ on $  x_{i,t} $ are denoted by $ w_{i,t}$. The residual differences are denoted by $\Delta w_{i,t} = w_{i,t} - w_{i,t-1}$. After restricting the sample to male workers with no missing observations for $ \Delta w_{i,t}$ and a wage growth that does not exceed $150 \%$ in absolute value, our sample size consists of $n=624$ individuals. For each individual, we observe wages between 1978 and 1987, for a total of $M=10$ time periods. The standard model is given by \begin{align}
    \label{bonmodel} & w_{i,t} = f_i + w_{i,t}^P + w_{i,t}^T   \:, \; \; \; \; \; \; i=1,\dots,n \; , \; t=1,\dots,M \\ &  w_{i,t}^P = w_{i,t-1}^P + \epsilon_{i,t} \:,  \nonumber \\ & w_{i,t}^T =  \eta_{i,t} \nonumber \\ & \eta_{i,1} = \eta_{i,M} = 0. \nonumber
\end{align}
Here, $f_i$ is an individual-level fixed effect, and $\{\epsilon_{i,t} \}_{t=1}^M$, $\{\eta_{i,t} \}_{t=2}^{M-1}$ are mean-zero latent factors. The model in (\ref{bonmodel}) and various generalizations of it are widely used to study income dynamics (\citealp{hall1982sensitivity};  \citealp{abowd1989covariance}). We think of $w_{i,t}^P$ as the permanent component  and $\eta_{i,t}$ as the transitory component. Thus, the distribution of permanent income is determined by  $\epsilon_{i,t}$ and the  distribution of transitory income by  $\eta_{i,t}$. From first differencing the model in (\ref{bonmodel}), we can write \begin{align}
    \label{bonmodel2} \Delta w_{i,t} =  \epsilon_{i,t} + \eta_{i,t} - \eta_{i,t-1}  \:, \; \; \; \; \; \; i=1,\dots,n \; , \; t=2,\dots,M .
\end{align} We view $ \mathbf{Y}_i = (\Delta w_{i,2}, \dots , \Delta w_{i,M} ) $ as the observations for $i=1,\dots,n$. The model in (\ref{bonmodel2}) is a special case of the multi-factor model $ \mathbf{Y}_i = \mathbf{A} \mathbf{X}_i $ in Section \ref{4.3}. In this setting, the latent factors are $$ \mathbf{X}_i = \{  \eta_{i,2},  \dots , \eta_{i,M-1} ,   \epsilon_{i,2} , \dots , \epsilon_{i,M}   \}.$$ As noted in the literature (e.g. \citealp{geweke2000empirical};  \citealp{bonhomme2010generalized}), assuming a Gaussian model for the latent factors implies a distribution for the wage growth $\Delta w$ that is inconsistent with the observed higher-order moments (e.g. kurtosis) in the data. This raises the concern that some components of the wage shocks $\mathbf{X}_i$ may be non-Gaussian. Over $M=10$ time periods, there are $17$ (the dimension of $\mathbf{X}_i$) latent distributions that must be estimated. All implementation details are provided in the Appendix.

\subsection{Analysis} \label{6.3}

\begin{figure}[H]
    \centering
    \begin{subfigure}[b]{0.48\textwidth}
        \includegraphics[width=\textwidth]{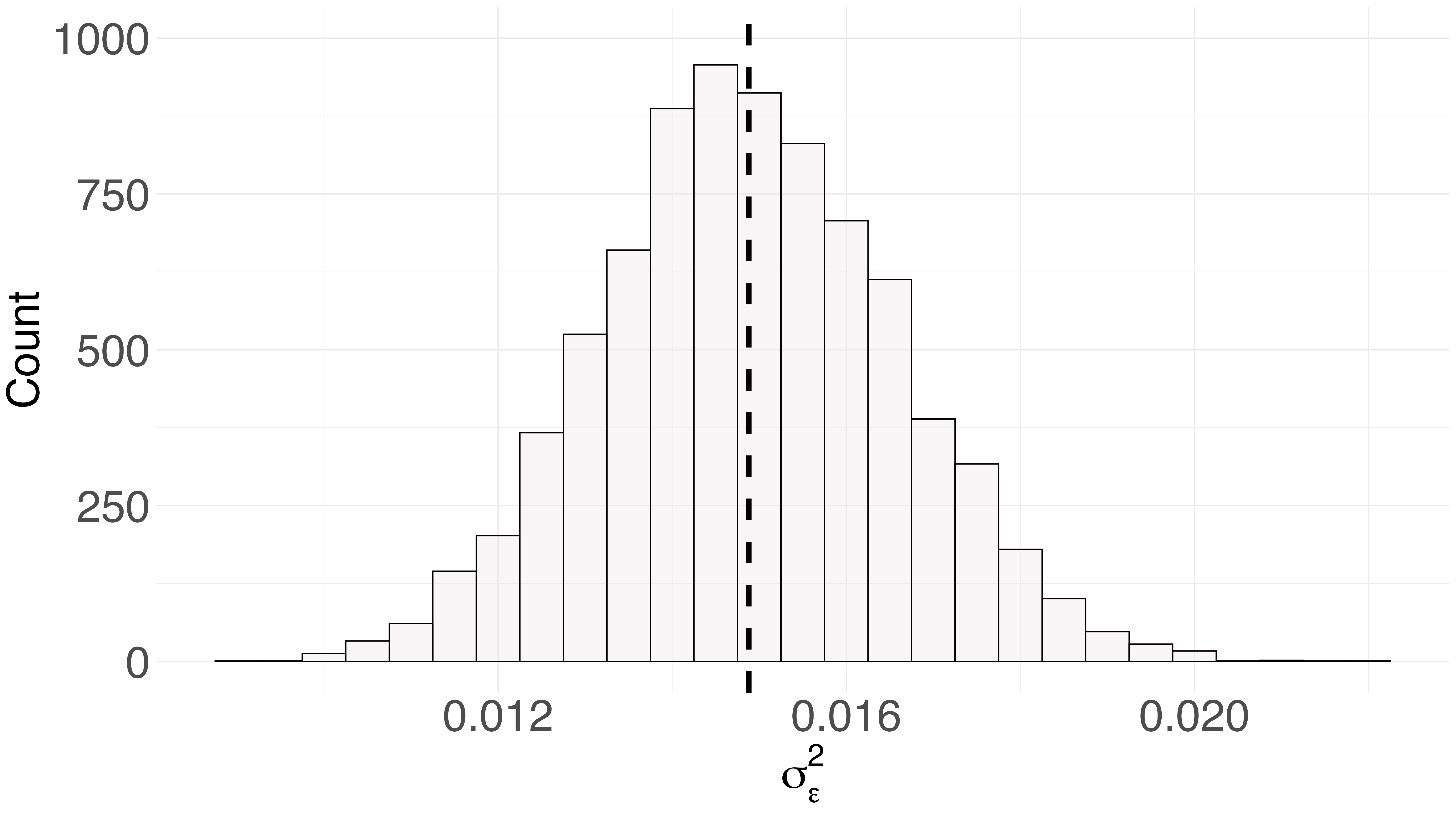}
        \caption{Permanent shock $\epsilon$}
        \label{eps-hist}
    \end{subfigure}
    \hfill
    \begin{subfigure}[b]{0.48\textwidth}
        \includegraphics[width=\textwidth]{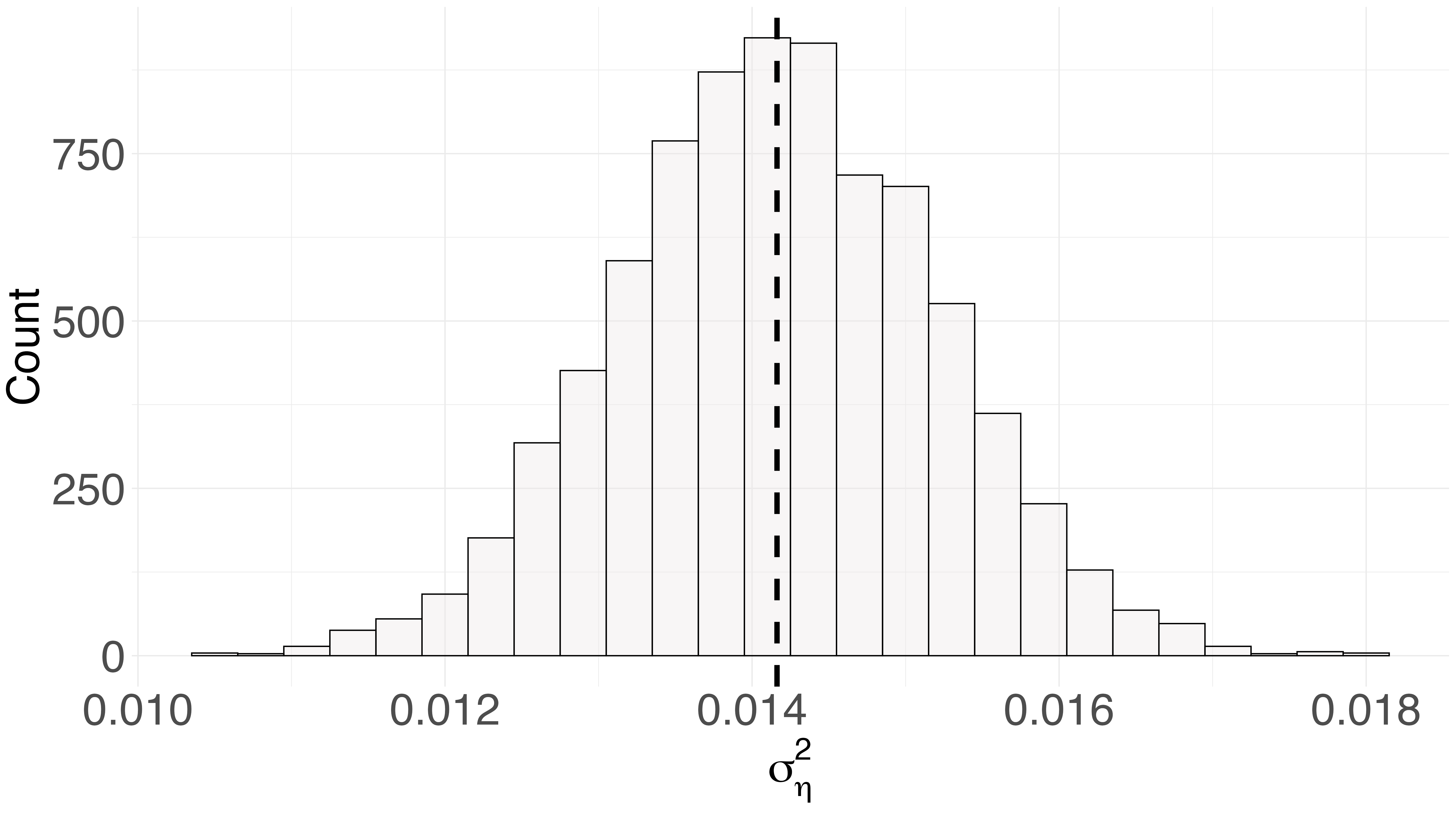}
        \caption{Transitory shock $\eta$}
        \label{eta-his}
    \end{subfigure}
    \caption{Posterior distribution of the average $\sigma_{\epsilon}^2$ and $\sigma_{\eta}^2$ across the time period. }
    \label{eta-eps-his}
\end{figure}

\begin{figure}[H]
    \centering
    \begin{subfigure}[b]{0.48\textwidth}
        \includegraphics[width=\textwidth]{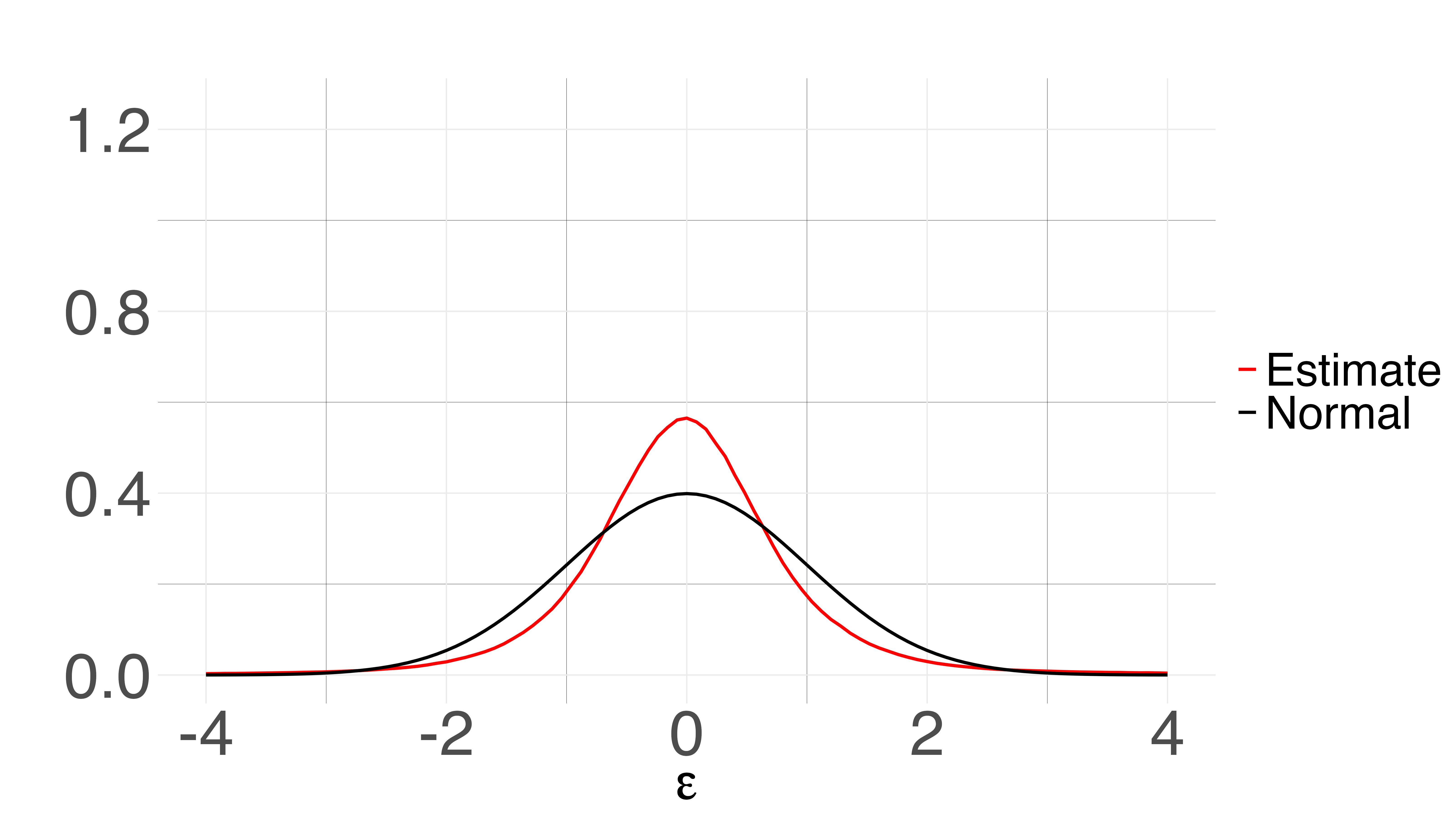}
        \caption{Permanent shock $\epsilon$}
        \label{eps-fig}
    \end{subfigure}
    \hfill
    \begin{subfigure}[b]{0.48\textwidth}
        \includegraphics[width=\textwidth]{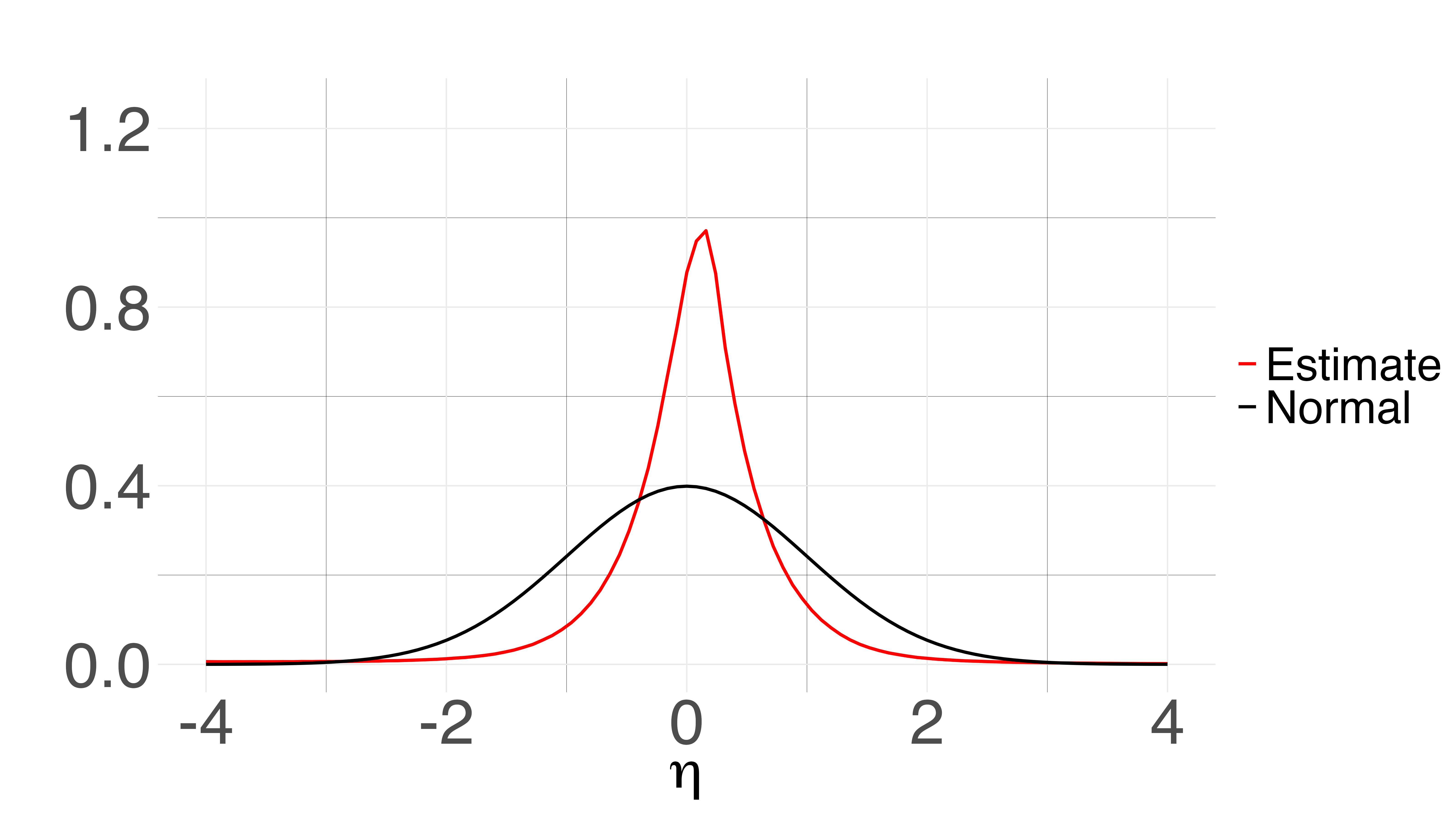}
        \caption{Transitory shock $\eta$}
        \label{eta-fig}
    \end{subfigure}
    \caption{Posterior mean of the standardized time average of $f_{\eta}$ and $f_{\epsilon}$,  with standard Gaussian overlay.   }
    \label{eta-eps-fig}
\end{figure}

\begin{table}[H]
    \centering
    \setlength{\tabcolsep}{8pt} 
    \begin{tabular}{lccc}
        \toprule
        \multicolumn{4}{c}{Wage growth $ \Delta w_{i,t}$ } \\
        \midrule
        & \textbf{Data} & \textbf{Normal} & \textbf{Quasi-Bayes} \\
        \midrule
        Variance & 0.055 & 0.057 & 0.041 \\
        Skewness & 0.001 & 0.00 & 0.00 \\
        Kurtosis & 10.158 & 3.000 & 7.370 \\ 
        \bottomrule
    \end{tabular}
    \caption{Model implied time averaged moments of $\Delta w_{i,t}$}
    \label{wage-mom}
\end{table}

Figure~\ref{eta-eps-his} illustrates the posterior distribution of the latent variance structure, with the solid black dashed line indicating the posterior mean. The posterior mean estimates of the variance are $\widehat{\sigma}_{\eta}^2 = 0.014$ and $\widehat{\sigma}_{\epsilon}^2 = 0.015$. The corresponding posterior means of the kurtosis are $\widehat{\kappa}_{\eta} = 20.46$ and $\widehat{\kappa}_{\epsilon} = 12.19$. Figure~\ref{eta-eps-fig} plots the posterior mean of the average standardized latent density, with a standard Gaussian density overlaid for comparison. These results indicate that the latent factors are non-Gaussian, symmetric, and leptokurtic.

 Noticeably, the transitory component shows a larger degree of kurtosis and departure from standard Gaussianity. These distributional observations are consistent with previous studies of U.S earning growth data (e.g. \citealp*{arellano2017earnings}; \citealp*{guvenen2021data}). Furthermore, we note that quasi-Bayes estimates do not contain the extreme tail oscillations that frequently arise in deconvolution-based estimators (e.g. \citealp{horowitz1996semiparametric}; \citealp{bonhomme2010generalized}). This suggests that such oscillations originate from the choice of estimator, rather than from the information content in the moments. 

 Table \ref{wage-mom} compares the observed moments of the wage growth residuals, $\Delta w$, with the model implied moments. Here, $\textbf{Normal}$ indicates a maximum likelihood estimate with normal latent factors. Since a normal distribution for the latent structure implies normality in the wage growth residuals, $\Delta w$, the model-implied kurtosis substantially underestimates the observed value. In contrast, quasi-Bayes does not match the variance exactly but achieves a significantly better fit for the higher moments. These higher moments are particularly important for counterfactuals that arise in consumption and earnings dynamics.

 Since our posterior is supported on nonparametric Gaussian mixtures, posterior samples can be easily utilized for any downstream application.  As an illustration, we examine the influence of permanent and transitory shocks on wage mobility. For each period $t$, this influence can be characterized by the functionals: 
\begin{align*}
\gamma_{t,\epsilon}(s) &= \mathbb{E} \big[ \epsilon_{i,t} \,\big|\, \Delta w_{i,t} = s \big], \\
\gamma_{t,\eta}(s) &= \mathbb{E} \big[ \eta_{i,t} - \eta_{i,t-1} \,\big|\, \Delta w_{i,t} = s \big].
\end{align*}
These functionals resemble some of the conditional means that arise in Gaussian measurement error models within the empirical Bayes literature (e.g. \citealp{gu2023invidious}). In our framework, they can be estimated in general multifactor models, including settings with non-Gaussian latent factors.

Our objective is to examine how permanent and transitory components contribute to wage mobility. For instance, are large wage changes driven predominantly by transitory or permanent shocks? Moreover, at what threshold does this change? Averaging over time, we define $\bar{\gamma}_{\epsilon}(s) = \frac{1}{M-1} \sum_{t=2}^M \gamma_{t,\epsilon}(s)$ and $\bar{\gamma}_{\eta}(s) = \frac{1}{M-1} \sum_{t=2}^M \gamma_{t,\eta}(s)$.
\begin{figure}[htbp]
  \centering
  \includegraphics[width=\textwidth]{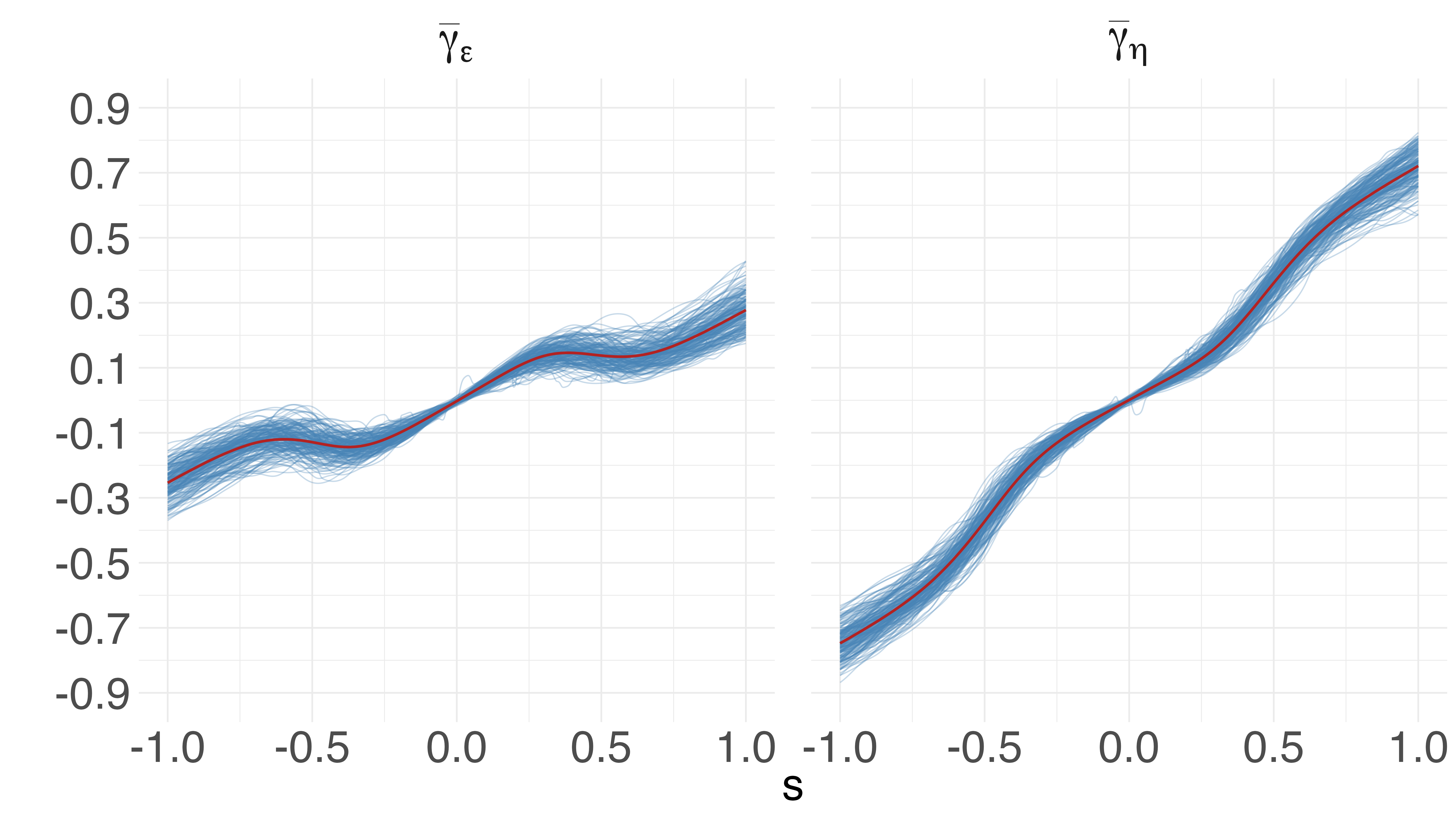}
  \caption{Posterior distribution of $\bar{\gamma}_{\epsilon}(s)$ and $\bar{\gamma}_{\eta}(s)$. Solid red lines mark the posterior means.}
  \label{posterior-dists}
\end{figure}

In Figure \ref{posterior-dists}, we see that the posterior distributions of $\bar{\gamma}_{\epsilon}$ and $\bar{\gamma}_{\eta}$ are tightly concentrated around their posterior means. The convergence of these functionals depends only on contraction in the weak metric (Theorem~\ref{main-contract}) and is therefore much faster than direct recovery of the latent factors.\footnote{For further discussion, see Remark~\ref{functionals} and the discussion in Section \ref{conclus}.} As a consequence, estimation of these functionals is more precise and posterior  uncertainty is low. At $s = 1$, the posterior means are  
\[
\E\big[\bar{\gamma}_{\epsilon}(1) \,\big|\, \mathcal{D}_n\big] = 0.279
\quad\text{and}\quad
\E\big[\bar{\gamma}_{\eta}(1) \,\big|\, \mathcal{D}_n\big] = 0.721
\]  
Thus, a log wage growth of \(+100\%\) is expected to be of transitory origin for \(72.1\%\) of its value. Figure~\ref{posmeanz} presents the transitory--permanent decomposition for the  entire wage mobility path, covering growth rates between \(\pm 100\%\). In particular, moderate to large wage changes are predominantly of transitory origin.

\begin{figure}
  \centering
  \includegraphics[width=0.9\textwidth]{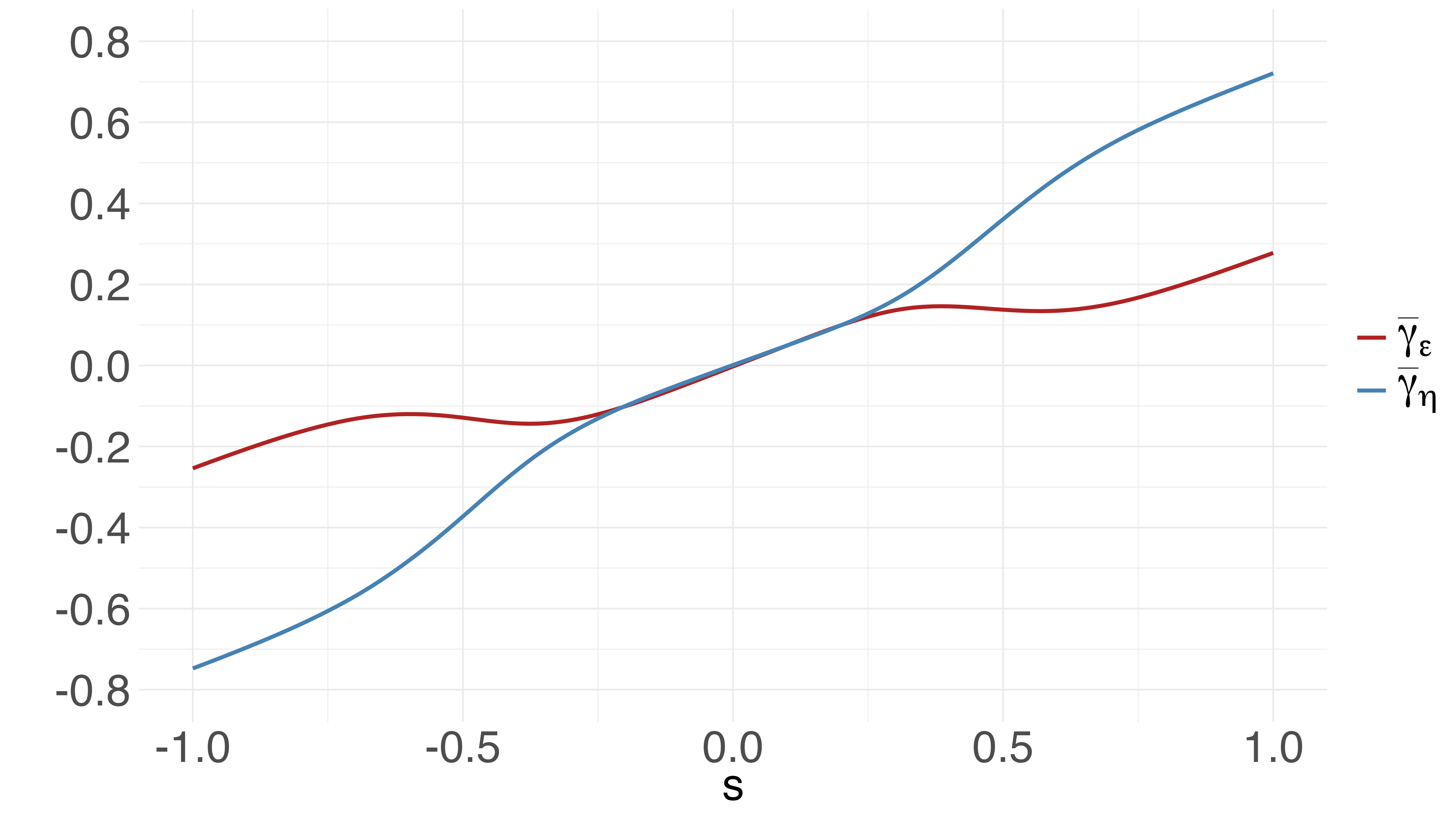}
  \caption{Posterior means of $\bar{\gamma}_{\epsilon}(s)$ and $\bar{\gamma}_{\eta}(s)$.}
  \label{posmeanz}
\end{figure}

\section{Main Results} \label{sec4}
In this section, we verify the main assumptions outlined in Section \ref{sec3} and explore their implications for various latent variable models. In all cases, we focus on the DP-GM prior discussed in Section \ref{dpgm}. In this case, the prior is supported on nonparametric Gaussian mixtures through the mixture representation \begin{align}
    \label{dpgmm2} \phi_{P,\Sigma} (x)     = \int_{\R^d} \phi_{\Sigma}(x-z) d P(z) \;  , \; \; \; \; \; \; \;     (P, \Sigma) \sim \text{DP}_{\alpha} \otimes G.
\end{align}
This induces a prior $\nu$ on the space of probability distributions on $\R^d$. For ease of notation, we will reference the induced prior using the same notation as above, i.e $ \nu =  \text{DP}_{\alpha} \otimes G$. Before proceeding further, we state and discuss the main assumptions that we impose on the Dirichlet process base measure $\alpha$ and covariance prior $G$.

\begin{assumption}[Dirichlet Prior] \label{dir-prior}  The Dirichlet process $\text{DP}_{\alpha}$, in the sense of Definition \ref{dp1}, arises from a finite Gaussian mixture base measure $\alpha$. That is,  $\alpha = \sum_{j=1}^K p_j \mathcal{N}(\mu_j^*, \Sigma_j^*) $ for some $K \in \mathbb{N}$, mean vectors $(\mu_i^*)_{i=1}^K \in \R^d$   and positive definite matrices $(\Sigma_i^*)_{i=1}^K \in \mathbf{S}_+^d $. 
\end{assumption}

\begin{assumption}[Covariance Prior] \label{cov-prior} $(i)$ $G$ is a probability measure with support contained in the space of positive 
semi-definite matrices on $\R^{d \times d}$. $(ii)$ There exists $ \kappa \in (0,1   ], v_1  \geq  0 , v_2 \geq 0, v_3 \geq 0$ and universal constants $C,C',D,D',L' > 0 $ that satisfy \begin{align*} & (i) \; \; \; \; \;  G \big(  \Sigma :  \lambda_{d} ( \Sigma^{-1})  > t_d    \big) \leq  C \exp(- C' t_d^{\kappa}) \; , \\ & (ii) \; \; \; \; \; G  \bigg(   \Sigma : \bigcap_{1 \leq j \leq d} \bigg \{   t_j \leq  \lambda_j(\Sigma^{-1}) \leq  t_j(1+\delta)    \bigg\}     \bigg) \geq D t_1^{v_1} \delta^{v_2}  \exp(-D' t_d^{\kappa}) \; , \\ & (iii) \; \; \; \; \; G(\Sigma : \lambda_1 ( \Sigma^{-1} < t_1) \leq C \exp(-L' t_1^{-v_3} )
\end{align*} 
for every $  \delta \in (0,1) $ and  $0 < t_1 \leq t_2 \leq \dots \leq  t_d < \infty  $.
\end{assumption}
Since the researcher chooses the prior, Assumptions \ref{dir-prior} and \ref{cov-prior} can always be satisfied. Restricting $\alpha$ to be a finite Gaussian mixture in Assumption \ref{dir-prior} facilitates the computation as it is necessary to sample from the base measure. It is not to be interpreted as restrictive, as any distribution can be approximated arbitrarily well by a Gaussian mixture with sufficiently many components. Assumption \ref{cov-prior} $(i-ii)$ allows for a large class of covariance priors and is commonly used in density estimation (\citealp{shen2013adaptive}; \citealp{ghosal2017fundamentals}). In dimension $d=1$, it holds with $\kappa=1/2$ if $G$ is the distribution of the square of an inverse-Gamma distribution. In dimension $d > 1$, it holds with $\kappa = 1$ if $G$ is the inverse-Wishart distribution.\footnote{In higher dimensions, a computationally attractive choice is the class of Lewandowski-Kurowicka-Joe (LKJ) priors \citep{lewandowski2009generating}.} Assumption \ref{cov-prior} $(iii)$, for $v_3 > 0$, is equivalent to Assumption 4.2 in \citet{rousseau2024wasserstein}. In $d=1$, it holds with a truncated inverse-Gamma prior.

In some cases, researchers may have prior estimates of the latent distribution from a preliminary parametric model, or certain features (e.g. moments) may be estimable using the observables. Similar to the analysis in \citet{mcauliffe2006nonparametric}, this can be accomodated through a data-dependent (empirical Bayes) $\text{DP}_{\alpha}$ prior. For example, in measurement error models (Example \ref{ex1}), a reasonable choice  is $\widehat{\alpha} = \mathcal{N}(\E_n(Y), \widehat{\text{Var}}(X))$. The following remark provides some clarification on the effect of data-dependent priors.
\begin{remark}[Data dependent Priors] \label{data-dep-prior}
   In Assumption \ref{dir-prior}, the base measure $\alpha$ can be data dependent. All our main results remain valid with a data dependent prior, provided that, asymptotically in probability, $\|\widehat{\mu}_j\|$ is bounded, and the eigenvalues of $\widehat{\Sigma}_j$ are bounded away from zero and infinity. This is analogous to the treatment of data dependent priors in pure Bayesian models (\citealp{rousseau2017asymptotic}).
\end{remark}

In the remainder of this section, we focus on specific results for the latent variable models introduced in Section \ref{sec2}. For each model, we examine the following questions in detail: (i) What are the minimal conditions on $F_X$
  that ensure quasi-Bayes consistency? (ii) How do convergence rates depend on the smoothness of 
$F_X$? (iii) Are rates faster for certain classes of distributions (e.g. infinite Gaussian mixtures), and what is the interplay with the model ill-posedness?

\subsection{Measurement Error} \label{4.1}
Consider the classical measurement error model, introduced in Example \ref{ex1}, where we observe a random sample from
 \begin{align}
    \label{mdeconv}   W = X + \epsilon \; \; \; \; , \; \; \E[\epsilon] = 0.
\end{align}
In this setting, $W \in \R^d $ is the observed vector, $X \in \R^d$ is an unobserved vector, and $\epsilon$ is a nuisance error that is 
 independent of $X$. We are interested in recovering the latent distribution of $X$.  As the individual contributions of $X$ and $\epsilon$ cannot be separately identified from an observation of $W$, identification of the latent distribution of $X$ generally requires some further auxiliary information on the distribution $F_{\epsilon}$ of the unobserved error $\epsilon$. As in \citet{kato2018uniform}; \citet{arellano2023recovering}, we consider the setting where a researcher has auxiliary information in the form of a random sample  ${\epsilon}_1^{\star}, \dots , {\epsilon}_n^{\star} \stackrel{i.i.d}{\sim} F_{\epsilon} $.\footnote{If the auxiliary sample has size $m_n = o(n)$, the results remain valid, albeit with a possibly slower rate of convergence. There are no restrictions on the dependence of the auxiliary sample with the original sample.} As a basic premise, we impose the following low level condition on the observable $Y = (W , \epsilon^{\star})$.\begin{condition2}[Data] \label{data}  $(i)$ $(W_i)_{i=1}^n $ is a sequence of independent and identically distributed random vectors. $(ii)$ $(\epsilon_i^{\star})_{i=1}^n \stackrel{i.i.d}{\sim} F_{\epsilon} $ where $F_{\epsilon}$  is the distribution of the unknown error $\epsilon$ in (\ref{mdeconv}). $(iii)$ $(Y,\epsilon)$ have finite second moments.
\end{condition2}
We now proceed with stating the main results for the model in detail. From the discussion in Example \ref{ex1}, the population characteristic restriction is given by \begin{equation}
    \label{deconv-rest} \mathcal{G}( \mathbb{P}_Y ,\varphi_X)(t) = \E[e^{\mathbf{i}t' W}]  - \varphi_{X}(t) \E[e^{\mathbf{i}t' \epsilon}].
\end{equation}
With $Y = (W , \epsilon^{\star})$ as the observable, the feasible analog $\mathcal{G}(\mathbb{P}_{n,Y} ,\varphi_X)$, obtained by using the empirical measure $\mathbb{P}_{n,Y}$, is given by \begin{equation}
    \label{deconv-rest-feas} \mathcal{G}( \mathbb{P}_{n,Y} ,\varphi_X)(t) = \E_n[e^{\mathbf{i}t' W}]  - \varphi_{X}(t) \E_n[e^{\mathbf{i}t' \epsilon}].
\end{equation}
Our first result examines the limiting behavior of quasi-Bayes under minimal distributional assumptions on the data. Specifically, we avoid imposing any strong distributional structure on $(X,\epsilon)$, instead relying solely on primitive moment and tail bounds. As these basic assumptions will not preclude distributions with irregular components (e.g. mass points), it is necessary to quantify the behavior through a metric that is well-defined for a broad class of probability measures. To that end, we make use of the Wasserstein metric: \begin{equation}
    \label{wass-d} \mathbf{W}_1(F,F') = \inf \bigg\{ \E \big(  \| Z-Z'  \|   \big) : Z \sim F \:,\:Z' \sim F'  \bigg \}.
\end{equation}
Among other desirable properties, it is well known (see, e.g., \citealp{villani2009optimal}) that convergence in $\mathbf{W}_1(.)$ implies convergence in distribution: $\mathbf{W}_1(F_n, F_X) \to 0 \implies F_n \rightsquigarrow F_X.$ The following result demonstrates that the quasi-Bayes posterior contracts to the latent distribution in Wasserstein distance.

\begin{theorem}[Wasserstein Contraction]
\label{deconv-1} Suppose $d=1$. Assume that $(i)$ $F_X( t \in \R : \left| t \right|  > z) \leq C \exp(- C' z^{\chi})$ for some $\chi, C,C' > 0$ , $(ii)$ $\inf_{\| t \|_{\infty} \leq T} \left| \varphi_{\epsilon}(t)    \right| \geq b \exp(-B T^{\zeta})   $ for some  $b,B,\zeta > 0$, $(iii)$ Assumption \ref{cov-prior} holds with $\kappa \leq 1/2$ and some $v_3 \geq 1$. Then, for any sequence $T_n = o(n^{1/5})$, there exists a universal constant $D > 0$ such that \begin{align*}
    \nu \bigg(F : \mathbf{W}_1(F,F_X) >  D \max \bigg \{ \frac{1}{T_n} , (\log n)^{-1/\zeta}     \bigg \}   \: \bigg| \: \mathcal{D}_n   
 \bigg) \xrightarrow{\mathbb{P}} 0.
\end{align*}
In particular, for any $T_n$ satisfying $(\log n)^{-1/\zeta} \lessapprox T_n = o(n^{1/5})$, we have \begin{align*}
    \nu \bigg(F : \mathbf{W}_1(F,F_X) >  D  (\log n)^{-1/\zeta}       \: \bigg| \: \mathcal{D}_n   
 \bigg) \xrightarrow{\mathbb{P}} 0.
\end{align*}
\end{theorem}
Theorem \ref{deconv-1} is very general in its scope, as it does not require any strong distributional assumptions on $(X,\epsilon)$ and relies only on an auxiliary sample of $\epsilon$. It remains valid even when $(X,\epsilon)$ contain discrete components. Intuitively, regardless of the underlying structure, the moment  in~(\ref{deconv-rest}) are always valid identifying restrictions.  This result appears to be the first in the literature to establish a  deconvolution estimation guarantee under such minimal assumptions.

The exponential lower bound on  $
\varsigma_T = \inf_{\|t\| \leq T} \left| \varphi_{\epsilon}(t) \right| $  
allows for measurement errors whose characteristic functions decay rapidly. For instance, if $\epsilon$ follows a Gaussian distribution, the lower bound is attained exactly at $\zeta = 2$. Remarkably, the contraction result in Theorem~\ref{deconv-1} holds for any sequence $T_n = o(n^{1/5})$, regardless of the distribution of $\epsilon$. This result extends to dimensions $d > 1$, although in such cases additional constraints on the growth rate of $T_n$ are required.  

Finally, we note that the slow convergence rate in Theorem~\ref{deconv-1} is typical in deconvolution problems (see, e.g., \citealp{carroll1988optimal}). Without further regularity conditions on $(F_X, F_{\epsilon})$, the rate in Theorem~\ref{deconv-1} cannot be improved; it is minimax optimal \citep*{dedecker2013minimax}.

 In the remainder of this section, we focus on establishing contraction rates in settings where the latent distribution or measurement error are known to be sufficiently regular, either through Sobolev type constraints or restrictions on the nonparametric functional form. An important example of such a setting is when the latent distribution $F_X$ is a nonparametric Gaussian mixture:
\begin{align}
    \label{d=gaussm} f_X(x) = \phi_{\Sigma_0} \star F_0(x) = \int_{\R^d} \phi_{\Sigma_0} (x-z) d F_0(z).
\end{align}
Here, $\Sigma_0 \in \mathbf{S}_+^d$  is a positive definite matrix and $F_0$ is a mixing distribution. If $F_0$ is a discrete distribution with finite support, this reduces to a Gaussian mixture with finitely many components. Our focus on this class is partially motivated by the observation that finite Gaussian mixtures are widely used to model complex heterogeneity in latent variable models (e.g. \citealp{cunha2010estimating}; \citealp{attanasio2020human}). Importantly, if the mixing distribution $F_0$ is continuously distributed or has unbounded support, the specification in (\ref{d=gaussm}) allows for Gaussian mixtures with infinite components. 


As the specification in (\ref{d=gaussm}) lies between a parametric and fully nonparametric model, a natural question is whether fast rates of convergence can be expected in this setup and, in particular, how is this influenced by the model ill-posedness. Intuitively, as the prior concentrates on infinite Gaussian mixtures, the expectation is that the prior should exhibit negligible bias in approximating a latent density with Gaussian mixture structure as in (\ref{d=gaussm}). In the next result, we formalize this intuition. To this end, our main assumption regarding the Gaussian mixture specification is as follows.
\begin{condition2}[Nonparametric Mixture] \label{exact-g}   $f_X = \phi_{\Sigma_0} \star F_0$ where $\Sigma_0 \in \R^{d \times d}$ is a positive definite matrix and   $F_0$ is a mixing distribution that satisfies $ F_0 \big(  t \in \R^d : \| t \| > z  \big) \leq C \exp(- C' z ^{\chi}) $ for universal constants $C,C',\chi > 0$.
\end{condition2}
Condition (\ref{exact-g}) imposes an exponential tail on the mixing distribution $F_0$. Importantly, this covers all Gaussian mixtures that have (possibly infinite) modes contained in a compact subset of $\R^d$. The following result demonstrates that, with an appropriately scaled covariance prior, the quasi-Bayes posterior contracts towards the true latent density at a polynomial (and sometimes nearly parametric) rate.

\begin{theorem}[Rates with Gaussian Mixtures]
\label{deconv-2} Suppose Condition \ref{exact-g} holds and the covariance prior is taken as $G_n \sim G / \sigma_n^2$ where $G$ satisfies Assumption \ref{cov-prior} and $\sigma_n^2$  is as specified below. Let $\kappa, \chi  > 0$ be as in Assumption \ref{cov-prior} and  Condition \ref{exact-g}, respectively. 

\begin{enumerate}
    \item[$(a)$]  Let  $ \lambda = \max \{ \chi^{-1}(d+2) + d/2 , d+1    \} $. Suppose the model is mildly ill-posed: $ \inf_{\| t \|_{\infty} \leq T} \left| \varphi_{\epsilon}(t)   \right| \geq c T^{-\zeta}  $ for some $c > 0$ and $\zeta \geq 0$. Suppose $T_n \asymp  \sqrt{\log n} \sqrt{\log \log n}  $ and $\sigma_n^2 \asymp (\log n)^{- \lambda/\kappa} $. Then there exists a universal constant $D > 0$ such that \begin{align*}
        \nu \bigg( \phi_{P,\Sigma} : \| f_X  -  \phi_{P,\Sigma}      \|_{L^2} >  D \frac{(\log n)^{(\lambda + \zeta)/2} (\log \log n)^{\zeta/2} }{\sqrt{n}}   \: \bigg| \: \mathcal{D}_n  \bigg) \xrightarrow{\mathbb{P}} 0
    \end{align*}
    \item[$(b)$] Let  $ \lambda = \max \{ \chi^{-1}(d+2) + d/2 , d+1  , d/ \zeta + 1/2   \} $. Suppose the model is severely ill-posed: $ \inf_{\| t \|_{\infty} \leq T} \left| \varphi_{\epsilon}(t)   \right| \geq c \exp(-B T^{\zeta} ) $ for some $c,B > 0$  and $\zeta \in (0,2]$. Suppose $T_n \asymp (\log n)^{1/\zeta} $ and  $\sigma_n^2  \asymp (\log n)^{- \lambda/\kappa} $. Then, there exists a universal constant $D>0$, sequence $L_n \uparrow \infty$ and $V \in (0,1/2)$ such that
     \begin{align*}
     \begin{cases}   \nu \bigg( \phi_{P,\Sigma} : \| f_{X}  -  \phi_{P,\Sigma}      \|_{L^2} >  L_nn^{-1/2}    \: \bigg| \: \mathcal{D}_n   \bigg) \xrightarrow{\mathbb{P}} 0  & \zeta \in (0,2) \\ \nu \bigg( \phi_{P,\Sigma} :  \| f_{X}  -  \phi_{P,\Sigma}      \|_{L^2} > D  n^{-V}  \: \bigg| \: \mathcal{D}_n  \bigg) \xrightarrow{\mathbb{P}} 0 & \zeta=2.    \end{cases}
    \end{align*}
Furthermore, the sequence $(L_n)_{n=1}^{\infty}$ grows slower than any polynomial in the sense that  $L_n = o(n^{\epsilon})$ for any $\epsilon > 0$.
\end{enumerate}
\end{theorem}
Under a mildly ill-posed regime, the quasi-Bayes posterior attains nearly (up to a logarithmic factors) parametric rates of convergence. If the model is severely ill-posed at a rate $\zeta$ that is lower than Gaussian type ill-posedness, nearly parametric rates in the form  $n^{-1/2 + \epsilon}$ for any $\epsilon > 0$ can be obtained. In the extreme case where the model exhibits Gaussian ill-posedness, the model bias (which is determined by the decay rate of a Gaussian mixture characteristic function) has similar order to the variance. In this case, the rate is still polynomial but the exponent $V$ depends on second order factors such as the constant $B$ and the eigenvalues of the matrix $\Sigma_0$ in representation (\ref{d=gaussm}). In a univariate framework with a kernel deconvolution based estimator, a similar finding is established in \citet{butucea2008sharp2}.

\begin{remark}[On Covariance Scaling] \label{var-scale}
Theorem \ref{deconv-2} differs from Theorem \ref{deconv-1} in that the covariance prior must be appropriately scaled. This scaling strategy, originally employed by \citet{ghosal2007posterior} to establish contraction rates for pure Bayesian density estimation, was later relaxed by \citet{shen2013adaptive}. We conjecture that the scaling requirement in our result may be an artifact of the proof and could potentially be removed. We leave this investigation for future work. In our simulations and empirical illustrations, we do not apply any scaling to the covariance prior.
 \end{remark}

As Gaussian mixtures can approximate a density arbitrarily well, it is straightforward to modify the preceding results to obtain consistency for a general density.\footnote{Here, consistency in a norm $\| . \|$, is defined as $ \nu_{} \big( \phi_{P,\Sigma} :  \| f_X - \phi_{P,\Sigma}  \|_{} > \epsilon  \:\big| \: \mathcal{D}_n \big) = o_{\mathbb{P}}(1) \; \; \forall \; \epsilon > 0$} In the interest of comparison with existing results in the literature, we focus on the more difficult task of establishing a rate of convergence. To that end, consider the case where the true latent density $f_X$ is not an exact Gaussian mixture and thus, the prior model exhibits non-negligible bias. In this model, the Gaussian mixture bias is closely related to the quantity \begin{align} \label{gm-bias} \mathcal{B}(\sigma) =   \inf_{(P,\Sigma) : \lambda_{\min}(\Sigma) \geq \sigma^2} \|f_X \star f_{\epsilon}  - \phi_{P, \Sigma_{}}   \star F_{\epsilon}     \|_{L^2} \asymp \inf_{(P,\Sigma) : \lambda_{\min}(\Sigma) \geq \sigma^2} \| \varphi_{\epsilon}\big(  \varphi_{X} - \varphi_{P,\Sigma}  \big)   \|_{L^2} .  \end{align}
\begin{remark}[Bias Bounds] For a conservative upper bound on the rate at which $\mathcal{B}(\sigma) $ tends to $0$ as $\sigma \rightarrow 0$, observe that by picking $P = F_X$ and $\Sigma = \sigma^2 I$, we obtain 
\begin{align*}
   \mathcal{B}(\sigma) \leq  \| f_X \star f_{\epsilon} - \phi_{F_X,\sigma^2 I}   \|_{L^2} \asymp  \| (e^{- \| t \|^2 \sigma^2 /2} - 1 ) \varphi_{X} (t) \varphi_{\epsilon}(t)  \|_{L^2} \lessapprox \sigma^2 .
\end{align*}
The preceding bound follows from $ \left| e^{-\| t \|^2 \sigma^2 / 2} - 1 \right| \leq \|t \|^2 \sigma^2 / 2 $. As such, the concluding estimate is valid provided that $ \int_{\R^d} \| t \|^4 \left|  \varphi_{X} (t) \right|^2 \left| \varphi_{\epsilon}(t)  \right|^2 < \infty  $ or equivalently $f_{X} \star f_{\epsilon} \in \mathbf{H}^2 $. For $f_{X} \star f_{\epsilon}$ with smoothness order greater than $2$, this argument does not provide better rates. In particular, with higher smoothness, the choice $ P = f_X$  is not optimal in (\ref{gm-bias}): one can typically do  better by allowing $ P = P_{X,\epsilon,\sigma}$ to depend implicitly on all the features $(f_{\epsilon},f_{X},\sigma)$. 
\end{remark}
Intuitively, the smoothness of the convolution $ f_X \star f_{\epsilon} $ is the sum of the smoothness of the individual functions. As such, we expect the Gaussian mixture bias to decay at a rate proportional to the combined Sobolev smoothness (or equivalently the characteristic function decay) of $f_X$ and $f_{\epsilon}$. To derive the limit theory, we consider the following popular characterization of the Gaussian mixture bias.
\begin{condition2}[Gaussian Mixture Bias] \label{deconv-bias}
$f_X \in \mathbf{H}^p$ for some $p > 1/2$. For all $\sigma > 0 $ sufficiently small and some $\chi > 0$, there exists $(i)$ a mixing distribution $F_{\sigma} = F_{X,\epsilon,\sigma}$ supported on  the cube $ I_{\sigma} = [-C  (\log \sigma^{-1})^{1/\chi},C  (\log \sigma^{-1})^{1/\chi} ]^d$ and $(ii)$ a covariance matrix $\Sigma_{\sigma} = \Sigma_{X,\epsilon,\sigma} $ with $\lambda_{\min}(\Sigma_{\sigma}) \geq  \sigma^2$ such that  the Gaussian mixture $\phi_{F_{\sigma},\Sigma_{\sigma}}$ satisfies $\| f_X \star F_{\epsilon} - \phi_{F_{\sigma},\Sigma_{\sigma}} \star F_{\epsilon}  \|_{L^2} \leq  D \sigma^{p + \zeta}  $ for some $\zeta > 0$, where $C,D < \infty$ are universal constants.
\end{condition2}
Variations of Condition \ref{deconv-bias}, usually with additional smoothness constraints on $\log f_X$, are commonly used in Bayesian density estimation and deconvolution. For an explicit construction of $(F_{\sigma},\Sigma_{\sigma})$ in Condition \ref{deconv-bias}, see \citet{shen2013adaptive}; \cite{ghosal2017fundamentals} for density estimation and \cite{rousseau2024wasserstein} for deconvolution. We note that, at least in our quasi-Bayes setup, Condition \ref{deconv-bias} can be weakened further. In particular, it suffices that the approximation holds with respect to the weaker Fourier norm $\| . \|_{\mathbb{B}(T_n)}$. Finally, we note that our procedures do not require knowledge of the optimal mixing distribution $F_{\sigma}$ in Condition \ref{deconv-bias} for implementation. The existence is purely a proof device towards obtaining theoretical guarantees.

The following result derives the quasi-Bayes  contraction rate when the Gaussian mixture bias is as in Condition \ref{deconv-bias} and $ \varphi_{\epsilon}(.) $ decays at a mildly ill-posed rate of order $\zeta \geq 0$.

\begin{theorem}[Rates with Gaussian Mixture Bias]
\label{deconv-3} Suppose Condition \ref{deconv-bias} holds with $\chi > 0$. Let $\kappa > 0$ be as in Assumption \ref{cov-prior} and define $\lambda = \max \{ \chi^{-1}(d+2), \chi^{-1} d + 1   \}$. Suppose the covariance prior is  $G_n \sim G / \sigma_n^2$ where $G$ satisfies Assumption \ref{cov-prior} and $\sigma_n^2 \asymp (\log \log n)^{-1}T_n^{2-d/\kappa}(\log n)^{-\lambda/\kappa - 1}   $. If $ \inf_{\| t \|_{\infty} \leq T} \left| \varphi_{\epsilon}(t)   \right| \gtrapprox T^{-\zeta}  $ for some  $\zeta \geq 0$ and $T_n \asymp n^{1/[2(p + \zeta) + d]} \sqrt{\log n} $, there exists a universal constant $D > 0$ such that \begin{align*}
        \nu_{} \bigg( \phi_{P,\Sigma} :  \| f_X -  \phi_{P,\Sigma}      \|_{L^2} >  D n^{ \frac{-p}{2(p + \zeta) + d}} (\log n)^{(\lambda + \zeta)/2 + d/4}  \:     \bigg| \: \mathcal{D}_n  \bigg) \xrightarrow{\mathbb{P}} 0.
    \end{align*}

\end{theorem}
The obtained rates are minimax-optimal, up to a logarithmic factor. When $\zeta=0$, we recover the usual optimal rates for multivariate density estimation. In dimension $d=1$, the rates agree with the univariate Bayesian deconvolution  rates (with known error distribution) obtained in \citet{donnet2018posterior}.

\begin{remark}[Convergence in stronger metrics] \label{rem1} An interesting feature of the quasi-Bayes framework is that all formal analysis takes place directly in the Fourier domain. Intuitively, we first establish contraction rates in the Fourier domain and then use it to deduce $\| \cdot \|_{L^2}$ and $\mathbf{W}_1$ rates for the distribution. This approach can be adapted to other metrics. Indeed, Fourier distances are widely used to derive bounds in classical metrics (see \citealp{bobkov2016proximity} for a comprehensive survey). For example, the elementary bound $\| f_X \|_{L^\infty} \leq \| \varphi_{X} \|_{L^1}$ yields the following variant of Theorem \ref{deconv-3}.

\end{remark}

\begin{corollary}[$L^{\infty}$ Rates with Gaussian Mixture Bias]
\label{deconv-3-2} Suppose the hypothesis of Theorem \ref{deconv-3} holds. Then, there exists a universal constant $D > 0$ such that \begin{align*}
        \nu_{} \bigg( \phi_{P,\Sigma} : \| f_X -  \phi_{P,\Sigma}      \|_{L^\infty} >  D n^{ \frac{-p+d/2}{2(p + \zeta) + d}} (\log n)^{(\lambda + \zeta + d)/2  }  \:     \bigg| \: \mathcal{D}_n  \bigg) \xrightarrow{\mathbb{P}} 0.
    \end{align*}

\end{corollary}
Interestingly, when $d=1$ and $\zeta=0$, the $L^{\infty}$ quasi-Bayes contraction rates in Corollary \ref{deconv-3-2} coincide with the univariate density pure-Bayes rates obtained in \citet{gine2011rates}. For $d>1$ or $\zeta > 0$, the contraction rates appear to be novel; in particular, we could not find any pure-Bayes results for these cases in the literature. For simplicity of exposition, all remaining results in Section \ref{sec4} are stated exclusively for the Wasserstein and $\| \cdot \|_{L^2}$ metrics. By analogous arguments to those used in Corollary \ref{deconv-3-2}, these results can be extended to other metrics.

\begin{remark}[On Functionals] \label{functionals}
In the preceding results, the estimation rates range from logarithmic to polynomial, depending on the regularity of $(X, \epsilon)$. These rates are stated with respect to the direct recovery of the latent distribution under strong metrics such as $\mathbf{W}_1$ and $L^2$. However, if a researcher is interested in a class of regular functionals of $X$, the convergence is typically much faster. For example, the empirical Bayes NPMLE literature (e.g. \citealp{gu2023invidious}) often reports $\sqrt{n}$ rates (up to log factors) in their convergence guarantees. This is because their quantity of interest only depends on convergence of the induced density of observables. In our setting, this exactly corresponds to contraction in the weak metric (see Theorem~\ref{main-contract}). In particular, all of our results also have corresponding versions for fast (up to $\sqrt{n}$) rates in the weak metric. Intuitively, even if $X$ is irregular, the distribution of observables $f_{W} =f_X \star f_{\epsilon}$ may be very smooth, allowing for fast recovery of functionals that depend on it. 
\end{remark}

\begin{remark}[Posterior Mean] \label{posteriormean}
As a formal estimator of the latent density $f_X$ (or distribution $F_X$), it is natural to consider the posterior mean: $$\E[\phi \:\big| \; \mathcal{D}_n] = \int \phi \; d \nu(\phi \: \big| \: \mathcal{D}_n). $$
In all the preceding results, without imposing any additional conditions, the posterior mean converges at the same rate as the contraction rate of the quasi-Bayes posterior.\footnote{This also holds for all the remaining results in this section.}
\end{remark}

\subsection{Repeated Measurements} \label{4.2}
Consider the repeated measurements model, introduced in Example \ref{ex2}, where we observe a random sample from the model:
 \begin{align*}
     & Y_1 = X + \epsilon_1 \;  ,  \; \; \;  \E[\epsilon_1] = 0 ,  \\  & Y_2 = X + \epsilon_2 \;  ,  \; \; \;  \E[\epsilon_2] = 0 . \nonumber
\end{align*}
In this setting, $Y_1,Y_2 \in \R$ are observed measurements, $X \in \R$ is unobserved, and $\epsilon_1,\epsilon_2$ are unobserved nuisance errors. We assume $\E[\epsilon_1|X,\epsilon_2] = 0$ and that $\epsilon_2 $ is independent of $X$. The primitive assumption that we impose on the observations is summarized in the following condition.
\begin{condition2}[Data] \label{data2}  $(i)$ $(Y_{1,i},Y_{2,i})_{i=1}^n $ is a sequence of independent and identically distributed random variables. $(ii)$ $Y_1$ and $Y_2$ have finite second moments.
\end{condition2}
From the discussion in Example \ref{ex2}, the population characteristic restriction is \begin{equation}
    \label{rep-multfac} \mathcal{G}( \mathbb{P}_Y ,\varphi_X)(t) = \E[e^{\mathbf{i}t' Y_2}] \partial_t \log \varphi_X(t)  -  \E[\mathbf{i} Y_1e^{\mathbf{i}t Y_2}].
\end{equation}
With $Y = (Y_1 , Y_2)$ as the observable, the feasible analog $\mathcal{G}(\mathbb{P}_{n,Y} ,\varphi_X)$, obtained by using the empirical measure $\mathbb{P}_{n,Y}$, is given by \begin{equation}
    \label{rep-rest-feas} \mathcal{G}( \mathbb{P}_{n,Y} ,\varphi_X)(t) = \E_n[e^{\mathbf{i}t' Y_2}] \partial_t \log \varphi_X(t)  -  \E_n[\mathbf{i} Y_1e^{\mathbf{i}t Y_2}].
\end{equation}
Similar to the analysis in Section \ref{4.1}, our first result examines the limiting behavior of quasi-Bayes under minimal distributional assumptions on the data.
\begin{theorem}[Wasserstein Contraction]
\label{repmeas-1} Assume $(i)$ $F_X( t \in \R : \left| t \right|  > z) \leq C \exp(- C' z^{\chi})$ for some $\chi, C,C' > 0$, $(ii)$ $\inf_{\| t \|_{\infty} \leq T}  \left| \varphi_{Y_2}(t)  \right|    \geq b \exp(-B T^{2})   $ and $ \sup_{\| t \|_{\infty} \leq T} \left|\partial_t \log \varphi_X (t)  \right| \leq D (1+T)^K$ for some $b,B,D,K > 0 $,
$(iii)$ Assumption \ref{cov-prior} holds with some $ \kappa \in (0,1] $ and $v_3 \geq 1$. Then, for any sequence $T_n \lessapprox  \sqrt{\log n}$, there exists a universal constant $D > 0$ such that \begin{align*}
    \nu \bigg(F : \mathbf{W}_1(F,F_X) >  D \max \bigg \{ \frac{1}{T_n} , (\log n)^{-1/2}     \bigg \}   \: \bigg| \: \mathcal{D}_n   
 \bigg) \xrightarrow{\mathbb{P}} 0.
\end{align*}
 In particular, for any $T_n \asymp \sqrt{\log n}$, we have \begin{align*}
    \nu \bigg(F : \mathbf{W}_1(F,F_X) >   D (\log n)^{-1/2}       \: \bigg| \: \mathcal{D}_n   
 \bigg) \xrightarrow{\mathbb{P}} 0.
\end{align*}
\end{theorem}
 The exponential lower bound on $ \inf_{\| t \|_{\infty} \leq T}\left|\varphi_{Y_2} (t)\right|$ is a very weak restriction. Analogous to the identification restrictions with measurement error, if the term $\varphi_{Y_2}(t) \partial_t \log \varphi_X(t)$ in the population objective function vanished, $\varphi_X$ would not be identifiable. The upper bound on $ \sup_{\|t \|_{\infty} \leq T}\left| \partial_t \log \varphi_X(t) \right|$ is very conservative—the quantity is typically uniformly bounded when $\varphi_X$ decays at a polynomial rate, and it grows at rate $T$ in the extreme case of a Gaussian decay. To the best of our knowledge, Theorem~\ref{repmeas-1} is the first Wasserstein convergence guarantee for the repeated measurements model.\footnote{We could not find any prior results, whether frequentist or Bayesian, in the literature.} Additionally, it is the first nonparametric Bayes contraction rate for this model. 
 
 Next, similar to the analysis in Section \ref{4.1}, we focus on establishing further contraction rates in settings where the latent distribution is known to be sufficiently regular. Specifically, we consider the case when the latent distribution is a nonparametric Gaussian mixture:\begin{align} \label{d=gaussm-rep} f_X(x) = \phi_{\sigma_0^2} \star F_0(x) = \int_{\R} \phi_{\sigma_0^2} (x-z) d F_0(z) .
\end{align} 
Here, $\sigma_0 > 0$ and $F_0$ is a mixing distribution. If the mixing distribution $F_0$ is continuously distributed or has unbounded support, the specification in (\ref{d=gaussm-rep}) allows for Gaussian mixtures with infinite components. Our main assumption on the Gaussian mixture specification is as follows.
\begin{condition2}[Nonparametric Mixture] \label{exact-g-2} $(i) $ $f_X = \phi_{\sigma_0^2} \star F_0$ for some $ \sigma_0 > 0$ and mixing distribution $F_0$ that satisfies $ F_0 \big(  t \in \R : \left| t \right| > z  \big) \leq  C \exp(- C' z ^{\chi}) $ for some $\chi > 0$ and universal constants $C,C' > 0$. $(ii)$ $  \sup_{\| t \| \leq T} \left|  \partial_t \log \varphi_{F_0} (t) \right| \leq C (1+T) $  for a universal constant $C> 0$.
\end{condition2}
Condition $(i)$ imposes an exponential tail on the mixing distribution $F_0$. Importantly, this covers all Gaussian mixtures that have (possibly infinite) modes contained in a compact subset of $\R$. Condition $(ii)$ is similar to the hypothesis of Theorem \ref{repmeas-1}. The following result demonstrates that, with an appropriately scaled covariance prior, the quasi-Bayes posterior contracts towards the true latent density at a polynomial rate.

\begin{theorem}[Rates with Gaussian Mixtures]
\label{deconv-prod-t2} Suppose $ \inf_{\|t \| \leq T} \left| \varphi_{Y_2} (t)  \right| \geq c \exp(- c' T^2) $ for some constants $c,c' > 0$ and Condition \ref{exact-g-2} holds. Let $\kappa, \chi  > 0$ be as in Assumption \ref{cov-prior} and  Condition \ref{exact-g-2}, respectively. Define $   \lambda = \max \{ (3/\chi)+1/2 , 2    \}$. Suppose the covariance prior is taken as $G_n \sim G / \sigma_n^2$ where $G$ satisfies Assumption \ref{cov-prior} and $\sigma_n^2  \asymp (\log n)^{- \lambda/\kappa}$. If $T_n \asymp (\log n)^{1/2} $, there exists a universal constant $D>0$ and $V \in (0,1/2)$ such that $$ \nu \bigg( \phi_{P,\Sigma} :  \| f_{X}  -  \phi_{P,\Sigma}      \|_{L^2} > D  n^{-V}  \: \bigg| \: \mathcal{D}_n  \bigg) \xrightarrow{\mathbb{P}} 0.  $$
\end{theorem}
The main difference in this setting, compared to Theorem \ref{deconv-2}, is that a nonparametric Gaussian mixture specification on $F_X$ influences both the regularity and ill-posedness in the model. To see this, note that, up to sampling uncertainty, the quasi-Bayes objective function is $$ Q(F) = \| \varphi_{Y_2} ( \partial_t \log \varphi_X - \partial_t \log  \varphi_F)   \|_{\mathbb{B}(T)}. $$
Thus, the model ill-posedness is driven by the decay of $\varphi_{Y_2}$. Since $\varphi_{Y_2} = \varphi_{X} \times \varphi_{\epsilon_2}$, a rapid decay of $\varphi_X$ also influences $\varphi_{Y_2}$. The contraction rate in Theorem \ref{deconv-prod-t2} is obtained through examining the interplay and dependence between $\varphi_{Y_2}$  and $\varphi_{X}$. As expected, similar to the discussion following Theorem \ref{deconv-2}, the exponent $V$ depends on second order factors, e.g. the constants $(\sigma_0^2, c')$.

\begin{remark}[On Symmetric Restrictions] \label{rem-sym} If $(X, \epsilon_1, \epsilon_2)$ are mutually independent, the roles of $Y_1$ and $Y_2$ become interchangeable in the identifying restrictions. This provides us with two sets of moment identifying restrictions. The simplest approach to accommodate this setting is to construct the quasi-Bayes posterior using both sets of moment restrictions:
\begin{align*}
   &  \mathcal{G}_1( \mathbb{P}_Y ,\varphi_X)(t) = \E[e^{\mathbf{i}t' Y_2}] \partial_t \log \varphi_X(t)  -  \E[\mathbf{i} Y_1e^{\mathbf{i}t Y_2}] \\ & \mathcal{G}_2( \mathbb{P}_Y ,\varphi_X)(t) = \E[e^{\mathbf{i}t' Y_1}] \partial_t \log \varphi_X(t)  -  \E[\mathbf{i} Y_2e^{\mathbf{i}t Y_1}].
\end{align*}
We recommend this choice unless there are compelling empirical arguments to treat the nuisance errors asymmetrically.
\end{remark}

\subsection{Multi-Factor Models} \label{4.3}
Following \citet{bonhomme2010generalized}, consider the multi-factor model, where we observe a random sample from the model \begin{align}
 \label{mult-fac-main}    \mathbf{Y} = \mathbf{A} \mathbf{X}.
\end{align}
In this setting, $\mathbf{Y} = (Y_1,\dots,Y_L)' \in \R^L$ is a vector of $L$ measurements, $\mathbf{X} = (X_1,\dots,X_K)'$ is a vector of $K$ latent and mutually independent factors and $\mathbf{A}$ is a known $L \times K$ matrix of factor loadings. The primitive assumption that we impose on the observations is summarized in the following condition.
\begin{condition2}[Data] \label{data3}  $(i)$ $( \mathbf{Y}_i )_{i=1}^n $ is a sequence of independent and identically distributed random vectors. $(ii)$ Finite second moment: $\E\big( \| \mathbf{Y} \|^2 \big) < \infty $. $(iii)$ The factors $X_1,\dots,X_K$ are  mutually independent and demeaned, i.e $\E[X_k] = 0$ for $k=1,\dots,K$.
\end{condition2}
Suppose we are interested in the latent distribution $F_{X_k}$ for some $k \in \{ 1 , \dots , K \}$. Let $\mathcal{V}_{ \mathbf{Y}   } (t)$  denote the vector of upper triangular elements of $\nabla \nabla' \log \varphi_{\mathbf{Y}} (t) $. Let $V(\mathbf{A}_k)$  denote the vector of upper triangular elements of
$\mathbf{A}_k \mathbf{A}_k'$ and define $\mathbf{Q} = [ V(\mathbf{A}_1) , \dots , V(\mathbf{A}_K) ] $ to be the matrix with columns given by those vectors. From the discussion in Example \ref{ex4}, if $\mathbf{Q}^* = (\mathbf{Q}'\mathbf{Q})^{-1} \mathbf{Q}'$ and $\mathbf{Q}_k^*$ denotes the $k^{th}$ row of $\mathbf{Q}^*$, the population characteristic restriction is given by \begin{equation}
    \label{multfac-rest} \G(\mathbb{P}_Y,\varphi_{X_k})(t) = \varphi_{\mathbf{Y}}^2 (t) \big[\mathbf{Q}_k^*\mathcal{V}_{\mathbf{Y}}(t) - (\log \varphi_{X_k})''(t' \mathbf{A}_k) ].
\end{equation}
With $\mathbf{Y}$ as the observable, the feasible analog $\mathcal{G}(\mathbb{P}_{n,Y} ,\varphi_X)$, obtained by using the empirical measure $\mathbb{P}_{n,Y}$, is given by \begin{equation}
    \label{rep-multfac-feas} \mathcal{G}( \mathbb{P}_{n,Y} ,\varphi_{X_k})(t) = \widehat{\varphi}_{\mathbf{Y}}^2 (t) \big[\mathbf{Q}_k^*\widehat{\mathcal{V}}_{\mathbf{Y}}(t) - (\log \varphi_{X_k})''(t' \mathbf{A}_k) ],
\end{equation}
where $\widehat{\varphi}_{\mathbf{Y}} = \E_n[e^{\mathbf{i} t' \mathbf{Y}}]  $ and $\widehat{\mathcal{V}}_{\mathbf{Y}}(t)$ is the vector of upper triangular elements of $\nabla \nabla ' \log \widehat{\varphi}_{\mathbf{Y}}(t)$. 

In this model, the identifying restriction $\mathcal{G}(\mathbb{P}_Y, \varphi_{X_k}) = \mathbf{0}$ only identifies the latent distribution up to an unknown mean. Intuitively, this is because the quantity $(\varphi_{X_k})''(t)$ does not depend on $\E[X_k]$. Thus, the assumption that $\E[X_k] = 0$ is a normalization that pins down the distribution uniquely. To account for this mean-zero restriction in our setup, we simply demean samples from the usual quasi-Bayes posterior. Specifically, we consider
\begin{align} \label{qb-posterior-demean}  \overline{\nu}_{}(. \: | \:\mathcal{D}_n) \sim   Z - \E[Z]  \; \; \; \; \; \; \text{where} \; \; \; \; \;   Z \sim \nu_{}(.|\:\mathcal{D}_n) . \end{align}
Similar to the analysis in Section \ref{4.1}, our first result examines the limiting behavior of quasi-Bayes under minimal distributional assumptions on the data.
\begin{theorem}[Wasserstein Contraction]
\label{multfac-1}
Assume $(i)$ $F_{X_k}( t \in \R : \left| t \right|  > z) \leq C \exp(- C' z^{\chi})$ for some $\chi, C,C' > 0$, $(ii)$ $\inf_{\| t \|_{\infty} \leq T}  \left| \varphi_{\mathbf{Y}}(t)  \right|    \geq b \exp(-B T^{2})   $ and $ \sup_{\| t \|_{\infty} \leq T} \left| ( \log \varphi_{X_k})'' (t)  \right| \leq D (1+T)^K$ for some $b,B,D,K > 0 $,
$(iii)$ Assumption \ref{cov-prior} holds with some $ \kappa \in (0,1] $ and $v_3 \geq 1$. Then, for any sequence $T_n \lessapprox  \sqrt{\log n}$, there exists a universal constant $D > 0$ such that \begin{align*}
    \overline{\nu} \bigg(F : \mathbf{W}_1(F,F_{X_k}) >  D \max \bigg \{ \frac{1}{T_n} , (\log n)^{-1/2}     \bigg \}   \: \bigg| \: \mathcal{D}_n   
 \bigg) \xrightarrow{\mathbb{P}} 0.
\end{align*}
 In particular, for any $T_n \asymp \sqrt{\log n}$, we have \begin{align*}
    \overline{\nu} \bigg(F : \mathbf{W}_1(F,F_{X_k}) >   D (\log n)^{-1/2}       \: \bigg| \: \mathcal{D}_n   
 \bigg) \xrightarrow{\mathbb{P}} 0.
\end{align*}
\end{theorem}
The exponential lower bound on \( \inf_{\| t \|_{\infty} \leq T}\left|\varphi_{\mathbf{Y}} (t)\right| \) is a very weak restriction and is similar to the condition imposed in Theorem \ref{repmeas-1}. The upper bound on \( \sup_{\|t \|_{\infty} \leq T}\left| \big( \log \varphi_X \big)''(t) \right| \) is very conservative—the quantity is typically uniformly bounded, even in extreme case of Gaussian decay. To the best of our knowledge, Theorem~\ref{multfac-1} is the first nonparametric Bayes and Wasserstein convergence guarantee for the Multi-Factor model.

Next, similar to the analysis in Section \ref{4.2}, we focus on establishing further contraction rates in settings where the latent distribution is known to be sufficiently regular. Specifically, we consider the case when the latent distribution is a nonparametric Gaussian mixture. Our main assumption on the nonparametric Gaussian mixture specification is as follows.
\begin{condition2}[Nonparametric Mixture] \label{exact-g-3} $(i) $ $f_{X_k} = \phi_{\sigma_0^2} \star F_0$ for some $ \sigma_0 > 0$ and mixing distribution $F_0$ that satisfies $ F_0 \big(  t \in \R : \left| t \right| > z  \big) \leq  C \exp(- C' z ^{\chi}) $ for some $\chi > 0$ and universal constants $C,C' > 0$. $(ii)$  $  \sup_{\| t \| \leq T} \left|  ( \log \varphi_{F_0})'' (t) \right| \leq C T $  for some universal constant $C > 0$.
\end{condition2}
The following result demonstrates that, with an appropriately scaled covariance prior, the quasi-Bayes posterior contracts towards the true latent density at a polynomial rate.

\begin{theorem}[Rates with Gaussian Mixtures]
\label{multfac-2} Suppose $ \inf_{\|t \| \leq T} \left| \varphi_{\mathbf{Y}} (t)  \right| \geq c \exp(- c' T^2) $ for some constants $c,c' > 0$ and Condition \ref{exact-g-3} holds. Let $\kappa, \chi  > 0$ be as in Assumption \ref{cov-prior} and  Condition \ref{exact-g-2}, respectively. Define $   \lambda = \max \{ (3/\chi)+1/2 , 2    \}$. Suppose the covariance prior is taken as $G_n \sim G / \sigma_n^2$ where $G$ satisfies Assumption \ref{cov-prior} and $\sigma_n^2  \asymp (\log n)^{- \lambda/\kappa}$. If $T_n \asymp (\log n)^{1/2} $, there exists a universal constant $D>0$ and $V \in (0,1/2)$ such that $$ \overline{\nu} \bigg( \phi_{P,\Sigma} :  \| f_{X}  -  \phi_{P,\Sigma}      \|_{L^2} > D  n^{-V}  \: \bigg| \: \mathcal{D}_n  \bigg) \xrightarrow{\mathbb{P}} 0.  $$
\end{theorem}
Similar to the discussion following Theorem \ref{deconv-prod-t2}, the exponent $V$ depends on second order factors such as the constants $(\sigma_0^2, c')$. 
\begin{remark}[On Joint Estimation] \label{joint-pos}
 If all the latent factors  $(X_1,\dots,X_K)$ are of interest, the identifying restrictions can be used to estimate them jointly. Let $\mathcal{V}_{\mathbf{X}}(t) $ denote the vector with elements $  \{ (\log \varphi_{X_i})''(t' \mathbf{A}_i) \}_{i=1}^K  $. The restriction map for the joint distribution is then given by \begin{align*}
     \mathcal{G}(\mathbb{P}_Y, \varphi_{X_1},\dots,\varphi_{X_K} ) =   {\varphi}_{\mathbf{Y}}^2(t) \big[ {\mathcal{V}}_{\mathbf{Y}}(t) -  \mathbf{Q} \mathcal{V}_{\mathbf{X}}(t)    \big].
 \end{align*}
 As the latent factors are assumed to be mutually independent, we can use independent priors $(\nu_1,\dots,\nu_K)$ for $(F_{X_1},\dots,F_{X_K})$.
\end{remark}

\section{Simulations} \label{sec5}
In this section, we provide simulation evidence on the finite sample performance of quasi-Bayes posteriors. The setup and designs used follows \citet{bonhomme2010generalized}. Two measurements, \( (Y_1, Y_2) \), are generated according to the specification:  
\begin{align*}
    Y_1 &= X + \epsilon_1, \quad \epsilon_1 \sim F_{\epsilon}, \\
    Y_2 &= X + \epsilon_2, \quad \epsilon_2 \sim F_{\epsilon},
\end{align*}
where $(X, \epsilon_1,\epsilon_2 )$ are mutually independent and $X \sim F_X$ is a latent variable whose distribution is of interest. Similar to the analysis in \citet{bonhomme2010generalized}, we can difference the measurements to obtain
 \begin{align} &
  \label{reg1}  \widetilde{Y}_1 =  \bigg(   \frac{Y_1 + Y_2}{2}  \bigg) = X +  \bigg(   \frac{\epsilon_1 + \epsilon_2}{2}  \bigg) \:, \\ &  \widetilde{Y}_2 = \label{reg2} \bigg( \frac{Y_1 - Y_2}{2}   \bigg) = \bigg(  \frac{\epsilon_1 - \epsilon_2}{2}  \bigg)  .
\end{align}
If the nuisance errors are assumed to be symmetric, we can use the observations in (\ref{reg2}) as an auxiliary sample for the error in (\ref{reg1}). This provides us with a sample to construct the measurement error quasi-Bayes posterior of Section \ref{4.1}. 

\subsection{Implementation} \label{5.2}

\textbf{HM}, \textbf{LV}, \textbf{BR} refer to the estimators proposed in \cite{horowitz1996semiparametric}, \cite{li1998nonparametric} and \cite{bonhomme2010generalized}, respectively. \textbf{Decon} denotes the classical kernel deconvolution estimator \citep{carroll1988optimal}. Following \citet{bonhomme2010generalized}; \citet{arellano2023recovering}, the tuning parameters for these estimators are chosen using the optimal bandwidth selection plug-in method proposed by \citet{delaigle2004practical}. This tuning selection procedure is data-driven and so it varies across designs and data realizations. 
Let $\textbf{QB}$ denote the quasi-Bayes posterior mean. In contrast to the other estimators, our implementation of quasi-Bayes (\textbf{QB}) does not vary any tuning parameters—the same choice of \( T \) and prior is applied across all designs and data realizations. 

Samples from the quasi-Bayes posterior are generated using \texttt{Stan}. The concentration parameter of the $\text{DP}_{\alpha}$ prior is $\beta=1$ and the base measure is $\alpha = \mathcal{N}(\E_n(\widetilde{Y}_1), \widehat{Var}(\widetilde{Y}_1))$. The scale prior is $ \sigma \sim \text{Inv-Gamma}(1,1)$. We use the same prior for all designs.  Further implementation details are provided in the Appendix.

\subsection{Results} \label{5.3}

\begin{table}[H]
\caption{$\sqrt{n} \times \text{MISE}$ with $n=1000$. Normal errors: $\epsilon_1, \epsilon_2 \sim \mathcal{N}(0,1)$}
\label{normal-table}
    \centering
    \begin{tabular}{lcccc@{\hspace{15pt}}cc}
         \midrule
         & \textbf{HM} & \textbf{LV} & \textbf{BR} & \textbf{Decon} & \multicolumn{2}{c}{\textbf{QB}} \\
        \cmidrule(lr){6-7}
       \textbf{Distribution} & & & & & $\bm{T=2}$ & $\bm{T=3}$ \\
        \midrule
        $\mathcal{N}(0,1)$ & 0.285 & 0.250 & 0.155 & 0.164 & 0.033 & 0.046 \\
        Laplace$(0,1)$ & 1.107 & 1.171 & 0.822 & 0.791 & 0.246 & 0.226 \\
        Gamma$(2,1)$ & 0.981 & 1.107 & 0.759 & 0.822 & 0.304 & 0.298 \\
        Gamma$(5,1)$ & 0.506 & 0.411 & 0.348 & 0.348 & 0.034 & 0.035 \\
        Exp$(\mathcal{N}(0,1))$ & 7.595 & 7.271 & 11.067 & 9.170 & 1.514 & 1.350 \\
        $0.5\mathcal{N}(-2,1) + 0.5 \mathcal{N}(2,1)$ & 1.962 & 1.993 & 1.519 & 1.297 & 0.036 & 0.038 \\
        \bottomrule
    \end{tabular}
\end{table}
Given an estimator $\widehat{f}$ of the density $f_X$, we denote the MISE by
\[
\mathrm{MISE} = \mathbb{E} \left[ \int \big( \widehat{f}(x) - f_X(x) \big)^2 \, dx \right].
\]
In Table \ref{normal-table}, we report the MISE for the density $f_X$ with Gaussian error distributions. The \textbf{QB} estimator significantly outperforms alternatives, even when the tuning parameter $T$ is held fixed across designs and data realizations. This stands in contrast to traditional Fourier-based deconvolution methods, which are highly sensitive and require careful tuning. One possible interpretation is that the identifying strength in the Fourier moments is significantly amplified when non-trivial regularity, introduced here via a nonparametric prior, is introduced into the parameter space. In Table \ref{laplace-table}, we provide additional results under Laplace errors.

\begin{table}[H]
\caption{$\sqrt{n} \times \text{MISE}$ with $n=1000$. Laplace errors: $\epsilon_1, \epsilon_2 \sim$ Laplace$(0,1)$}
\label{laplace-table}
    \centering
    \begin{tabular}{lcccc@{\hspace{15pt}}cc}
        \midrule
        & \textbf{BR} & \textbf{Decon} & & \multicolumn{2}{c}{\textbf{QB}} \\
        \cmidrule(lr){5-6}
        \textbf{Distribution}  & & & & $\bm{T=2}$ & $\bm{T=3}$ \\
        \midrule
        \textbf{$\mathcal{N}(0,1)$} & 0.133 & 0.148 & & 0.054 & 0.080 \\
        Laplace$(0,1)$ & 0.632 & 0.569 & & 0.279 & 0.265 \\
        Gamma$(2,1)$ & 0.569 & 0.569 & & 0.357 & 0.365 \\
        Gamma$(5,1)$ & 0.221 & 0.218 & & 0.040 & 0.042 \\
        Exp$(\mathcal{N}(0,1))$ & 5.059 & 6.957 & & 1.644 & 1.510 \\
        \textbf{$0.5\mathcal{N}(-2,1) + 0.5 \mathcal{N}(2,1)$} & 1.201 & 0.917 & & 0.055 & 0.060 \\
         \bottomrule
    \end{tabular}
\end{table}

\begin{figure}[H]
    \centering
    \includegraphics[width=1\textwidth]{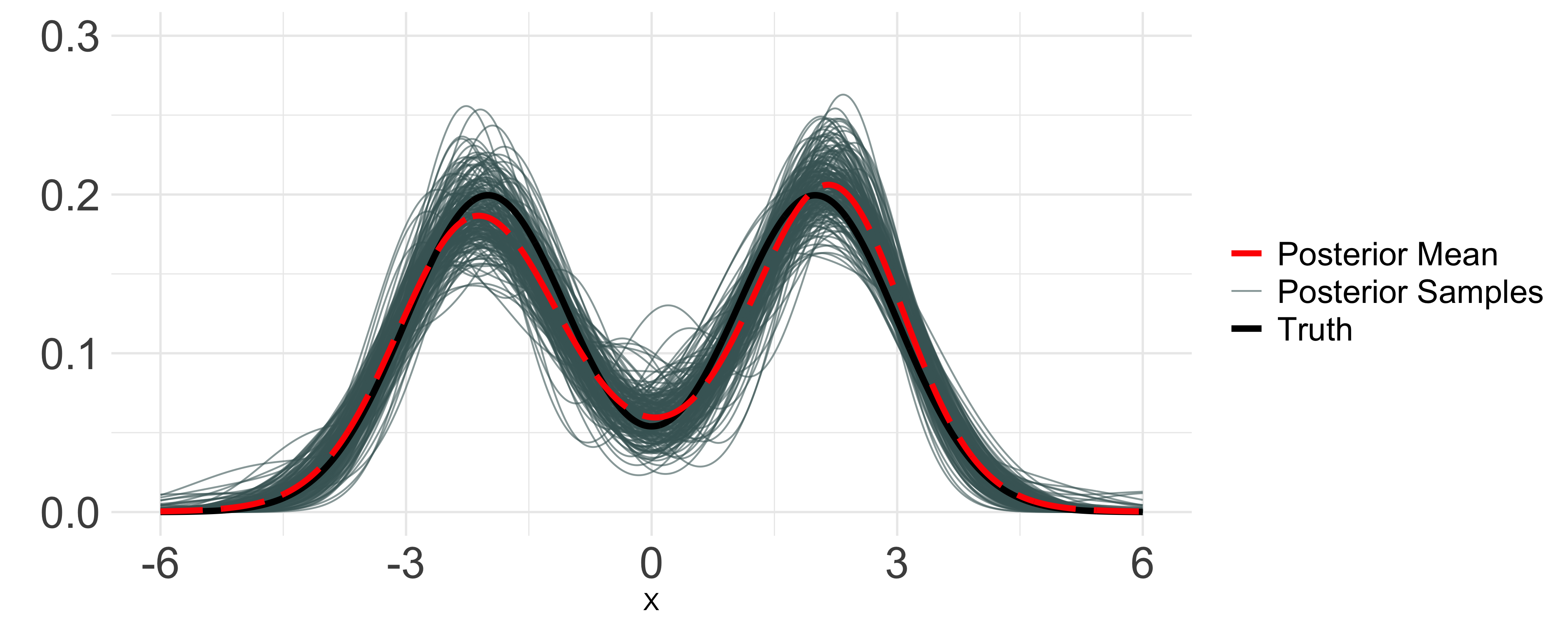}
    \caption{\textbf{QB} estimation of $f_X$. $n=1000$, $X \sim 0.5 \mathcal{N}(-2,1) + 0.5 \mathcal{N}(2,1) $. $\epsilon_1,\epsilon_2 \sim \mathcal{N}(0,1)$}
    \label{deconv-multimod}
\end{figure}

The results in Tables \ref{normal-table} and \ref{laplace-table} also demonstrate that convergence is exceptionally fast when the latent distribution has a nonparametric Gaussian mixture structure, consistent with the statement of Theorem \ref{deconv-2}.

In Table \ref{wass-table}, we report the $L^1$ Wasserstein distance.\footnote{An advantage of focusing on the $\mathbf{W}_1$ metric is that it is well-defined even between mutually singular distributions (e.g. discrete and continuous).} For comparison, we consider the empirical Bayes (denoted as \textbf{EB}) non-parametric maximum likelihood estimator (NPMLE) of $F_X$. The likelihood defining the \textbf{EB} estimator assumes Gaussian noise—we estimate the noise scale using the standard deviation of the measurements in (\ref{reg2}). In addition to the distributions mentioned above, we include two additional discrete designs. The first is a two-point discrete distribution,
$
F_D = \frac{1}{8} \delta_{x=0} + \frac{7}{8} \delta_{x=2}.
$
This distribution was previously used in  \citet{koenker2019comment} to assess the $\mathbf{W}_1$ performance of the \textbf{EB} estimator. The second is a discrete uniform  distribution, \( U(\mathcal{S}) \), which assigns equal probability mass to the set of points
$
\mathcal{S} = \{-0.5, -0.25, 0, 0.25, 0.5\}.
$

\begin{table}[H]
\caption{$n=1000$. Mean Wasserstein $L^1$ error : $\E \big[\mathbf{W}_1(\widehat{F},F_X)\big]$}
\label{wass-table}
    \centering
    \setlength{\tabcolsep}{5pt} 
    \begin{tabular}{lccc@{\hspace{20pt}}ccc}
        \midrule
        & \multicolumn{3}{c}{\textbf{Normal Errors}} & \multicolumn{3}{c}{\textbf{Laplace Errors}} \\
        \cmidrule(lr){2-4} \cmidrule(lr){5-7}
        & \textbf{EB} & \multicolumn{2}{c}{\textbf{QB}} & \textbf{EB} & \multicolumn{2}{c}{\textbf{QB}} \\
        \cmidrule(lr){3-4} \cmidrule(lr){6-7}
       \textbf{Distribution}  &  & $\bm{T=2}$ & $\bm{T=3}$ &  & $\bm{T=2}$ & $\bm{T=3}$ \\
        \midrule
        $\mathcal{N}(0,1)$ & 0.317 & 0.053 & 0.056 & 0.406 & 0.066 & 0.069 \\
        Laplace$(0,1)$ & 0.348 &  0.138 & 0.133 & 0.453 & 0.144 & 0.140 \\
        Gamma$(2,1)$ & 0.327 & 0.158 & 0.165 & 0.453 & 0.169 & 0.173 \\
        Gamma$(5,1)$ & 0.377 & 0.227 & 0.229 & 0.496 & 0.228 & 0.226 \\
        Exp$(\mathcal{N}(0,1))$ & 0.332 & 0.364 & 0.359 & 0.483 & 0.367 & 0.361 \\
        $0.5\mathcal{N}(-2,1) + 0.5 \mathcal{N}(2,1)$ & 0.358 & 0.090 & 0.095 & 0.468 & 0.106 & 0.109 \\ $F_{D}$ & 0.096 & 0.350 & 0.265 & 0.164 & 0.384 & 0.307 \\
        $U(\mathcal{S})$ & 0.198 & 0.107 & 0.109 & 0.342 & 0.118 & 0.119 \\
        \bottomrule
    \end{tabular}
\end{table}

 The \textbf{QB} estimator performs remarkably well across all designs. As expected, under misspecification, the \textbf{EB} Wasserstein risk increases substantially when the true errors are non-Gaussian. Since the quasi-Bayes likelihood only depends on moments of the observable $(Y_1,Y_2)$, there is no misspecification, and the results are very similar across the two error designs. Consistent with the findings in \citet{koenker2019comment}, the \textbf{EB} performs exceptionally well for the two-point discrete distribution $F_D$ but experiences a significant decline in accuracy with continuous designs. As the distribution $U(\mathcal{S})$ illustrates, the accuracy loss can also be significant for discrete distributions with a modest number of support points. In this case, the worse performance was primarily due to the \textbf{EB} concentrating nearly all its mass on only two support points.

\section{Conclusion} \label{conclus}
This paper introduces and develops a quasi-Bayes framework for a large class of latent variable models. Simulations demonstrate that our quasi-Bayes procedures are viable and perform remarkably well relative to existing alternatives. As an application, we used our method to model the latent structure in U.S earnings data. In this section, we discuss several remarks and extensions concerning optimal weighting, the estimation of functionals, and frequentist inference.

The quasi-Bayes characteristic likelihood in (\ref{feas-quas}) is based on identity weighted moment conditions. This framework admits a straightforward extension to allow for optimal weighting. Interestingly, in this case, the objective function shares similarities with the definition proposed in \citet{carrasco2000generalization}. In practice, the integral in (\ref{feas-quas}) is typically approximated using discrete measure over grid points ${t_1,\dots,t_M} \in \mathbb{B}(T)$. As such, the classical optimally weighted GMM objective at these points can be used to approximate the optimally weighted characteristic likelihood, albeit at the cost of some discretization bias. A formal treatment that accounts for the computational discretization bias would be useful. Ideally, one can also exploit the structure of the problem for greater computational efficiency. For instance, since $\left| \varphi_{Y}(t) \right| \rightarrow 0$ quickly as $\|t \| \rightarrow \infty$, it is reasonable to expect that placing more emphasis on accurate optimal weighting at lower frequency levels will be more beneficial.\footnote{I thank Victor Chernozhukov for this observation; see also \cite{ben2021identification}.}

An important question concerns the relevance of these results for downstream applications that only depend on a functional of the latent distribution. As discussed in Remark~\ref{functionals} and Section \ref{sec6}, when the functional of interest is sufficiently regular, our derived rates admit significantly faster counterparts for that functional. Indeed, this will usually hold whenever the functional is a conditional mean that conditions on information from the observables. In general, however, the regularity of an arbitrary functional is often unknown and may vary substantially with the distribution of the latent unobservables. For this reason, our exposition focuses primarily on the direct estimation of the latent distribution. This is especially important in non-Gaussian models (e.g. \citealp{guvenen2014nature}; \citealp{arellano2017earnings}), where the rates for functionals are typically much slower than in the Gaussian setting.

Quasi-Bayes is an attractive tool in that it not only provides point estimates but also an automatic measure of uncertainty through the posterior distribution. It is natural to question whether such credible sets have valid frequentist guarantees. We note that  inferential guarantees are usually difficult to obtain in any model with an infinite dimensional latent distribution. For example, downstream empirical applications of the NPMLE empirical Bayes framework (e.g. \citealp{gu2023invidious}; \citealp{kline2024discrimination}) typically rely exclusively on point estimation, as no  distributional theory is available. Similarly, common models for earnings dynamics (e.g. \citealp{arellano2017earnings}; \citealp{arellano2024heterogeneity})  provide inferential guarantees only under correct parametric specification. For suitable parametric versions of our models and priors, classical results by \citet{chernozhukov2003mcmc} imply the frequentist validity of optimally weighted quasi-Bayes credible sets. Extending these results to our nonparametric setting is an important topic for further investigation.\footnote{Recently, \cite{kankanala2023gaussian} showed that optimally weighted quasi‐Bayes credible have valid frequentist coverage for regular functionals in a nonparametric conditional moment restriction model. } 

Finally, we note that our work provides the first analysis of infinite dimensional latent heterogeneity within a quasi-Bayes framework. Many of our results also provide the first nonparametric Bayes guarantees for the models under consideration. The quasi-Bayes approach is particularly appealing in these settings because the models are often highly ill-posed and require non-trivial forms of regularization. A key takeaway from our findings is that classical nonparametric Bayes priors can significantly enhance the finite sample information content in the moments. This estimation framework complements much of the foundational analysis in the literature (e.g. \citealp{bonhomme2010generalized}) that derived the identifying moments. In this article, we focused on the class of models characterized by Fourier moment conditions, given their widespread use in the literature. The analysis could be extended to many other settings that use moments to characterize infinite dimensional latent heterogeneity (e.g.  \citealp{bonhomme2012functional}).

\bibliographystyle{style}

\bibliography{main.bib}

\newpage

 \section{Appendix : Implementation} \label{append-implem} 
 The integral in the likelihood defining the quasi-Bayes posterior is approximated via Monte Carlo integration, with $M$ grid points. This can be viewed as a scaled GMM objective, based on  $M$ moment conditions. For accuracy in high dimensions, the grid points are generated using Latin hypercube sampling.\footnote{Generated using the \texttt{lhs} package in \textsf{R}.} Note that the quasi-Bayes likelihood is differentiable in the $\text{DP}_{\alpha}$ prior parameters. Samples from the quasi-Bayes posterior are generated using Hamiltonian Monte Carlo (HMC) on \texttt{Stan}, using two chains with a burn-in of 4000 iterations and 4000 post-burn-in samples.

 \subsection{Simulations}
The number of replications used to produce all the results is $1000$. We use the same prior for all designs. The concentration parameter of the $\text{DP}_{\alpha}$ prior is $\beta=1$ and the base measure is $\alpha = \mathcal{N}(\E_n(\widetilde{Y}_1), \widehat{Var}(\widetilde{Y}_1))$. The scale prior is $ \sigma \sim \text{Inv-Gamma}(1,1)$. The sample size is $n=1000$ and the $\text{DP}_{\alpha}$ prior is truncated at $K=15$ terms.\footnote{We did not observe any significant differences when increasing $K$.} We use $M=2000$ grid points in the Monte Carlo approximation to the likelihood.
 \subsection{Empirical Application}
 We use the measurements $\Delta w$ to construct the joint multi-factor quasi-Bayes posterior. All the measurements are normalized as a first step. In the normalized model, since the entries of $\mathbf{A} $ belong to the set $ \{-1,0,1\}$, it follows that each latent factor has mean zero and variance less than $1$. Hence, for the $\text{DP}_{\alpha}$ prior, a natural choice is $\alpha = \mathcal{N}(0,1)$ as the base measure.\footnote{I thank several seminar participants for this implementation suggestion.} For the prior on the standard deviation component $\sigma$, we use a Half-Cauchy$(0,1)$ prior.\footnote{Half-Cauchy priors  are an attractive choice for a weakly informative prior on the scale, particularly in settings with potential low dispersion \citep{gelman2006prior}.} We use $\beta=1$ as the  $\text{DP}_{\alpha}$ concentration parameter.  
 
  The dataset consists of a panel of $n=624$ individuals, observed over 10 time periods. The model contains 17 latent distributions. Each latent distribution is modeled with a $\text{DP}_{\alpha}$ prior, truncated at $K=10$ terms. We use $M=10000$ grid points in the Monte Carlo approximation to the likelihood.

  To choose the grid of identifying restrictions
\(\mathbb{B}(T) = \{ t \in \mathbb{R}^M : \|t\|_{\infty} \leq T \}\), we began with \(T = 0\) and incrementally increased its value by \(\delta = 0.05\) 
until the implied posterior variance of \(\Delta w_2\) resembled the observed empirical variance $\E_n(\left| \Delta w_2 \right|^2)$. Our chosen value $T=0.55$ was the first instance that provided an adequate fit: $\E(\left| \Delta_{w_2}  \right|^2 \: | \: \mathcal{D}_n ) - \E_n(\left| \Delta w_2 \right|^2) \approx 0.01$. We view this as a heuristic choice. Optimal selection of $T$ in a high-dimensional setting like this is an important topic for further investigation.

\section{Appendix : Proofs} \label{proofs}
The following notation is frequently referred to in the proofs and so is listed here for convenience. Let $\mathcal{F}$ denote the Fourier transform operator. Let $\lambda$ denote the Lebesgue measure on $\R^d$. We denote the diameter of a set $A \subset \R^d$ by $\mathcal{D}(A) = \sup_{x,y \in A} \|   x-y\| $. Given a positive definite matrix $\Sigma \in \mathbf{S}_+^d$, we denote the ordered eigenvalues by $\lambda_1(\Sigma) \leq \dots \leq \lambda_d(\Sigma)$. 

\begin{lemma}
\label{aux1}
Suppose $ W \in \mathbb{R}^{d^*}  $ and $Z \in \mathbb{R}^d$ for some $d^*,d \in \mathbb{N}$. For every $t \in \R^d$, define  \begin{align*}
    \chi_t(W,Z)  =  We^{\mathbf{i} t' Z} = W\cos(t'Z) + \mathbf{i}  W \sin(t'Z).
\end{align*}
Suppose $\E\big( \| Z \|^2  \big) < \infty $ and $ \E\big( \| W \|^2  \big) < \infty $. Then, there exists a universal constant $D > 0$ such that for every $T > 0$ we have
\begin{align*}
  \E  \bigg( \sup_{ \| t \|_{\infty} \leq T} \| \E_n \big[ \chi_t(W,Z)     \big] - \E \big[ \chi_t(W,Z)     \big] \| \bigg) \leq D \frac{ \max \big \{ \sqrt{\log T} , 1  \big \}  }{\sqrt{n}}.
\end{align*}

\end{lemma}

\begin{proof}[Proof of Lemma \ref{aux1}]
It suffices to show the result for the real and imaginary part separately. We verify it for the real part, the imaginary part is completely analogous. If $W= (W_1,W_2,\dots,W_{d_1})$, it suffices to verify the result for each vector sub component $ W_k e^{\mathbf{i} t' Z} $ where $k \in \{1 , \dots, d^* \}$. Fix any $k \in \{1 , \dots , d^* \}$. For the remainder of this proof, we continue assuming $\chi_t(W,Z) =  W_k \cos(t'Z)$. Define the class of functions $\mathcal{F} = \{   \chi_t(W,Z) : \| t \|_{\infty} \leq T \}$. Since $\|W \|$ is an envelope of $\mathcal{F}$, an application of \citep[Remark 3.5.14]{gine2021mathematical} implies that  there exists a universal constant $ L > 0 $ such that \begin{align*}
\E  \bigg(\sup_{ \| t \|_{\infty} \leq T}  \left| \E_n \big[ \chi_t(W,Z)     \big] - \E \big[ \chi_t(W,Z)     \big] \right| \bigg) \leq \frac{L}{\sqrt{n}} \int_{0}^{8 \| W \|_{L^2(\mathbb{P})}     }  \sqrt{\log N_{[]}(\mathcal{F}, \| . \|_{L^2(\mathbb{P})},  \epsilon  ) } d \epsilon.
\end{align*}
Let $\{  t_i \}_{i=1}^M  $  denote a minimal $\delta > 0$ covering of  $[-T,T]^{d}$. Define the functions $$ e_{i}(W,Z) = \sup_{ t \in \R^d :   \| t \|_{\infty} \leq T , \|  t- t_i \|_{\infty} < \delta  } \ \big| \:   \chi_t(W,Z)  -   \chi_{t_i}(W,Z)   \:   \big| \; \; \; \; \; i=1,\dots,M . $$
It follows that $ \big \{   \chi_{t_i}(W,Z)   - e_i \; ,  \;  \chi_{t_i}(W,Z)  + e_i   \big \}_{i=1}^M  $ is a bracket covering for $\mathcal{F}$. Since the mapping $t \rightarrow \chi_t(W,Z)$ has Lipschitz constant bounded by $ \| W \| \|Z \| $, Cauchy-Schwarz implies that $
    \| e_i \|_{L^2(\mathbb{P})}^2    \leq  \| W \|_{L^2(\mathbb{P})}^2 \| Z \|_{L^2(\mathbb{P})}^2  \delta^2$. Since $M \leq (3 T \delta^{-1})^d$, it follows that there exists a universal constant $L >0$ such that \begin{align*}
    \int_{0}^{8 \| W \|_{L^2(\mathbb{P})}     }  \sqrt{\log N_{[]}(\mathcal{F}, \| . \|_{L^2(\mathbb{P})},  \epsilon  ) } d \epsilon \leq L \max \big \{ \sqrt{\log T} ,1   \big  \}.
\end{align*}
\end{proof}

\begin{lemma}
\label{aux2}
Suppose $ W \in \mathbb{R}^{d^*}  $ and $Z \in \mathbb{R}^d$ for some $d^*,d \in \mathbb{N}$. For every $t \in \R^d$, define  \begin{align*}
    \chi_t(W,Z)  =  We^{\mathbf{i} t' Z} = W\cos(t'Z) + \mathbf{i}  W \sin(t'Z).
\end{align*}
Suppose $\E\big( \| Z \|^2  \big) < \infty $ and $ \E\big( \| W \|^2  \big) < \infty $. Then, there exists a universal constant $L > 0$ such that for any sequence $T_n \uparrow \infty$ with $\log(T_n) = o(n)$, we have that
\begin{align*}
  \mathbb{P}  \bigg( \sup_{ \| t \|_{\infty} \leq T_n} \left| \E_n \big[ \chi_t(W,Z)     \big] - \E \big[ \chi_t(W,Z)     \big] \right| \leq D \frac{\sqrt{\log T_n}}{\sqrt{n}} \bigg)  \rightarrow 1.
\end{align*}

\end{lemma}

\begin{proof}[Proof of Lemma \ref{aux2}]
It suffices to show the result for the real and imaginary part separately. We verify it for the real part, the imaginary part is completely analogous. If $W= (W_1,W_2,\dots,W_{d^*})$, it suffices to verify the result for each vector sub component $ W_k e^{\mathbf{i} t' Z} $ where $k \in \{1 , \dots, d^* \}$. Fix any $k \in \{1 , \dots , d^* \}$. For the remainder of this proof, we continue assuming $\chi_t(W,Z) =  W_k \cos(t'Z)$. For a given sequence of deterministic constants $L_n \uparrow \infty$, define   \begin{align*}  &  \chi_{1,t}(W,Z)  =  \chi_t(W,Z)  \mathbbm{1} \big \{  \| W \| \leq L_n   \big \} \;, \\ &  \chi_{2,t}(W,Z)  =  \chi_t(W,Z)  \mathbbm{1} \big \{  \| W \| > L_n   \big \}  \; .      \end{align*}
Observe that \begin{align*}
    (\E_n - \E_{})[\chi_t(W,Z) ]   =\sum_{i=1}^n  \Xi_{1,t} (W_i,Z_i)  + \sum_{i=1}^n \Xi_{2,t} (W_i,Z_i) \; ,
\end{align*}
where $ \Xi_{1,t} (W,Z) = n^{-1} [\chi_{1,t}(W,Z) - \E_{} \chi_{1,t}(W,Z) ] $ and $ \Xi_{2,t} (W,Z) = n^{-1} [\chi_{2,t}(W,Z) - \E_{} \chi_{2,t}(W,Z) ]$. 

First, we derive a bound for $ \sum_{i=1}^n \Xi_{2,t} (W_i,Z_i)$. Observe that \begin{align*}
    \mathbb{P} \bigg( \sup_{\| t \|_{\infty} \leq T_n}  \left|  \sum_{i=1}^n \Xi_{2,t} (W_i,Z_i)    \right| > \frac{\sqrt{\log(T_n)}}{\sqrt{n}}   \bigg)  & \leq  \frac{\sqrt{n}}{ \sqrt{\log(T_n)} } \E \bigg( \sup_{\| t \|_{\infty} \leq T_n} \sum_{i=1}^n \left|  \Xi_{2,t} (W_i,Z_i)   \right|      \bigg) \\ & \leq  \frac{2 \sqrt{n}}{ \sqrt{\log(T_n)} } \E \big( \| W \|  \mathbbm{1} \big \{ \|W \| > L_n    \big \}   \big) \\ & \leq \frac{2 \sqrt{n}}{ \sqrt{\log T_n} L_n }  \E \big( \| W \|^2  \mathbbm{1} \big \{ \|W \| > L_n    \big \}   \big) .
\end{align*}
Since $ \E\big( \| W \|^2  \big) < \infty $, the term on the right is $o(1)$ when $L_n = \sqrt{n}/ \sqrt{\log T_n}$. 

It remains to bound the first sum $\sum_{i=1}^n  \Xi_{1,t} (W_i,Z_i) $ when $L_n = \sqrt{n}/ \sqrt{\log T_n}$. We have \begin{align*}
    &  \sup_{\| t \|_{\infty} \leq T_n} \E \big[ \left| \Xi_{1,t} (W,Z)  \right|^2    \big] \leq n^{-2} \| W \|_{L^2(\mathbb{P})}^2 \\ &  \sup_{\| t \|_{\infty} \leq T_n} \left| \Xi_{1,t} (W,Z)     \right| \leq 2 n^{-1} L_n.
\end{align*}
 The preceding bounds, Lemma \ref{aux1} and   \citep[Theorem 3.3.9]{gine2021mathematical} imply that there exists universal constants $C,D > 0$ which satisfy \begin{align*}
    & \mathbb{P}  \bigg( \sup_{  \|t  \|_{\infty} \leq T_n} \left| \E_n \big[ \chi_t(Z)     \big] - \E \big[ \chi_t(Z)     \big] \right| >   D \frac{\sqrt{\log T_n}}{\sqrt{n}}  + x  \bigg) \\ & \leq  \exp \bigg(  - \frac{x^2}{C n^{-1} \big[ L_n \sqrt{\log T_n}/ \sqrt{n} + \| W \|_{ L^2(\mathbb{P})   }^2 + x L_n   ]  }     \bigg)
\end{align*}
for all $x > 0$. The choice $x =  \sqrt{\log T_n} / \sqrt{n}  $  with  $L_n = \sqrt{n} / \sqrt{\log T_n}$ yields
 \begin{align*}
     \mathbb{P}  \bigg( \sup_{  \|t  \|_{\infty} \leq T_n} \left| \E_n \big[ \chi_t(Z)     \big] - \E \big[ \chi_t(Z)     \big] \right| >   D \frac{\sqrt{\log T_n}}{\sqrt{n}}  +  \frac{\sqrt{\log T_n}}{\sqrt{n}}   \bigg) \leq \exp \big(  - E \log(T_n)       \big)
 \end{align*}
for some universal constant $ E > 0$. Since $T_n \uparrow \infty$, the claim follows.
\end{proof}

\begin{lemma}
\label{aux3}

 Consider a measurable partition $\R^d = \bigcup_{j=0}^N  V_j $ and points $z_j \in V_j$ for $j=1,\dots,N$. Let $F^* = \sum_{j=1}^N w_j \delta_{z_j}  $ denote the discrete probability measure with weight $w_j$ at $z_j$. Then, for any probability measure $F$ on $\R^d$ with $\int_{\R^d} \| x \|^2 d F(x) < \infty  $, we have that

\begin{enumerate}
    \item $$ \left| \mathcal{F}[F](t) - \mathcal{F}[F^*] (t) \right| \leq \| t \| \sup_{j=1,\dots,N} \mathcal{D}(V_j) + 2 \sum_{j=1}^N \left|  F(V_j) - w_j  \right|. $$
    \item     \begin{align*}
         \| \nabla \mathcal{F}[F](t) - \nabla \mathcal{F}[F^*] (t) \| & \leq \bigg( \int_{\R^d} \| x \|^2 dF (x)   \bigg)^{1/2}  \bigg \{  \sum_{j=1}^N \left| F(V_j) - w_j  \right|   \bigg \}^{1/2}    \\ &    + \bigg( 1+ \sup_{ v \in \bigcup_{j=1}^N V_j  } \| v \|  \| t \| \bigg) \sup_{j=1,\dots,N} \mathcal{D}(V_j)     + \sup_{i=1,\dots,N} \| z_i \| \sum_{j=1}^N \left| F(V_j) - w_j  \right|.
    \end{align*}

    \item  \begin{align*}
      \sup_{k=1,\dots,d} \left|  \partial_{t_k}^2 \F[F](t)   - \partial_{t_k}^2 \F[F^*](t)   \right| &  \leq \bigg( \int_{\R^d} \| x \|_{\infty}^4 dF (x)   \bigg)^{1/2}  \bigg \{  \sum_{j=1}^N \left| F(V_j) - w_j  \right|   \bigg \}^{1/2}    \\ &   + \bigg( 2   \sup_{ v \in \bigcup_{j=1}^N V_j  } \| v \|   + \sup_{ v \in \bigcup_{j=1}^N V_j  }   \| v \|^2  \| t \| \bigg) \sup_{j=1,\dots,N} \mathcal{D}(V_j)   \\ &    + \sup_{i=1,\dots,N} \| z_i \|^2 \sum_{j=1}^N \left| F(V_j) - w_j  \right| .
     \end{align*}
\end{enumerate}
Here, $\mathcal{D}(A)$ denotes the diameter of a set $A$.
\end{lemma}

\begin{proof}[Proof of Lemma \ref{aux3}] Observe that
 \begin{align*}
    & \mathcal{F}[F](t) - \mathcal{F}[F^*] (t ) \\ & =  \int_{\R^d} e^{ \mathbf{i} t'x} dF(x) - \int_{\R^d} e^{ \mathbf{i} t'x} dF^*(x)   \\ &  =  \int_{\R^d} e^{ \mathbf{i} t'x} dF(x) - \sum_{j=1}^N w_j e^{\mathbf{i} t' z_j}  \\ & =  \int_{V_0} e^{\mathbf{i} t'x} dF(x) + \sum_{j=1}^N \int_{V_j} \big( e^{\mathbf{i} t'x} - e^{\mathbf{i} t' z_j}    \big) dF(x) + \sum_{j=1}^N e^{\mathbf{i} t' z_j} \big[ F(V_j) - w_j     ].
\end{align*}
Since the mapping $\mu \rightarrow e^{\mathbf{i} t'  \mu } $ has Lipschitz constant at most $\| t \|$ and $ \sum_{j=1}^N F(V_j) = 1  $, we obtain \begin{align*}
    \left| \mathcal{F}[F](t) - \mathcal{F}[F^*] (t)   \right| & \leq F(V_0) + \| t \|  \sup_{j=1,\dots,N}  \mathcal{D}(V_j) \sum_{j=1}^N F(V_j)    + \sum_{j=1}^N \left|  F(V_j) - w_j  \right| \\ & \leq F(V_0) + \| t \|  \sup_{j=1,\dots,N} \mathcal{D}(V_j)   + \sum_{j=1}^N \left|  F(V_j) - w_j  \right|.
\end{align*}
Since $ V_0 = \R^d \setminus \cup_{j \geq 1} V_j  $ and $F(\R^d) = 1 = \sum_{j=1}^N w_j$, we obtain $$  F(V_0) = \sum_{j=1}^N w_j - \sum_{j=1}^N F(V_j) \leq  \sum_{j=1}^N \left| F(V_j) - w_j  \right| . $$
For the second claim, first observe that the moment condition on $F$ ensures that the gradient exists. By differentiating the preceding expression for $\mathcal{F}[F](t) - \mathcal{F}[F^*] (t )$, we obtain \begin{align*}
    &  \nabla\mathcal{F}[F](t) - \nabla \mathcal{F}[F^*] (t ) \\ & = \mathbf{i} \int_{V_0} x e^{\mathbf{i} t'x} dF(x) + \mathbf{i}  \sum_{j=1}^N \int_{V_j} \big(  x e^{\mathbf{i} t'x} -  z_j e^{\mathbf{i} t' z_j}    \big) dF(x) + \mathbf{i}\sum_{j=1}^N z_j e^{\mathbf{i} t' z_j} \big[ F(V_j) - w_j     ] .
\end{align*}
For $\mu \in V_j$, the mapping $\mu \rightarrow  \mu e^{\mathbf{i} t'  \mu } $ has Lipschitz constant at most $ 1+  \sup_{v \in V_j }\| v \|  \| t \| $ and by Cauchy-Schwarz, $ \int_{V_0} \| x \| d F(x) \leq \big( \int_{\R^d} \| x \|^2 dF (x)   \big)^{1/2} \big\{  F(V_0)   \big \}^{1/2}$. By using the same bound for $F(V_0)$ as above, we obtain \begin{align*}
   & \|  \nabla\mathcal{F}[F](t) - \nabla \mathcal{F}[F^*] (t ) \| \\ & \leq \bigg( \int_{\R^d} \| x \|^2 dF (x)   \bigg)^{1/2}  \bigg \{  \sum_{j=1}^N \left| F(V_j) - w_j  \right|   \bigg \}^{1/2} +  \bigg( 1+  \sup_{ v \in \cup_{j=1}^N V_j  } \| v \|    \| t \|  \bigg) \sup_{j=1,\dots,N} \mathcal{D}(V_j) \\ & \;  + \sup_{i=1,\dots,N} \| z_i \| \sum_{j=1}^N \left| F(V_j) - w_j  \right|.
\end{align*}
For the third claim, let $z_{j,k}$ denote the $k^{th}$ element of $z_j$. Twice differentiating (with respect to $t_k$) the expression for $\mathcal{F}[F](t) - \mathcal{F}[F^*] (t )$  yields \begin{align*}
    &   \partial_{t_k}^2 \mathcal{F}[F](t) - \partial_{t_k}^2 \mathcal{F}[F^*] (t ) \\ & = - \int_{V_0} x_k^2 e^{\mathbf{i} t'x} dF(x) -  \sum_{j=1}^N \int_{V_j} \big(  x_k^2 e^{\mathbf{i} t'x} -  z_{j,k}^2 e^{\mathbf{i} t' z_j}    \big) dF(x) - \sum_{j=1}^N z_{j,k}^2 e^{\mathbf{i} t' z_j} \big[ F(V_j) - w_j     ] .
\end{align*}
For $\mu = (\mu_1,\dots,\mu_d) \in V_j$ and a fixed component $\mu_k$  the mapping $\mu \rightarrow \mu_k^2 e^{\mathbf{i} t' \mu} $ has Lipschitz constant at most $2 \sup_{v \in V_j} \| v \| + \sup_{v \in V_j} \| v \|^2 \| t \|  $ and by Cauchy-Schwarz, we have the bound $ \int_{V_0}  x_k^2 d F(x) \leq \big( \int_{\R^d}  \| x \|_{\infty}^4  dF (x)   \big)^{1/2} \big\{  F(V_0)   \big \}^{1/2}$. The claim follows from an analogous bound to the preceding case.

\end{proof}

\begin{lemma}
\label{aux4}

Consider a measurable partition $\R^d = \bigcup_{j=0}^N  V_j $ and points $z_j \in V_j$ for $j=1,\dots,N$. Let $F^* = \sum_{j=1}^N w_j \delta_{z_j}  $ denote the discrete probability measure with weight $w_j$ at $z_j$. Then, there exists a universal constant $D > 0 $ such that for any positive definite matrix $\Sigma \in \R^{d \times d}$ with minimum eigenvalue $\underline{\sigma}^2 > 0$ and probability measure $F$ on $\R^d$, we have that \begin{align*} 
  & (i) \; \; \; \| \varphi_{F,\Sigma}  - \varphi_{F^*,\Sigma}  \|_{L^2} \leq D \bigg[ \underline{\sigma}^{-(d+2)/2} \sup_{j=1,\dots,N} \mathcal{D}(V_j) + \underline{\sigma}^{-d/2}  \sum_{j=1}^N \left|  F(V_j) - w_j  \right|    \bigg] ,  \\ & (ii)  \; \; \| \varphi_{F,\Sigma}  - \varphi_{F^*,\Sigma}  \|_{\mathbb{B}(T)} \leq D  \bigg[ T^{(d+2)/2} \sup_{j=1,\dots,N} \mathcal{D}(V_j) + T^{d/2}  \sum_{j=1}^N \left|  F(V_j) - w_j  \right|   \bigg].
\end{align*}
\end{lemma}

\begin{proof}[Proof of Lemma \ref{aux4}]
By Lemma \ref{aux3}, we have \begin{align*}
    \left| \mathcal{F}[F](t) - \mathcal{F}[F^*] (t) \right| \leq \| t \| \sup_{j=1,\dots,N} \mathcal{D}(V_j) + 2 \sum_{j=1}^N \left|  F(V_j) - w_j  \right| .
\end{align*}
Since $  e^{- t' \Sigma t} \leq e^{- \| t \|^2 \underline{\sigma}^2} $, it follows that \begin{align*}
    & \|  \varphi_{F,\Sigma} - \varphi_{F^*,\sigma}  \|_{L^2}^2 \\ & = \int_{\R^d} \left|  \varphi_{\Sigma} (t)  \right|^2  \left| \mathcal{F}[F](t) - \mathcal{F}[F^*] (t) \right| ^2 dt \\ & = \int_{\R^d} e^{- t' \Sigma t}  \left| \mathcal{F}[F](t) - \mathcal{F}[F^*] (t) \right| ^2 dt  \\ & \leq \int_{\R^d} e^{- \|t \|^2 \underline{\sigma}^2}  \left| \mathcal{F}[F](t) - \mathcal{F}[F^*] (t) \right| ^2 dt \\ & \leq  2 \sup_{j=1,\dots,N} \{ \mathcal{D}(V_j) \}^2 \int_{\R^d} e^{- \|t \|^2 \underline{\sigma}^2} \| t \|^2 dt +  8 \bigg \{ \sum_{j=1}^N \left|  F(V_j) - w_j  \right| \bigg \}^2 \int_{\R^d} e^{- \|t \|^2 \underline{\sigma}^2} dt.
 \end{align*}
The first claim follows from observing that the two integrals scale with rate at most $\underline{\sigma}^{-(d+2)}$ and $\underline{\sigma}^{-d}$, respectively. The second claim follows by an analogous argument from truncating the integral to the set $\{ t \in \R^d : \|t\|_{\infty} \leq T  \}$ and using the trivial bound $e^{-t' \Sigma t} \leq 1$.

\end{proof}

\begin{lemma}
\label{aux5}
Suppose $F$ is a probability measure supported on $[-R,R]^d$ for some $ R > 0$. Then, given any $k \geq 1$, there exists a discrete probability measure $ F' $ with at most $(k+1)^d+1$ support points in $[-R,R]^d$ such that for every $t \in \R^d$, we have \begin{align*}
    & (i) \; \; \; \;  \left|  \F[F](t) - \F[F'] (t)  \right| \leq  2 \frac{\| t \|^{k+1}  (e \sqrt{d} R)^{k+1}  }{(k+1)^{k+1}} \; , \\ & (ii) \; \; \; \;  \|  \nabla \F[F](t) - \nabla \F[F'] (t)  \| \leq 2 \sqrt{d}R \frac{\| t \|^{k}  (e \sqrt{d} R)^{k}  }{k^{k}} \; , \\ & (iii) \; \;  \sup_{l=1,\dots,d}  \left| \partial_{t_l}^2 \F[F](t) - \partial_{t_l}^2  \F[F'] (t)  \right| \leq 2 dR^2\frac{\| t \|^{k-1}  (e \sqrt{d} R)^{k-1}  }{(k-1)^{k-1}} .
\end{align*}
In particular, there exists universal constants $C,D > 0$ such that for all $T,R $ sufficiently large and $\epsilon \in (0,1)$, the choice $ k = \lceil C \max \{ \log(\epsilon^{-1}) , R T   \} \rceil     $  satisfies \begin{align*}
    & (i) \; \; \; \;  \sup_{\| t \|_{\infty} \leq T}  \left|  \F[F](t) - \F[F'] (t)  \right| \leq  D \epsilon \; , \\ & (ii) \; \; \; \; \sup_{\| t \|_{\infty} \leq T}   \|  \nabla \F[F](t) - \nabla \F[F'] (t)  \| \leq D \epsilon \; , \\ & (iii)  \; \; \; \sup_{\| t \|_{\infty} \leq T} \sup_{l=1,\dots,d}  \left| \partial_{t_l}^2 \F[F](t) - \partial_{t_l}^2  \F[F'] (t)  \right| \leq D \epsilon.
\end{align*}
Furthermore, the support points of $F'$ can be chosen on the grid $ \mathcal{Z} = \{ T^{-1} \epsilon   (z_1,\dots,z_d) : z_i \in \mathbb{Z} \; , \; \left| z_i \right| \leq \lceil  R / ( T^{-1} \epsilon  )   \rceil     \} $, with a multiplictive penalty of at most $R$ and $R^2$ in cases $(ii)$ and $(iii)$, respectively.

\end{lemma}

\begin{proof}[Proof of Lemma \ref{aux5}]
Given any $k \geq 1$, by \citep[Lemma A.1]{ghosal2001entropies}, there exists a discrete measure $F'$ with at most $(k+1)^d + 1$ support points on $[-R,R]^d$ such that \begin{align*}
    \int_{\R^d}  z_1^{l_1} \dots z_d^{l_d} dF(z) =  \int_{\R^d}  z_1^{l_1} \dots z_d^{l_d} dF'(z) \; \; \; \; \; \; \; \; \forall \; \; \;  0  \leq l_1, \dots , l_d \leq k.
\end{align*}
For any $t,z \in \R^d$ and $j \in \mathbb{Z}$, a multinomial expansion yields $$ (t'z)^j = \sum_{k_1 + \dots +  k_d = j, k_1 \geq 0 , \dots , k_d \geq 0} \frac{j!}{k_1! \dots k_d!} \prod_{i=1}^d t_i^{k_i} z_i^{k_i}.  $$
It follows that $F$ and $F'$ assign the same expectation to $(t'z)^j$ for every $t \in \R^d$, provided that $j \leq k$. This yields
 \begin{align*}
      \F[F](t) -   \F[F'] (t)   =      \int_{\R^d}  e^{\mathbf{i} t' z} d(F-F')(z)     & = \int_{\R^d} \sum_{j=0}^{\infty} \frac{(\mathbf{i} t'z)^j}{j!}    d(F-F')(z)    \\ & = \int_{\R^d}   \sum_{j=k+1}^{\infty}   \frac{(\mathbf{i} t'z)^j}{j!}     d(F-F')(z)  .
\end{align*}
Observe that for every $x \in \R$, we have the bound $ \left| \sum_{j=k+1}^{\infty} (\mathbf{i}x)^j / j!  \right| =    \left|  e^{\mathbf{i}x} - \sum_{j=0}^{k} (\mathbf{i}x)^j / j!   \right| \leq \left|x \right|^{k+1} / (k+1)! \leq \left|x \right|^{k+1} e^{k+1} / (k+1)^{k+1} $. For any $z \in [-R,R]^d$, we have that $ \left| t'z \right| \leq  \| t \| \| z \| \leq \| t \| \sqrt{d} R$. Since $F,F'$ are probability measures and have support contained in $[-R,R]^d$, it follows that \begin{align*}
     \left|  \F[F](t) - \F[F'] (t)  \right|   &  \leq \int_{\R^d}   \left| \sum_{j=k+1}^{\infty}    \frac{(\mathbf{i} t'z)^j}{j!}  \right|    dF + \int_{\R^d} \left| \sum_{j=k+1}^{\infty}    \frac{(\mathbf{i} t'z)^j}{j!}  \right|    dF' \\ & \leq   2 \frac{\| t \|^{k+1}  (e \sqrt{d} R)^{k+1}  }{(k+1)^{k+1}}.
\end{align*}
For the second claim, the same reasoning as above implies that $F$ and $F'$ assign the same expectation to the vector $z (t'z)^j$ for every $t \in \R^d$, provided that $j \leq k-1$. This yields
\begin{align*}
     \nabla \F[F](t) - \nabla \F[F'] (t)   =       \mathbf{i} \int_{\R^d}  z e^{\mathbf{i} t' z} d(F-F')(z)     & = \mathbf{i} \int_{\R^d} z \sum_{j=0}^{\infty}  \frac{(\mathbf{i} t'z)^j}{j!}    d(F-F')(z)    \\ & =\mathbf{i} \int_{\R^d}  z  \sum_{j=k}^{\infty}   \frac{(\mathbf{i} t'z)^j}{j!}     d(F-F')(z)  .
\end{align*}
We have $ \| z \| \leq \sqrt{d} R  $ for every $z$ in the support of $F,F'$. From using the bound in the preceding case with $k$ replacing $k+1$, it follows that \begin{align*}
    \| \nabla \F[F](t) - \nabla \F[F'] (t) \|  & \leq  \sqrt{d}R  \bigg(  \int_{\R^d}   \left| \sum_{j=k}^{\infty}    \frac{(\mathbf{i} t'z)^j}{j!}  \right|    dF + \int_{\R^d} \left| \sum_{j=k}^{\infty}    \frac{(\mathbf{i} t'z)^j}{j!}  \right|    dF'   \bigg) \\ & \leq  2 \sqrt{d}R \frac{\| t \|^{k}  (e \sqrt{d} R)^{k}  }{k^{k}}.
\end{align*}

For the third claim, let $z_l$ denote the $l^{th}$ coordinate of $z = (z_1,\dots,z_d)$. The same reasoning as above implies that $F$ and $F'$ assign the same expectation to the vector $z_l^2 (t'z)^j$ for every $t \in \R^d$, provided that $j \leq k-2$. This yields \begin{align*}
     \partial_{t_l}^2  \F[F](t) - \partial_{t_l}^2  \F[F'] (t)   =       - \int_{\R^d}  z_l^2 e^{\mathbf{i} t' z} d(F-F')(z)     & = - \int_{\R^d} z_l^2 \sum_{j=0}^{\infty}  \frac{(\mathbf{i} t'z)^j}{j!}    d(F-F')(z)    \\ & = - \int_{\R^d}  z_l^2  \sum_{j=k-1}^{\infty}   \frac{(\mathbf{i} t'z)^j}{j!}     d(F-F')(z)  .
\end{align*}
Since $ z_l^2 \leq \| z  \|^2 \leq d R^2  $, the claim follows from an analogous bound to the preceding case.

For the final claim regarding the location of the support points, suppose $F' = \sum_{i=1}^N p_i \delta_{\mu_i}$ is a discrete probability measure that satisfies the requirements of the first part of the Lemma. Let $F^* = \sum_{i=1}^N p_i \delta_{\mu_i^*} $ denote the probability measure obtained by replacing each $\mu_i$ with $\mu_i^* \in \arg \min \limits_{t \in \mathcal{Z}} \| \mu_i - t   \| $. From the definition of $\mathcal{Z}$, it follows that $  \| \mu_i  - \mu_i^* \| \leq D T^{-1}   \epsilon $. We claim $F'$ satisfies all the same bounds. For the first bound, observe that \begin{align*}
   \sup_{\| t \|_{\infty} \leq T} \left| \mathcal{F}[F'] - \mathcal{F}[F^*]   \right|   & =  \sup_{\| t \|_{\infty} \leq T} \left|  \sum_{j=1}^N p_j \big[  e^{\mathbf{i}t \mu_j} - e^{\mathbf{i}t \mu_j^*}   ]      \right| \leq  \sup_{\| t \|_{\infty} \leq T}  \sup_{j=1,\dots,N} \left|    e^{\mathbf{i}t \mu_j} - e^{\mathbf{i}t \mu_j^*}   \right|.
\end{align*}
Since the mapping $\mu \rightarrow e^{\mathbf{i}t \mu}$ has Lipschitz constant bounded by $ \| t \|$, it follows that \begin{align*}
    \sup_{\| t \|_{\infty} \leq T} \left| \mathcal{F}[F'] - \mathcal{F}[F^*]   \right|  \leq \sqrt{d} T \sup_{j=1,\dots,N}  \| \mu_j - \mu_j^*  \| \leq D \epsilon.
\end{align*}
For the second bound, observe that for $\mu \in [-R ,R]^d$, the mapping $\mu \rightarrow \mu e^{\mathbf{i}t \mu}$ has Lipschitz constant at most $1 + \sqrt{d} R \| t \| $. Hence
\begin{align*}
     \sup_{\| t \|_{\infty} \leq T}  \| \nabla \F[F'](t) - \nabla \F[F^*] (t) \| & \leq \sup_{\| t \|_{\infty} \leq T}  \sup_{j=1,\dots,N} \| \mu_j e^{\mathbf{i}t' \mu_j} - \mu_j^* e^{\mathbf{i}t' \mu_j^*}   \| \\ &  \leq D (RT+1) \sup_{j=1,\dots,N}  \| \mu_j - \mu_j^*  \| \\ & \leq D R \epsilon.
\end{align*}
For the third bound, let $\mu_{j,l}$ denote the $l^{th}$ component of $\mu_j$. On $[-R,R]^d$, the mapping  $\mu_j \rightarrow \mu_{j,l}^2 e^{\mathbf{i}t' \mu_j}$ has Lipschitz constant at most $ 2 \sqrt{d} R + d R^2 \| t \|$. Hence \begin{align*}
    \sup_{\| t \|_{\infty} \leq T} \sup_{l=1,\dots,d}  \left| \partial_{t_l}^2 \F[F](t) - \partial_{t_l}^2  \F[F'] (t)  \right| & \leq  \sup_{\| t \|_{\infty} \leq T} \sup_{l=1,\dots,d} \sup_{j=1,\dots,N} \left|  \mu_{j,l}^2 e^{\mathbf{i}t' \mu_j} - (\mu_{j,l}^*)^2 e^{\mathbf{i}t' \mu_j^*}    \right| \\ & \leq D (R + R^2 T) \sup_{j=1,\dots,N}  \| \mu_j - \mu_j^*  \| \\ & \leq D R^2 \epsilon.
\end{align*}

\end{proof}

\begin{lemma}
\label{aux6}
Suppose $F$ is a probability measure supported on $[-L,L]^d$ for some $ L  > 0$ and $\Sigma \in \R^{d \times d}$ is a positive-definite matrix with smallest eigenvalue $\underline{\sigma}^2  > 0 $. Then, for all $\epsilon \in (0,1)$, there exists a discrete probability measure $F'$  with at most  $ D  \max \big \{ ( \log(\epsilon^{-1}))^d, (L / \underline{\sigma})^d ( \log(\epsilon^{-1}))^{d/2}   \big \}  $ support points on $[-L,L]^d$  such that  $\| \varphi_{F, \Sigma} - \varphi_{F',\Sigma}    \|_{L^2} \leq D' \underline{\sigma}^{-d/2} \epsilon  $, where $D,D' > 0$ are universal constants. Furthermore, the support points can be chosen such that $\inf_{i \neq j} \| \mu_i - \mu_j \| \geq \underline{\sigma} \epsilon$.
\end{lemma}

\begin{proof}[Proof of Lemma \ref{aux6}]
Let $D$ denote a generic universal constant that may change from line to line. By Lemma \ref{aux5}, there exists a discrete measure $F'$ with at most $k^d + 1$ support points on $[-L,L]^d$ such that \begin{align*}
    \left|  \F[F](t) - \F[F'] (t)  \right| \leq  2 \frac{\| t \|^k  (e \sqrt{d} L)^k  }{k^k} \; \; \; \; \forall \; t \in \R^d.
\end{align*}
Observe that $  \left| \varphi_{\Sigma} \right|  = e^{- t' \Sigma t /2}  \leq e^{- \| t \|^2 \underline{\sigma}^2/2}$. From using the preceding bound, we obtain for every $ M > 0$, the estimate \begin{align*}
    \|   \varphi_{F,\Sigma}  - \varphi_{F',\Sigma}   \|_{L^2}^2 & = \int_{\R^d} \left| \varphi_{\Sigma}    (t)        \right|^2   \left|  \F[F](t) - \F[F'](t)    \right|^2 dt \\ & \leq D  \bigg[ \int_{\| t \|> M} \left| \varphi_{\Sigma}    (t)        \right|^2  dt +  \bigg( \frac{L e \sqrt{d}}{k}   \bigg)^{2k}  \int_{\| t \| \leq  M} \| t\|^{2k} dt\bigg] \\ & \leq   D \bigg[    \int_{\| t \|> M} e^{- \| t \|^2 \underline{\sigma}^2 } dt  +  \bigg( \frac{L e \sqrt{d}}{k}   \bigg)^{2k}  \int_{\| t \| \leq  M} \| t\|^{2k} dt\bigg].
\end{align*}
By change of variables (to spherical coordinates) the first integral scales at rate  $D M^{d-2} e^{-M^2 \underline{\sigma}^2 }  \underline{\sigma}^{-2} $  and the second integral scales at rate $D M^{2k+d}/ (2k+d) $. The choice $M = \sqrt{4 \log (\epsilon^{-1})} / \underline{\sigma}   $ leads to $$ \int_{\| t \|> M} e^{- \| t \|^2 \underline{\sigma}^2 } dt \leq D M^{d-2} e^{-M^2 \underline{\sigma}^2 } \underline{\sigma}^{-2}  \leq D \epsilon^2 \underline{\sigma}^{-d}.     $$
Furthermore, if $k \geq  e^3 \sqrt{d}  \max \{ L M , \log(\epsilon^{-1})  \} $, we have that \begin{align*}
    \bigg( \frac{L e \sqrt{d}}{k}   \bigg)^{2k}  \int_{\| t \| \leq  M} \| t\|^{2k} dt \leq D \bigg( \frac{L e \sqrt{d} M}{k}   \bigg)^{2k} \frac{M^d}{2k+d} \leq D \epsilon^2.
\end{align*}
The claim follows from observing that the number of support points in $F'$ is $ N =    k^d +1  $.

For the final claim regarding the separation of the support points, suppose $F' = \sum_{i=1}^N p_i \delta_{\mu_i}$ is a discrete probability measure that satisfies the requirements of the first part of the Lemma. Let $\mathcal{Z}$ denote a maximal $ \underline{\sigma} \epsilon $ separated subset of $[-L,L]^d$. For each $\mu_i$, select $\mu_i^* \in \arg \min \limits_{t \in \mathcal{Z}} \| \mu_i - t   \| $ and let $F^* = \sum_{i=1}^N p_i \delta_{\mu_i^*} $. From the definition of $\mathcal{Z}$, it follows that $ \sup_{i=1}^N  \| \mu_i  - \mu_i^* \| \leq   \underline{\sigma} \epsilon $. Observe that \begin{align*}
     \|   \varphi_{F',\Sigma} - \varphi_{F^*,\Sigma}  \|_{L^2} & =    \bigg \| \sum_{j=1}^N p_j  \big [ e^{i t \mu_j} -  e^{i t \mu_j^*} \big] e^{-t' \Sigma t /2}             \bigg \|_{L^2} \\ & \leq  D \sum_{j=1}^N p_j \| \big ( e^{i t \mu_j} -  e^{i t \mu_j^*} \big) e^{-t' \Sigma t /2}   \|_{L^2} \\ & \leq D \sup_{j=1,\dots,N } \| \big ( e^{i t \mu_j} -  e^{i t \mu_j^*} \big) e^{-t' \Sigma t /2}   \|_{L^2}.
\end{align*}
Since the mapping $\mu \rightarrow e^{it \mu}$ has Lipschitz constant bounded by $ \| t \|$, we obtain \begin{align*}
    \| \big ( e^{i t \mu_j} -  e^{i t \mu_j^*} \big) e^{-t' \Sigma t /2}   \|_{L^2}^2 &  = \int_{\R^d} \left|  e^{i t \mu_j} -  e^{i t \mu_j^*} \right|^2 e^{- t' \Sigma t} dt \\ & \leq D \| \mu_j - \mu_j^* \|^2 \int_{\R^d}  \| t \|^2 e^{-t' \Sigma t} dt \\ & \leq D \underline{\sigma}^2 \epsilon^2 \int_{\R^d}  \| t \|^2 e^{-t' \Sigma t} dt .
\end{align*}
Since $e^{-t' \Sigma t}  \leq  e^{ - \| t \|^2 \underline{\sigma}^2 }$, the integral on the right scales with rate at most $ D \underline{\sigma}^{-(d+2)}  $. It follows that \begin{align*}
    \| \big ( e^{i t \mu_j} -  e^{i t \mu_j^*} \big) e^{-t' \Sigma t /2}   \|_{L^2} \leq D \underline{\sigma} \epsilon  \underline{\sigma}^{-(d+2)/2} = D  \underline{\sigma}^{-d/2} \epsilon.
    \end{align*}
In particular, we obtain $$ \| \varphi_{F, \Sigma} - \varphi_{F^*,\Sigma}    \|_{L^2} \leq \| \varphi_{F, \Sigma} - \varphi_{F',\Sigma}    \|_{L^2} + \| \varphi_{F^*, \Sigma} - \varphi_{F',\Sigma}    \|_{L^2} \leq D \underline{\sigma}^{-d/2} \epsilon. $$
Finally, note that if $\mu_i^* = \mu_j^*$ for some $i \neq j$, $F^*$ can be reduced to a discrete measure with $N^* \leq N$ unique support points.
    
\end{proof}

\begin{lemma}
\label{aux7}
Fix any positive definite matrix $\Sigma_0 \in \R^{d \times d}$ and denote by $ \sigma_0^2$, the smallest eigenvalue of $\Sigma_0$. Then, there exists a universal constant $D > 0$ (only depending on $d$) such that for any distribution $P$ and positive definite matrix $\Sigma$ satisfying $\| \Sigma - \Sigma_0 \| \leq \sigma_0^2/2 $, we have \begin{align*}
    \| \varphi_{P,\Sigma} - \varphi_{P,\Sigma_0}  \|_{L^2} \leq D \sigma_0^{-d/2 -1}  \| \Sigma - \Sigma_0 \|.
\end{align*}
\end{lemma}

\begin{proof}[Proof of Lemma \ref{aux7}]
For any distribution $P$ and positive definite matrix $\Sigma$, we have that \begin{align*}
    \| \varphi_{P,\Sigma} - \varphi_{P,\Sigma_0}   \|_{L^2}^2  =  \int_{\R^d} \left| \mathcal{F}[P](t)   \right|^2 \left| e^{- t' \Sigma t /2}    -  e^{-t' \Sigma_0 t} \right|^2 dt &  \leq  \int_{\R^d}  \left| e^{- t' \Sigma t /2}    -  e^{- t' \Sigma_0 t/2} \right|^2 dt \\ & = \int_{\R^d}   e^{- t' \Sigma_0 t }   \left| 1 - e^{ t' ( \Sigma_0 - \Sigma)t/2 }   \right|^2 dt.
\end{align*}
The mapping $t \rightarrow  e^{ t' ( \Sigma_0 - \Sigma)t/2 }$ has gradient norm at most $ e^{\| t \|^2 \| \Sigma - \Sigma_0  \|/2  } \| \Sigma - \Sigma_0 \| \| t \| $. If the bound $  \| \Sigma - \Sigma_0  \| \leq \sigma_0^2/2 $ holds, it follows that \begin{align*}
    \int_{\R^d}   e^{- t' \Sigma_0 t }   \left| 1 - e^{ t' ( \Sigma_0 - \Sigma)t/2 }   \right|^2 dt & \leq  \int_{\R^d}   e^{- \| t\|^2 \sigma_0^2 }   \left| 1 - e^{ t' ( \Sigma_0 - \Sigma)t/2 }   \right|^2 dt \\ & \leq \| \Sigma - \Sigma_0  \|^2   \int_{\R^d}   e^{- \| t\|^2 \sigma_0^2/2 } \| t \|^2   dt .
\end{align*}
The claim follows from observing that the integral on the right scales with rate at most $D \sigma_0^{-(d+2)}$.

\end{proof}

\begin{lemma}
\label{aux10}
Suppose $F$ is a probability measure supported on $[-L,L]^d$ for some $ L  > 0$ and $\Sigma \in \R^{d \times d}$ is a positive semi-definite matrix. Then, for all $\epsilon \in (0,1)$, there exists a discrete probability measure $F'$  with at most  $ D  \max \big \{ ( \log(\epsilon^{-1}))^d, L^d T^d  \big \}  $ support points on $[-L,L]^d$  such that  $\| \varphi_{F, \Sigma} - \varphi_{F',\Sigma}    \|_{ \mathbb{B}(T) } \leq D'  \epsilon  $, where $D,D' > 0$ are universal constants (that do not depend on $\Sigma$). Furthermore, the support points can be chosen such that $\inf_{i \neq j} \| \mu_i - \mu_j \| \geq T^{-(d+2)/2} \epsilon $.
\end{lemma}

\begin{proof}[Proof of Lemma \ref{aux10}]
Let $D$ denote a generic universal constant that may change from line to line. By Lemma \ref{aux5}, there exists a discrete measure $F'$ with at most $k^d + 1$ support points on $[-L,L]^d$ such that \begin{align*}
    \left|  \F[F](t) - \F[F'] (t)  \right| \leq  2 \frac{\| t \|^k  (e \sqrt{d} L)^k  }{k^k} \; \; \; \; \forall \; t \in \R^d.
\end{align*}
 From using the preceding bound and noting that the eigenvalues of $\Sigma$ are non-negative, we obtain the estimate
 \begin{align*}
     \|   \varphi_{F,\Sigma}  - \varphi_{F',\Sigma}   \|_{ \mathbb{B}(T) }^2 & \leq \int_{\| t \|_{\infty} \leq T} \left| \varphi_{\Sigma}    (t)        \right|^2   \left|  \F[F](t) - \F[F'](t)    \right|^2 dt  \\ & \leq 2  \bigg( \frac{L e \sqrt{d}}{k} \bigg)^{2k}  \int_{\| t \|_{\infty} \leq T}  e^{- t' \Sigma t  } \| t \|^{2k} dt \\ & \leq 2  \bigg( \frac{L e \sqrt{d}}{k} \bigg)^{2k}  \int_{\| t \|_{\infty} \leq T} \| t \|^{2k} dt \\ & \leq D  \bigg( \frac{L e \sqrt{d} T }{k} \bigg)^{2k}  T^d. 
 \end{align*}
The quantity on the right is bounded above by $D \epsilon^2 $ if  $k \geq  e^3 \sqrt{d}  \max \{ L T , \log(\epsilon^{-1})  \} $. The claim follows from observing that the number of support points in $F'$ is $ N =    k^d +1  $.

For the final claim regarding the separation of the support points, suppose $F' = \sum_{i=1}^N p_i \delta_{\mu_i}$ is a discrete probability measure that satisfies the requirements of the first part of the Lemma. Let $\mathcal{Z}$ denote a maximal $ T^{-(d+2)/2} \epsilon $ separated subset of $[-L,L]^d$. For each $\mu_i$, select $\mu_i^* \in \arg \min \limits_{t \in \mathcal{Z}} \| \mu_i - t   \| $ and let $F^* = \sum_{i=1}^N p_i \delta_{\mu_i^*} $. From the definition of $\mathcal{Z}$, it follows that $ \sup_{i=1}^N  \| \mu_i  - \mu_i^* \| \leq   \underline{\sigma} \epsilon $. Observe that \begin{align*}
     \|   \varphi_{F',\Sigma} - \varphi_{F^*,\Sigma}  \|_{ \mathbb{B}(T)} & =    \bigg \| \sum_{j=1}^N p_j  \big [ e^{\mathbf{i} t \mu_j} -  e^{\mathbf{i} t \mu_j^*} \big] e^{-t' \Sigma t /2}             \bigg \|_{ \mathbb{B}(T)} \\ & \leq  D \sum_{j=1}^N p_j \| \big ( e^{\mathbf{i} t \mu_j} -  e^{\mathbf{i} t \mu_j^*} \big) e^{-t' \Sigma t /2}   \|_{ \mathbb{B}(T)} \\ & \leq D \sup_{j=1,\dots,N } \| \big ( e^{\mathbf{i} t \mu_j} -  e^{\mathbf{i} t \mu_j^*} \big) e^{-t' \Sigma t /2}   \|_{ \mathbb{B}(T)}.
\end{align*}
Since the mapping $\mu \rightarrow e^{\mathbf{i}t \mu}$ has Lipschitz constant bounded by $ \| t \|$, we obtain \begin{align*}
    \| \big ( e^{\mathbf{i} t \mu_j} -  e^{\mathbf{i} t \mu_j^*} \big) e^{-t' \Sigma t /2}   \|_{ \mathbb{B}(T)}^2 &  \leq \int_{\|t \|_{\infty} \leq T} \left|  e^{\mathbf{i} t \mu_j} -  e^{\mathbf{i}t \mu_j^*} \right|^2 e^{- t' \Sigma t} dt \\ & \leq D \| \mu_j - \mu_j^* \|^2 \int_{\|t \|_{\infty} \leq T}  \| t \|^2 e^{-t' \Sigma t} dt \\ & \leq D   T^{-(d+2)} \epsilon^2 \int_{\|t \|_{\infty} \leq T}  \| t \|^2  dt.
\end{align*}
The  integral on the right scales with rate at most $ D  T^{d+2}  $. It follows that $$ \| \varphi_{F, \Sigma} - \varphi_{F^*,\Sigma}    \|_{\mathbb{B}(T)} \leq \| \varphi_{F, \Sigma} - \varphi_{F',\Sigma}    \|_{\mathbb{B}(T)} + \| \varphi_{F^*, \Sigma} - \varphi_{F',\Sigma}    \|_{\mathbb{B}(T)} \leq D \epsilon. $$
Finally, note that if $\mu_i^* = \mu_j^*$ for some $i \neq j$, $F^*$ can be reduced to a discrete measure with $N^* \leq N$ unique support points.
\end{proof}

\begin{lemma}
\label{aux11}
Fix any positive definite matrix $\Sigma_0 \in \R^{d \times d}$ and denote by $ \sigma_0^2$, the smallest eigenvalue of $\Sigma_0$. Then, there exists a universal constant $D > 0$ (only depending on $d$) such that for any distribution $P$ and positive definite matrix $\Sigma$ satisfying $\| \Sigma - \Sigma_0 \| \leq \sigma_0^2/2 $, we have \begin{align*}
    \| \varphi_{P,\Sigma} - \varphi_{P,\Sigma_0}  \|_{ \mathbb{B}(T) } \leq D T^{(d+2)/2}   \| \Sigma - \Sigma_0 \|.
\end{align*}
\end{lemma}

\begin{proof}[Proof of Lemma \ref{aux11}]
For any distribution $P$ and positive definite matrix $\Sigma$, we have that \begin{align*}
    \| \varphi_{P,\Sigma} - \varphi_{P,\Sigma_0}   \|_{\mathbb{B}(T)}^2  \leq  \int_{\|t \|_{\infty} \leq T} \left| \mathcal{F}[P](t)   \right|^2 \left| e^{- t' \Sigma t /2}    -  e^{-t' \Sigma_0 t} \right|^2 dt &  \leq  \int_{\|t \|_{\infty} \leq T}  \left| e^{- t' \Sigma t /2}    -  e^{- t' \Sigma_0 t/2} \right|^2 dt \\ & = \int_{\|t \|_{\infty} \leq T}   e^{- t' \Sigma_0 t }   \left| 1 - e^{ t' ( \Sigma_0 - \Sigma)t/2 }   \right|^2 dt.
\end{align*}
The mapping $t \rightarrow  e^{ t' ( \Sigma_0 - \Sigma)t/2 }$ has gradient norm at most $ e^{\| t \|^2 \| \Sigma - \Sigma_0  \|/2  } \| \Sigma - \Sigma_0 \| \| t \| $. If the bound $  \| \Sigma - \Sigma_0  \| \leq \sigma_0^2/2 $ holds, it follows that \begin{align*}
    \int_{\|t \|_{\infty} \leq T}   e^{- t' \Sigma_0 t }   \left| 1 - e^{ t' ( \Sigma_0 - \Sigma)t/2 }   \right|^2 dt & \leq  \int_{\|t \|_{\infty} \leq T}   e^{- \| t\|^2 \sigma_0^2 }   \left| 1 - e^{ t' ( \Sigma_0 - \Sigma)t/2 }   \right|^2 dt \\ & \leq \| \Sigma - \Sigma_0  \|^2   \int_{\|t \|_{\infty} \leq T}   e^{- \| t\|^2 \sigma_0^2/2 } \| t \|^2   dt \\ & \leq  D \| \Sigma - \Sigma_0  \|^2 T^{d+2}.
\end{align*}
\end{proof}
\begin{lemma}
\label{aux12}
Suppose $F \sim \text{DP}_{\alpha} $ where the base measure $\alpha$ is a Gaussian measure on $\R^d$. Fix any $q \in \mathbb{N}$. Then, there exists a universal constant $C > 0$ such that for any sequence $u_n \uparrow \infty$,  \begin{align*}
    \mathbb{P} \bigg(  \bigg \{  \int \|x\|^q dF(x) \bigg \}^{1/q} > u_n    \bigg) \leq 2 e^{- C u_n^2}
\end{align*}
holds for all sufficiently large $n$.
\end{lemma}

\begin{proof}[Proof of Lemma \ref{aux12}]
Suppose $\alpha = N(\mu,\Sigma)$ for some mean vector $\mu \in \R^d$ and positive definite covariance matrix $\Sigma \in \mathbf{S}_+^d$. Without loss of generality, it suffices to verify the result with  $\mu=0$. By the stick breaking representation of $\text{DP}_{\alpha}$ we can write $$  F \stackrel{d}{=} \sum_{i=1}^{\infty} p_i \delta_{Z_i} \; \; , \; \; \int \|x\|^q dF(x) \stackrel{d}{=} \sum_{i=1}^{\infty} p_i \| Z_i \|^q, $$
where $(Z_1,Z_2,\dots) \stackrel{i.i.d}{\sim} \alpha $ is independent of $(p_1,p_2,\dots)$ and $(p_i)_{i=1}^{\infty}$ are non-negative random variables with $\sum_{i=1}^{\infty} p_i =1$. For an element $w = (w_1,w_2,\dots)$ with $w_i \in \R^d$, the $\ell^q(\R^d)$ norm is given by $ \| w \|^q = \sum_{i=1}^{\infty} \| w _i \|^q   $. Conditional on $(p_i)_{i=1}^{\infty}$, let $ \Z = ( p_i^{1/q} Z_1  , p_2^{1/q} Z_2  , \dots )$. Define $\gamma = (\E \| Z_1 \|^q )^{1/q}$ and note that $ \E \| \Z \|^q = \sum_{i=1}^{\infty} p_i \E \| Z_i \|^q   = \gamma^q  < \infty$. In particular, we can view $\Z$ (defined conditional on $(p_i)_{i=1}^{\infty}$) as a mean-zero Gaussian random element on $\ell^q(\R^d)$. From an application of \citep[Theorem 2.1.20]{gine2021mathematical}, it follows that \begin{align*}
    \mathbb{P} \bigg(  \| \Z \| > u + \gamma     \bigg) \leq 2 e^{- \frac{u^2}{2 \gamma^2}} \; \; \; \; \forall \; u > 0.
\end{align*}
The bound holds conditionally. However, as the term on the right is independent of $(p_i)_{i=1}^{\infty}$, it also holds unconditionally and we obtain \begin{align*}
    \mathbb{P} \bigg( \bigg \{  \int \|x\|^q dF(x) \bigg \}^{1/q} > u + \gamma  \bigg) \leq 2 e^{- \frac{u^2}{2 \gamma^2}} \; \; \; \; \forall \; u > 0.
\end{align*}
The claim follows by letting $u = u_n$ and noting that $u_n > 2 \gamma$ for all sufficiently large $n$.

\end{proof}

\begin{proof}[Proof of Theorem \ref{main-contract}]
Let $D > 0$  denote a generic universal constant that may change from line to line. 

\begin{enumerate}
\item[\textbf{$(i)$}]First, we derive a lower bound for the normalizing constant of the posterior. We aim to show there exists $C,C' > 0 $ such that \begin{align}
    \label{t1lb} \int \exp \bigg( - \frac{n}{2} \| \mathcal{G}(\mathbb{P}_{n,Y}, \varphi_{F})    \|_{\mathbb{B}(T_n)}^2 \bigg) d \nu(F) \geq C \exp \big(   - C' n    \epsilon_{n}^2   \big)  
\end{align}
holds with $\mathbb{P}$ probability approaching $1$.

Let $ \mathcal{S}_n$ be as in Assumption \ref{sampling-uncert}. By Assumption \ref{sampling-uncert}, we have \begin{align*}
     & \int \exp \bigg( - \frac{n}{2} \| \mathcal{G}(\mathbb{P}_{n,Y}, \varphi_{F})    \|_{\mathbb{B}(T_n)}^2 \bigg) d \nu(F)  \\ & \geq \int_{\mathcal{S}_n} \exp \bigg( - \frac{n}{2} \| \mathcal{G}(\mathbb{P}_{n,Y}, \varphi_{F})    \|_{\mathbb{B}(T_n)}^2 \bigg) d \nu(F) \\ & \geq \exp(-n D \epsilon_n^2) \int_{\mathcal{S}_n} \exp \bigg( - nD \| \mathcal{G}(\mathbb{P}_{Y}, \varphi_{F})    \|_{\mathbb{B}(T_n)}^2 \bigg) d \nu(F)
\end{align*}
with $\mathbb{P}$ probability approaching $1$. 

Let $F_n$ be as in Assumption \ref{weak-bias}. Since $\mathcal{G}(\mathbb{P}_{Y}, \varphi_{X}) = \mathbf{0}$, we have \begin{align*}
    \| \mathcal{G}(\mathbb{P}_{Y}, \varphi_{F})    \|_{\mathbb{B}(T_n)} & = \| \mathcal{G}(\mathbb{P}_{Y}, \varphi_{F}) -  \mathcal{G}(\mathbb{P}_{Y}, \varphi_{F_{n}}) + \mathcal{G}(\mathbb{P}_{Y}, \varphi_{F_{n}})   \|_{\mathbb{B}(T_n)}   \\ & \leq \| \mathcal{G}(\mathbb{P}_{Y}, \varphi_{F}) -  \mathcal{G}(\mathbb{P}_{Y}, \varphi_{F_{n}}) \|_{\mathbb{B}(T_n)} + \| \mathcal{G}(\mathbb{P}_{Y}, \varphi_{F_{n}}) - \mathcal{G}(\mathbb{P}_{Y}, \varphi_{X}) \|_{\mathbb{B}(T_n)}.
\end{align*}
for any probability distribution $F$. By Assumption \ref{weak-bias}, it follows that \begin{align*}
    & \int_{\mathcal{S}_n} \exp \bigg( - nD \| \mathcal{G}(\mathbb{P}_{Y}, \varphi_{F})    \|_{\mathbb{B}(T_n)}^2 \bigg) d \nu(F) \\ & \geq \exp(-n D \epsilon_n^2) \int_{\mathcal{S}_n} \exp \bigg( - n D \| \mathcal{G}(\mathbb{P}_{Y}, \varphi_{F})  - \mathcal{G}(\mathbb{P}_{Y}, \varphi_{F_n})    \|_{\mathbb{B}(T_n)}^2   \bigg) d \nu(F)
\end{align*}
Let $\mathcal{R}_n \supseteq \mathcal{S}_n $ be as in Assumption \ref{loc-conc}. Since $\mathcal{R}_n = \mathcal{S}_n \cup (\mathcal{R}_n \setminus \mathcal{S}_n)  $, we have \begin{align*}
    & \int_{\mathcal{S}_n} \exp \bigg( - n D \| \mathcal{G}(\mathbb{P}_{Y}, \varphi_{F})  - \mathcal{G}(\mathbb{P}_{Y}, \varphi_{F_n})    \|_{\mathbb{B}(T_n)}^2   \bigg) d \nu(F) \\ & = \int_{\mathcal{R}_n} \exp \bigg( - n D \| \mathcal{G}(\mathbb{P}_{Y}, \varphi_{F})   - \mathcal{G}(\mathbb{P}_{Y}, \varphi_{F_n})    \|_{\mathbb{B}(T_n)}^2   \bigg) d \nu(F) \\ & - \int_{\mathcal{R}_n \setminus \mathcal{S}_n } \exp \bigg( - n D \| \mathcal{G}(\mathbb{P}_{Y}, \varphi_{F})  - \mathcal{G}(\mathbb{P}_{Y}, \varphi_{F_n})    \|_{\mathbb{B}(T_n)}^2   \bigg) d \nu(F).
\end{align*}
By Assumption \ref{loc-conc} $(ii)$, we obtain \begin{align*}
     \int_{\mathcal{R}_n \setminus \mathcal{S}_n } \exp \bigg( - n D \| \mathcal{G}(\mathbb{P}_{Y}, \varphi_{F})  - \mathcal{G}(\mathbb{P}_{Y}, \varphi_{F_n})    \|_{\mathbb{B}(T_n)}^2   \bigg) d \nu(F) & \leq  \int_{\mathcal{R}_n \setminus \mathcal{S}_n } d \nu(F) \\ & \leq D \exp(- B_n)
\end{align*}
for some sequence $B_n \uparrow \infty$ with $n \epsilon_n^2 = o(B_n)$. It follows that \begin{align*}
    & \int_{\mathcal{S}_n} \exp \bigg( - n D \| \mathcal{G}(\mathbb{P}_{Y}, \varphi_{F})  - \mathcal{G}(\mathbb{P}_{Y}, \varphi_{F_n})    \|_{\mathbb{B}(T_n)}^2   \bigg) d \nu(F) \\ & \geq \int_{\mathcal{R}_n} \exp \bigg( - n D \| \mathcal{G}(\mathbb{P}_{Y}, \varphi_{F})   - \mathcal{G}(\mathbb{P}_{Y}, \varphi_{F_n})    \|_{\mathbb{B}(T_n)}^2   \bigg) d \nu(F) - D \exp(- B_n).
\end{align*}
By Assumption \ref{loc-conc} $(i)$, we have that \begin{align*}
    & \int_{\mathcal{R}_n} \exp \bigg( - n D \| \mathcal{G}(\mathbb{P}_{Y}, \varphi_{F})   - \mathcal{G}(\mathbb{P}_{Y}, \varphi_{F_n})    \|_{\mathbb{B}(T_n)}^2   \bigg) d \nu(F) \\ & \geq \int_{F \in \mathcal{R}_n, \| \mathcal{G}(\mathbb{P}_{Y}, \varphi_{F})   - \mathcal{G}(\mathbb{P}_{Y}, \varphi_{F_n})    \|_{\mathbb{B}(T_n)}^2 \leq D^2 \epsilon_n^2  } \exp \bigg( - n D \| \mathcal{G}(\mathbb{P}_{Y}, \varphi_{F})   - \mathcal{G}(\mathbb{P}_{Y}, \varphi_{F_n})    \|_{\mathbb{B}(T_n)}^2   \bigg) d \nu(F) \\ & \geq \exp(-n D \epsilon_n^2) \int_{F \in \mathcal{R}_n, \| \mathcal{G}(\mathbb{P}_{Y}, \varphi_{F})   - \mathcal{G}(\mathbb{P}_{Y}, \varphi_{F_n})    \|_{\mathbb{B}(T_n)}^2 \leq D^2 \epsilon_n^2  } d \nu(F) \\ & \geq \exp(-n D \epsilon_n^2).
\end{align*}
Since $n \epsilon_n^2 = o(B_n)$, the preceding bounds imply that there exists $c, C' > 0$ such that \begin{align*}
    & \int_{\mathcal{S}_n} \exp \bigg( - n D \| \mathcal{G}(\mathbb{P}_{Y}, \varphi_{F})  - \mathcal{G}(\mathbb{P}_{Y}, \varphi_{F_n})    \|_{\mathbb{B}(T_n)}^2   \bigg) d \nu(F) \\ & \geq \exp(-n D \epsilon_n^2) - D \exp( - B_n) \\ & \geq c \exp(-C' n \epsilon_n^2).
\end{align*}
The lower bound in (\ref{t1lb}) follows from combining all the preceding estimates.

\item[\textbf{$(ii)$}] For any set $\Omega$, the lower bound in part $(i)$ yields \begin{align*}
     \nu \big ( F \in \Omega \: \big| \: \mathcal{D}_n  \big ) & = \frac{ \int_{F \in \Omega}  \exp \big( - \frac{n}{2} \| \mathcal{G}(\mathbb{P}_{n,Y}, \varphi_{F})    \|_{\mathbb{B}(T_n)}^2 \big) d \nu(F) }{\int_{} \exp \big( - \frac{n}{2} \| \mathcal{G}(\mathbb{P}_{n,Y}, \varphi_{F})    \|_{\mathbb{B}(T_n)}^2 \big) d \nu(F) } \\ & \leq D \exp (C' n \epsilon_n^2) \int_{F \in \Omega}  \exp \bigg( - \frac{n}{2} \| \mathcal{G}(\mathbb{P}_{n,Y}, \varphi_{F})    \|_{\mathbb{B}(T_n)}^2 \bigg) d \nu(F)
\end{align*}
with $\mathbb{P}$ probability approaching $1$, for some universal constants $D, C'  > 0$. Fix any $R >  C '$ and define the set \begin{align*}
    \Omega = \{ F : \| \mathcal{G}(\mathbb{P}_{n,Y}, \varphi_{F}) -   \mathcal{G}(\mathbb{P}_{Y}, \varphi_{X})  \|_{\mathbb{B}(T_n)}^2  > 2 R \epsilon_n^2     \}
\end{align*}
Since $\mathcal{G}(\mathbb{P}_{Y}, \varphi_{X}) = \mathbf{0}$, it follows that \begin{align*}
    \nu \big ( F \in \Omega \: \big| \: \mathcal{D}_n  \big ) & \leq D \exp(C' n \epsilon_n^2) \int_{F \in \Omega}  \exp \bigg( - \frac{n}{2} \| \mathcal{G}(\mathbb{P}_{n,Y}, \varphi_{F})    \|_{\mathbb{B}(T_n)}^2 \bigg) d \nu(F) \\ & =  D \exp(C' n \epsilon_n^2) \int \limits_{F : \| \mathcal{G}(\mathbb{P}_{n,Y}, \varphi_{F})    \|_{\mathbb{B}(T_n)}^2 > 2 R \epsilon_n^2 } \exp \bigg( - \frac{n}{2} \| \mathcal{G}(\mathbb{P}_{n,Y}, \varphi_{F})    \|_{\mathbb{B}(T_n)}^2 \bigg) d \nu(F) \\ & \leq D \exp( [C' - R] n \epsilon_n^2) .
\end{align*}
Since $R > C'$ and $n \epsilon_n^2 \uparrow \infty$, the claim follows.
\end{enumerate}
\end{proof}

\begin{proof}[Proof of Corollary \ref{contract-strong}] For any set $\Omega$, the lower bound derived in the proof of part $(i)$ in Theorem \ref{main-contract} yields \begin{align*}
     \nu \big ( F \in \Omega \: \big| \: \mathcal{D}_n  \big ) & = \frac{ \int_{F \in \Omega}  \exp \big( - \frac{n}{2} \| \mathcal{G}(\mathbb{P}_{n,Y}, \varphi_{F})    \|_{\mathbb{B}(T_n)}^2 \big) d \nu(F) }{\int_{} \exp \big( - \frac{n}{2} \| \mathcal{G}(\mathbb{P}_{n,Y}, \varphi_{F})    \|_{\mathbb{B}(T_n)}^2 \big) d \nu(F) } \\ & \leq D \exp (C' n \epsilon_n^2) \int_{F \in \Omega}  \exp \bigg( - \frac{n}{2} \| \mathcal{G}(\mathbb{P}_{n,Y}, \varphi_{F})    \|_{\mathbb{B}(T_n)}^2 \bigg) d \nu(F)
\end{align*}
with $\mathbb{P}$ probability approaching $1$, for some universal constants $D, C'  > 0$. Define   \begin{align*}
    \Omega = \{ F : F \notin \mathcal{H}_n : \| \mathcal{G}(\mathbb{P}_{n,Y}, \varphi_{F})  - \mathcal{G}(\mathbb{P}_{n,Y}, \varphi_{X})    \|_{\mathbb{B}(T_n)} \leq L \epsilon_n  \}.
\end{align*}
If the hypothesis of Corollary \ref{contract-strong} holds for some $D' > C'$, the preceding bound and the conclusion of Theorem \ref{main-contract} yields \begin{align*}
    \nu \bigg ( F \in \mathcal{H}_n : \| \mathcal{G}(\mathbb{P}_{n,Y}, \varphi_{F})  - \mathcal{G}(\mathbb{P}_{Y}, \varphi_{X})    \|_{\mathbb{B}(T_n)}  \leq L \epsilon_n  \: \bigg| \: \mathcal{D}_n  \bigg ) \xrightarrow{\mathbb{P}} 1.
\end{align*}
The claim follows from the definition of the modulus $\omega_n(.)$

\end{proof}

\begin{proof}[Proof of Theorem \ref{deconv-1}] The proof proceeds through several steps which we outline below. Define the sequences \begin{align*}
    \sigma_n = n^{-1/5} T_n^{-1} \; \; , \; \; \epsilon_n^2 = n^{-1}\sigma_n^{-1 }.
\end{align*}
\begin{enumerate}
    \item[\textbf{$(i)$}] First, we aim to apply Theorem \ref{main-contract}. We proceed by verifying that Assumptions \ref{sampling-uncert} - \ref{loc-conc} hold. For Assumption \ref{sampling-uncert}, let $\mathcal{S}_n = \mathcal{R}_n $ where $\mathcal{R}_n$ is specified below in (\ref{rn-1}). For any $F$, \begin{align*}
    & \mathcal{G}(\mathbb{P}_{n,Y}, \varphi_{F})  - \mathcal{G}(\mathbb{P}_{Y}, \varphi_{F})  = (\E_n - \E)[e^{\mathbf{i}t' W}]  - \varphi_{F}(t) (\E_n - \E)[e^{\mathbf{i}t' \epsilon}].
\end{align*}
Since $\left| \varphi_F (t) \right| \leq 1$, Lemma \ref{aux2} implies that \begin{align*}
     \sup_{F \in \mathcal{S}_n} \sup_{t \in \mathbb{B}(T_n)} \left|   \mathcal{G}(\mathbb{P}_{n,Y}, \varphi_{F}) (t)  - \mathcal{G}(\mathbb{P}_{Y}, \varphi_{F}) (t)   \right| \leq D \frac{\sqrt{\log T_n}}{\sqrt{n}}
\end{align*}
with $\mathbb{P}$ probability approaching $1$. It follows that \begin{align*}
     \sup_{F \in \mathcal{S}_n} \| \mathcal{G}(\mathbb{P}_{n,Y}, \varphi_{F})  - \mathcal{G}(\mathbb{P}_{Y}, \varphi_{F})    \|_{\mathbb{B}(T_n)}^2   & = \sup_{F \in \mathcal{S}_n} \int_{\mathbb{B}(T_n)} \left|   \mathcal{G}(\mathbb{P}_{n,Y}, \varphi_{F}) (t)  - \mathcal{G}(\mathbb{P}_{Y}, \varphi_{F}) (t)   \right|^2 dt \\ & \leq D \frac{ T_n \log T_n}{n} \\ & \leq D \epsilon_n^2 
\end{align*}
with $\mathbb{P}$ probability approaching $1$.

For Assumption \ref{weak-bias}, we proceed through several steps. Fix any $ L > 1$ such that $ T_n \epsilon_n^{2L} \lessapprox \epsilon_n^2  $. As the latent distribution $F_X$ satisfies $F_X( t \in \R :  \| t \|  > z) \leq C \exp(- C' z^{\chi})$, there exists  a universal constant $R > 0$ such that the cube $I_n =[-R  (\log \epsilon_n^{-1})^{1/\chi},R  (\log \epsilon_n^{-1})^{1/\chi} ] $ satisfies $ 1- F_0(I_n) \leq D \epsilon_n^{L}$. Denote the probability measure induced from the restriction of $F_X$ to $I_n$ by $$ \overline{F}_X(A) = \frac{F_X(A \cap I_n)}{F_X(I_n)}  \; \; \; \; \; \forall \; \; \text{Borel}  \; A \subseteq \R. $$
With $ \overline{F}_X$ as above, observe that \begin{align*}
 \sup_{t \in \R} \left| \varphi_{X} (t) -  \varphi_{\overline{F}_X}(t) \right| =  \sup_{t \in \R} \left|  \int_{\R} e^{\mathbf{i} t' x} d(F_X - \overline{F}_X ) (x)     \right|   \leq \| F_X - \overline{F}_X   \|_{TV}  \leq 1- F_X(I_n) \leq D \epsilon_n^L.
\end{align*}
It follows that \begin{align*}
     \| \mathcal{G}(\mathbb{P}_{Y}, \varphi_{X})  - \mathcal{G}(\mathbb{P}_{Y}, \varphi_{\overline{F}_X })    \|_{\mathbb{B}(T_n)}^2   & = \int_{\mathbb{B}(T_n)} \left| \varphi_{\epsilon}(t)      \right|^2 \left|  \varphi_{X} (t) -  \varphi_{\overline{F}_X}(t)  \right|^2 dt \\ & \leq  D \epsilon_n^{2L} T_n \\ & \leq D \epsilon_n^2.
\end{align*}
Next, for any $\sigma  > 0$, observe that \begin{align*}
    \| \mathcal{G}(\mathbb{P}_{Y}, \varphi_{\overline{F}_X})  - \mathcal{G}(\mathbb{P}_{Y}, \varphi_{\overline{F}_X , \sigma^2 I})   \|_{\mathbb{B}(T_n)}^2 = \int_{\mathbb{B}(T_n)} \left| \varphi_{\epsilon}(t)    \right|^2 \left| (1- e^{ - \| t \|^2 \sigma^2 / 2})     \right|^2  \left| \varphi_{\overline{F}_X}(t)   \right|^2 dt
\end{align*}
Since $1- e^{-x} \leq x$,  we obtain \begin{align*}
    \| \mathcal{G}(\mathbb{P}_{Y}, \varphi_{\overline{F}_X})  - \mathcal{G}(\mathbb{P}_{Y}, \varphi_{\overline{F}_X , \sigma^2 I})   \|_{\mathbb{B}(T_n)}^2 & \leq  \frac{\sigma^4}{4}  \int_{\mathbb{B}(T_n)} \| t \|^4 dt  \leq D \sigma^4 T_n^{5}.
\end{align*}
In particular, with $\sigma = \sigma_n$, we obtain  \begin{align*}
    \| \mathcal{G}(\mathbb{P}_{Y}, \varphi_{\overline{F}_X})  - \mathcal{G}(\mathbb{P}_{Y}, \varphi_{\overline{F}_X , \sigma_n^2 I})   \|_{\mathbb{B}(T_n)}^2 \leq D \sigma_n^4 T_n^{5} \leq D \epsilon_n^2.
\end{align*}
By an application of Lemma \ref{aux10}, there exists a discrete probability measure $ P_{n} = \sum_{i=1}^N p_i \delta_{\mu_i}  $  where $N = D \max \{T_n (\log n)^{1/\chi}, \log n  \}$, $\mu_i \in I_{n}$ and
$$ \|   \varphi_{P_{n}, \sigma_n^2 I}   - \varphi_{\overline{F}_X, \sigma_n^2 I}       \|_{\mathbb{B}(T_n)} \leq D \epsilon_n .$$
From the second claim of Lemma \ref{aux10}, we can also assume without loss of generality that the support points satisfy $\inf_{k \neq j} \| \mu_k  - \mu_j \| \geq   \epsilon_n T_n^{-3/2}$. From combining the preceding estimates, it follows that Assumption \ref{weak-bias} holds with  $F_n =  \phi_{P_n , \sigma_n^2 I}$.

It remains to verify Assumption \ref{loc-conc}. Define  \begin{align*}
    \Omega_n =  \bigg \{ \sigma \in \R_+   : \sigma^2 \in \bigg[  \frac{ \sigma_n^2 }{1 + \epsilon_n T_n^{-3/2}  }  ,      \sigma_n^2   \bigg ]    \bigg \}
\end{align*}
Observe that for any distribution $P$ and  $\sigma \in \Omega_n$, an application of Lemma \ref{aux11} yields \begin{align*}
    \|  \varphi_{P, \sigma_n^2 I}     -  {\varphi}_{P, \sigma^2}   \|_{\mathbb{B}(T_n)} \leq  D T_n^{3/2} \left|   \sigma^2 - \sigma_n^2   \right|  \leq D \epsilon_n.
\end{align*}
Define $V_i = \{ t \in \R : \| t - \mu_i \| \leq  \epsilon_n^2 T_n^{-3/2} \} $ for $i=1,\dots,N$ and $V_0 = \R \setminus \bigcup_{i=1}^N V_i $. From the definition of the $\{\mu_i \}_{i=1}^N$, it follows that $\{V_0, V_1, \dots , V_N  \}$ is a disjoint partition of $\R$. For any fixed distribution $P$, an application of  Lemma \ref{aux4} yields  \begin{align*}
     \|  \varphi_{P, \sigma_n^2 I} - \varphi_{P_n , \sigma_n^2 I}   \|_{\mathbb{B}(T_n)} & \leq D \bigg[ \epsilon_n  + T_n^{1/2} \sum_{j=1}^N \left| P(V_j) - p_j \right|  \bigg] .
 \end{align*}
Define the set \begin{equation} \label{rn-1} \mathcal{R}_n =  \bigg \{  \phi_{P , \sigma^2} :  \sigma \in \Omega_n \:,\: \sum_{j=1}^N \left| P(V_j) - p_j   \right| \leq \epsilon_n  T_n^{-1/2}     \bigg   \} . \end{equation}
From the preceding bounds, it follows that any $F \in \mathcal{R}_n$ satisfies \begin{align*}
    \| \mathcal{G}(\mathbb{P}_{Y}, \varphi_{F})  - \mathcal{G}(\mathbb{P}_{Y}, \varphi_{F_n})    \|_{\mathbb{B}(T_n)} \leq \| \varphi_{F} - \varphi_{F_n}   \|_{\mathbb{B}(T_n)} \leq D \epsilon_n.
\end{align*}
Since the prior is a product measure $\nu = \text{DP}_{\alpha} \otimes G$, we have \begin{align*}  \nu(F \in \mathcal{R}_n) = 
     \int_{ \Sigma \in \Omega_n }\int_{P : \sum_{j=1}^N \left| P(V_j) - p_j   \right| \leq \epsilon_n T_n^{-1/2}    } d  \text{DP}_{\alpha}(P)  d G(\Sigma).
\end{align*}
As $ \text{DP}_{\alpha} $ is constructed using a Gaussian base measure $\alpha $, it is straightforward to verify that $ 
\inf_{j=1}^N \alpha(V_j) \geq C  \epsilon_n^{2} T_n^{-3/2}   \exp(- C' ( \log \epsilon_n^{-1})^{2/\chi}  ) $ for universal constants $C,C' > 0$.  By definition of $ \text{DP}_{\alpha}$, $(P(V_1),\dots, P(V_N)) \sim \text{Dir}(N, \alpha(V_1), \dots,  \alpha(V_N))$. As $N = D \max \{T_n (\log n )^{ 1/ \chi }, \log n \}$, an application of \citep[Lemma G.13]{ghosal2017fundamentals} and the definition of $(T_n,\epsilon_n^2)$ implies   \begin{align*}
    \int_{P : \sum_{j=1}^N \left| P(V_j) - p_j   \right| \leq \epsilon_n T_n^{-1/2}    } d  \text{DP}_{\alpha}(P) & \geq   C \exp \big( -C' \max \{ T_n (\log n)^{1/\chi},\log n   \} (\log n)^{\max(2/\chi,1)}    \big) \\ &          \geq C \exp(-C' n \epsilon_n^2).
\end{align*}
It remains to bound the outer integral. By Assumption \ref{cov-prior} (with $\kappa \leq 1/2$) and the definition of $\Omega_n$, there exists universal constant $C,C',C'' > 0$ such that \begin{align*}
     \int_{  \Sigma \in \Omega_n } d G(\Sigma)  \geq C \exp \big( - C' \sigma_n^{-2 \kappa}    \big)    \geq C \exp(- C' \sigma_n^{-1}) \geq  C \exp ( -C'' n \epsilon_n^2  ).
\end{align*}
Assumption \ref{loc-conc} follows. By Theorem \ref{main-contract}, we obtain \begin{align*}
    \nu \bigg ( F : \|  \widehat{\varphi}_{\epsilon} \big( \varphi_X - \varphi_F     \big)    \|_{\mathbb{B}(T_n)}  > D \epsilon_n  \: \bigg| \: \mathcal{D}_n  \bigg ) = o_{\mathbb{P}}(1).
\end{align*}
Note that $ \epsilon_n^2 \asymp n^{-4/5} T_n $.

\item[\textbf{$(ii)$}] We aim to apply Corollary \ref{contract-strong} with the metric $d = \mathbf{W}_1$. Given any $\delta > 0$ and $s > 1$, define the set \begin{align*}
    \mathcal{H}_n = \bigg \{ F = \phi_{P, \sigma^2} :  \sigma^{\frac{(s-1)}{1+2s} + 1} \leq \epsilon_n^{\frac{-2(s-1)}{1+2s}} n^{- \delta}  \: , \: 
 \bigg( \int \| x \|^s dP(x) \bigg)^{\frac{(s-1)}{s(1+2s)} + 1/s} \leq  \epsilon_n^{\frac{-2(s-1)}{1+2s}} n^{-\delta}     \bigg \}.
\end{align*}
First, we aim to show that for a suitable choice of $(s,\delta)$, we have $\nu(\mathcal{H}_n^c)  \lessapprox \exp(-B_n)$ for some $n \epsilon_n^2 = o(B_n) $. 
Observe that  \begin{align*}
\nu(\mathcal{H}_n^c) \leq&   G( \sigma : \sigma^{\frac{(s-1)}{1+2s} + 1} >  \epsilon_n^{\frac{-2(s-1)}{1+2s}} n^{- \delta} ) \\ & + \text{DP}_{\alpha} \bigg(P :  \bigg( \int \| x \|^s dP(x) \bigg)^{\frac{(s-1)}{s(1+2s)} + 1/s} >  \epsilon_n^{\frac{-2(s-1)}{1+2s}} n^{- \delta}     \bigg ).
\end{align*}
The first probability can be expressed as \begin{align*}
    G( \sigma : \sigma^{\frac{(s-1)}{1+2s} + 1} >  \epsilon_n^{\frac{-2(s-1)}{1+2s}} n^{- \delta}  ) = G(\sigma^2 : \sigma^2 > \epsilon_n^{-4(s-1)/3s}   n^{-2 \delta(2s+1)/3s }   ).
\end{align*}
By Assumption \ref{cov-prior}$(iii)$, we have  \begin{align*}
    & G(\sigma^2 : \sigma^2 > \epsilon_n^{-4(s-1)/3s}   n^{-2 \delta(2s+1)/3s }   ) \\ & \leq C \exp(- L'\epsilon_n^{-4 v_3 (s-1)/3s}   n^{-2 v_3 \delta(2s+1)/3s }    )
\end{align*}
for some universal constants $C,L' > 0$.  Since $\epsilon_n^2  \lessapprox  n^{-4/5}  T_n $, $v_3 \geq 1$ and $T_n = o(n^{1/5})$ we can pick $s > 1$ sufficiently large and $\delta > 0$ sufficiently small such that we obtain $ \epsilon_n^{-4 v_3 (s-1)/3s}   n^{-2 v_3 \delta(2s+1)/3s } = B_n  $ for some $n \epsilon_n^2 = o (B_n)$. Similarly, since $ \text{DP}_{\alpha} $ is constructed using a Gaussian base measure $\alpha $, we have \begin{align*}
   &  \text{DP}_{\alpha} \bigg(P :  \bigg( \int \| x \|^s dP(x) \bigg)^{\frac{(s-1)}{s(1+2s)} + 1/s} >  \epsilon_n^{\frac{-2(s-1)}{1+2s}} n^{- \delta}     \bigg ) \\ & \leq D \exp(-D' \epsilon_n^{-4  (s-1)/3s}   n^{-2  \delta(2s+1)/3s }  )
\end{align*}
for some $D,D' > 0$. As above, we can pick $s > 1$ sufficiently large and $\delta > 0$ sufficiently small such that $ \epsilon_n^{-4  (s-1)/3s}   n^{-2  \delta(2s+1)/3s } = B_n  $ for some $n \epsilon_n^2 = o (B_n)$. In the remainder of this proof, we fix $s > 1$ sufficiently large and $\delta > 0$ sufficiently small so that the requirements above are satisfied. In particular, we then obtain $\nu(\mathcal{H}_n^c)  \lessapprox \exp(-B_n)$ for some $n \epsilon_n^2 = o(B_n) $ so that the conditions of  Corollary \ref{contract-strong} are satisfied. 

To finish the proof, it remains to compute the modulus over $\mathcal{H}_n$. Let $K: \R \rightarrow (0, \infty)$ denote any symmetric density function on $\R$ with Fourier transform $\mathcal{F}(K)$  that is bounded, has support contained in $[-1,1]$ and has finite moments of order $s$. Define $K_R(x) = R K(Rx)$ for every $R > 0$ . Since $\mathbf{W}_1$ is a metric, we have for any $F \in \mathcal{H}_n$, the bound \begin{align*}
    \mathbf{W}_1(F,F_X) \leq \mathbf{W}_1(F, F \star K_R) + \mathbf{W}_1(F_X , F_X \star K_R) + \mathbf{W}_1(F \star K_R , F_X \star K_R).
\end{align*}
By definition of $\mathbf{W}_1$, we have that \begin{align*}
    & \mathbf{W}_1(F, F \star K_R) \leq \inf  \limits_{V \sim F , V' \sim F, \epsilon \sim K } \E \|  V -(V' + R^{-1} \epsilon)    \| \leq R^{-1} \E_{\epsilon \sim K} \| \epsilon \| \leq D R^{-1} .
\end{align*}
where $\inf$ is over all couplings for $(V,V' + R^{-1} \epsilon)$ and the last inequality follows from the coupling with $V = V'$. An analogous bound holds for $F_X$. By \citet[Theorem 6.15]{villani2009optimal}, we have $
    \mathbf{W}_1(F, F') \leq \int \| x \| d \left| F - F'   \right|(x) $
for all distributions $F,F'$. From this bound and an analogous argument to  \citet[Lemma 6]{nguyen2013convergence}, we have \begin{align*}
   &  \int \| x \| d \left| F \star K_R - F_X \star K_R   \right|(x) \\ & \leq D \bigg( \int  \| x \|^s d F + \int \| x \|^s d F_X + \int \| x \|^s d K_R      \bigg)^{\frac{(s-1)}{s(1+2s)} + 1/s} \| \varphi_{F \star K_R} - \varphi_{F_X \star K_R}    \|_{L^2}^{\frac{2(s-1)}{1+2s}}
\end{align*}
By Assumption, we have $ \int \| x \|^s d K_R  = R^{-s} \int \| x \|^s d K < \infty    $. Since  $F_X( t \in \R :  \| t \|  > z) \leq C \exp(- C' z^{\chi})$ for some $C,C', \chi  > 0$, it follows that $ \int \| x \|^s d F_X < \infty$. Furthermore, for every distribution $F = \phi_{P, \sigma^2} \in \mathcal{H}_n$, we have \begin{align*}
    \int \| x \|^s d F & \leq D  \bigg( \int \| x \|^s d P + \int \| x \|^s d \mathcal{N}(0, \sigma^2) \bigg)  \leq D \bigg( \int \| x \|^s d P + \sigma^s   \bigg) .
\end{align*} 
From combining the preceding bounds and using the definition of $\mathcal{H}_n$, it follows that \begin{align*}
    & \bigg( \int  \| x \|^s d F + \int \| x \|^s d F_X + \int \| x \|^s d K_R      \bigg)^{\frac{(s-1)}{s(1+2s)} + 1/s}  \leq D \epsilon_n^{\frac{-2(s-1)}{1+2s}} n^{- \delta}  
\end{align*}
for every distribution $F \in \mathcal{H}_n$. In particular, for all such $F$, we obtain \begin{align*}
    \mathbf{W}_1(F,F_X) \leq D  \bigg[ \epsilon_n^{\frac{-2(s-1)}{1+2s}} n^{- \delta}  \| \varphi_{F \star K_R} - \varphi_{F_X \star K_R}    \|_{L^2}^{\frac{2(s-1)}{1+2s}} + R^{-1}   \bigg].
\end{align*}
From the argument above, we also note that the constant $D$ can be chosen independent of $R \geq 1$. By assumption, we have $\inf_{\| t \|_{\infty} \leq R} \left| \varphi_{\epsilon}(t)    \right| \geq b \exp(-B R^{\zeta})   $ for some constants $b,B, \zeta > 0$. Let $R = R_n =  \min \{ c_0 (\log n)^{1/\zeta} , T_n \} $ for a constant $c_0$ sufficiently small such that $ 2 B c_0^{\zeta} (s-1) /(1+2s) = \xi < \delta  $. By assumption, $\mathcal{F}(K)$ is bounded and has support contained in $[-1,1]$. Therefore, \begin{align*}
    \| \varphi_{F \star K_{R_n}} - \varphi_{F_X \star K_{R_n}}    \|_{L^2}^2 & = \int_{\mathbb{B}(R_n)} \left| \varphi_F(t) - \varphi_X(t)   \right|^2 \left| \mathcal{F}(K)(t)   \right|^2 dt \\ &  = \int_{\mathbb{B}(R_n)} \frac{\left|  \varphi_{\epsilon}(t) \right|^2}{\left|  \varphi_{\epsilon}(t) \right|^2}   \left| \varphi_F(t) - \varphi_X(t)   \right|^2 \left| \mathcal{F}(K)(t)   \right|^2 dt \\ & \leq b \exp( 2 B R_n^\zeta) \int_{\mathbb{B}(R_n)} \left| \varphi_{\epsilon}(t)  \right|^2  \left| \varphi_F(t) - \varphi_X(t)   \right|^2 \left| \mathcal{F}(K)(t)   \right|^2 dt \\ & \leq D \exp( 2B R_n^\zeta)    \int_{\mathbb{B}(R_n)} \left| \varphi_{\epsilon}(t)  \right|^2  \left| \varphi_F(t) - \varphi_X(t)   \right|^2 dt \\ & \leq D \exp( 2B R_n^\zeta) \int_{\mathbb{B}(T_n)}  \left| \varphi_{\epsilon}(t)  \right|^2  \left| \varphi_F(t) - \varphi_X(t)   \right|^2 dt \\ & = D \exp( 2 B R_n^\zeta) \| \mathcal{G}(\mathbb{P}_{Y}, \varphi_{F})  - \mathcal{G}(\mathbb{P}_{Y}, \varphi_{X})    \|_{\mathbb{B}(T_n)}^2.  
\end{align*}
Since $R_n \leq  c_0 (\log n)^{1/\zeta} $, we obtain for all $F \in \mathcal{H}_n $, the bound \begin{align*}
    \mathbf{W}_1(F,F_X) \leq  D \bigg[ \epsilon_n^{\frac{-2(s-1)}{1+2s}} n^{- \delta + \xi } \| \mathcal{G}(\mathbb{P}_{Y}, \varphi_{F})  - \mathcal{G}(\mathbb{P}_{Y}, \varphi_{X})    \|_{\mathbb{B}(T_n)}^{\frac{2(s-1)}{1+2s}} +  \frac{1}{R_n}    \bigg ]
\end{align*}

By Lemma \ref{aux2}, we have that \begin{align*}
    \| \mathcal{G}(\mathbb{P}_{Y}, \varphi_{F})  - \mathcal{G}(\mathbb{P}_{Y}, \varphi_{X})    \|_{\mathbb{B}(T_n)}   &  \leq D \frac{\sqrt{T_n} \sqrt{\log T_n} }{\sqrt{n}}  +  \| \mathcal{G}(\mathbb{P}_{n,Y}, \varphi_{F})  - \mathcal{G}(\mathbb{P}_{n,Y}, \varphi_{X})    \|_{\mathbb{B}(T_n)} \\ & \leq D \epsilon_n + \| \mathcal{G}(\mathbb{P}_{n,Y}, \varphi_{F})  - \mathcal{G}(\mathbb{P}_{n,Y}, \varphi_{X})    \|_{\mathbb{B}(T_n)}.
\end{align*}
From combining the preceding bounds, it follows that the modulus satisfies  \begin{align*}
    \omega(d,\mathcal{H}_n,\epsilon_n) \leq D \big[  n^{-\delta + \xi } +  R_n^{-1} \big] .
\end{align*}
Since $R_n = \min \{ c_0 (\log n)^{1/\zeta} , T_n \}$, the claim follows.

\end{enumerate}

\end{proof}

\begin{proof}[Proof of Theorem \ref{deconv-2}]

The proof proceeds through several steps which we outline below. Depending on whether the model is mildly or severely ill-posed, define $\lambda$ as follows.
\begin{align*}
    \lambda = \begin{cases}
        \max \{ \chi^{-1}(d+2) + d/2 , d+1    \} & \text{mildly ill-posed} \\ \max \{ \chi^{-1}(d+2) + d/2 , d+1  , d/ \zeta  + 1/2  \} & \text{severely ill-posed}.
    \end{cases}
\end{align*}
Let $\epsilon_n^2 = n^{-1} (\log n)^{\lambda}  $.
\begin{enumerate}
    \item[\textbf{$(i)$}] We aim to apply Theorem \ref{main-contract}. We proceed by verifying that Assumptions \ref{sampling-uncert} - \ref{loc-conc} hold. 
    
Assumption \ref{sampling-uncert} follows from an analogous argument to that of Theorem \ref{deconv-1}. For Assumption \ref{weak-bias}, we proceed through several steps. By Condition \ref{exact-g}, the mixing distribution $F_0$ satisfies $F_0( t \in \R^d :  \| t \|  > z) \leq C \exp(- C' z^{\chi})$. There exists a universal constant $R > 0$ such that the cube $I_n =[-R  (\log \epsilon_n^{-1})^{1/\chi},R  (\log \epsilon_n^{-1})^{1/\chi} ]^d $ satisfies $ 1- F_0(I_n) \leq D \epsilon_n$. Denote the probability measure induced from the restriction of $F_0$ to $I_n$ by $$ \overline{F}_0(A) = \frac{F(A \cap I_n)}{F(I_n)}  \; \; \; \; \; \forall \; \; \text{Borel}  \; A \subseteq \R^d. $$
Note that the restricted probability measure satisfies \begin{align*}
 \sup_{t \in \R^d} \left| \varphi_{F_0} (t) -  \varphi_{\overline{F}_0}(t) \right| =  \sup_{t \in \R^d} \left|  \int_{\R^d} e^{\mathbf{i} t' x} d(F - \overline{F}_0 ) (x)     \right|   \leq \| F_0 - \overline{F}_0   \|_{TV}  \leq 1- F_0(I_n) \leq D \epsilon_n.
\end{align*}
By Condition \ref{exact-g}, $f_X = \phi_{\Sigma_0} \star F_0$ for some positive definite matrix $\Sigma_0$. As the eigenvalues of $\Sigma_0$ are bounded away from zero, it follows that
\begin{align*}
 \| \mathcal{G}(\mathbb{P}_{Y}, \varphi_{X})  - \mathcal{G}(\mathbb{P}_{Y}, \varphi_{ \overline{F}_0, \Sigma_0 })    \|_{\mathbb{B}(T_n)}^2  & = \int_{\R^d} \left|  \varphi_{\Sigma_0} (t)  \right|^2 \left|  \varphi_{\epsilon} (t)  \right|^2 \left|  \varphi_{F_0}(t)  - \varphi_{\overline{F}_0}(t) \right|^2 dt \\ & \leq  \int_{\R^d} \left|  \varphi_{\Sigma_0} (t)  \right|^2 \left|  \varphi_{F_0}(t)  - \varphi_{\overline{F}_0}(t) \right|^2 dt      \\    &  \leq  D \epsilon_n^2 \int_{\R^d} e^{- t' \Sigma_0 t} dt \\ & \leq D \epsilon_n^2.
\end{align*}
Let $\iota = \max \{ d, d/\chi + d/2   \} $. By Lemma \ref{aux6}, there exists a discrete probability measure $ F_0^* = \sum_{i=1}^N p_i \delta_{\mu_i}  $ where $N = D  \big(\log( \epsilon_n^{-1}) \big)^{ \iota  }$ and $\mu_i \in I_{n}$,  that satisfies
$$ \| \mathcal{G}(\mathbb{P}_{Y}, \varphi_{ F_0^*, \Sigma_0 })  - \mathcal{G}(\mathbb{P}_{Y}, \varphi_{ \overline{F}_0, \Sigma_0 })    \|_{\mathbb{B}(T_n)} \leq  \|  \varphi_{ \overline{F}_0, \Sigma_0  }  - \varphi_{ F_0^* , \Sigma_0  }       \|_{L^2} \leq D \epsilon_n.$$
From the second claim of Lemma \ref{aux6}, we can also assume without loss of generality that the support points satisfy $\inf_{k \neq j} \| \mu_k  - \mu_j \| \geq c_0 \epsilon_n$ for some constant $c_0 > 0$ (depending only on $\Sigma_0$). From combining the preceding estimates, Assumption \ref{weak-bias} holds with  $F_n =  \phi_{F_0^* , \Sigma_0}$.

It remains to verify Assumption \ref{loc-conc}. Define \begin{align*}
    \Omega_n =  \bigg \{ \Sigma \in \mathbf{S}_{+}^d    : \|  \Sigma - \Sigma_0 \| \leq  \epsilon_n   \bigg \}.
\end{align*}
Observe that for any distribution $P$ and  $\Sigma \in \Omega_n$, an application of Lemma \ref{aux7} yields \begin{align*} \| \varphi_{\epsilon}  \varphi_{P,\Sigma} - \varphi_{\epsilon}  \varphi_{P,\Sigma_0}     \|_{\mathbb{B}(T_n)}^2  \leq
   D \| \varphi_{P,\Sigma} - \varphi_{P,\Sigma_0}   \|_{L^2}^2 \leq D \| \Sigma - \Sigma_0 \|^2 \leq D \epsilon_n^2.
\end{align*}
Let $V_i = \{ t \in \R^d : \| t - \mu_i \| \leq  c_0 \epsilon_n / 2  \} $ for $i=1,\dots,N$ and $V_0 = \R^d \setminus \bigcup_{i=1}^N V_i $. From the definition of the $\{\mu_i \}_{i=1}^N$, it follows that $\{V_0, V_1, \dots , V_N  \}$ is a disjoint partition of $\R^d$. For any fixed $(P,\Sigma)$, an application of  Lemma \ref{aux4} yields \begin{align*}
     \| \varphi_{P, \Sigma_0} - \varphi_{ F_0^* , \Sigma_0}  \|_{L^2} & \leq D \bigg( \epsilon_n^{}  +  \sum_{j=1}^N \left| P(V_j) - p_j \right|  \bigg) .
 \end{align*}
Define the set \begin{equation} \label{rn-2} \mathcal{R}_n =  \bigg \{  \phi_{P , \Sigma} :  \Sigma \in \Omega_n \:,\: \sum_{j=1}^{N} \left| P(V_j) - p_j     \right| \leq    \epsilon_n     \bigg   \} . \end{equation}
From the preceding bounds, it follows that any $F \in \mathcal{R}_n$ satisfies \begin{align*}
    \| \mathcal{G}(\mathbb{P}_{Y}, \varphi_{F})  - \mathcal{G}(\mathbb{P}_{Y}, \varphi_{F_n})    \|_{\mathbb{B}(T_n)} \leq \| \varphi_{F} - \varphi_{F_n}   \|_{\mathbb{B}(T_n)} \leq D \epsilon_n.
\end{align*}
Since the prior is a product measure $\nu = \text{DP}_{\alpha} \otimes G$, we have \begin{align*}  \nu(F \in \mathcal{R}_n) = 
     \int_{ \Sigma \in \Omega_n }\int_{P : \sum_{j=1}^N \left| P(V_j) - p_j   \right| \leq \epsilon_n     } d  \text{DP}_{\alpha}(P)  d G(\Sigma).
\end{align*}
As $ \text{DP}_{\alpha} $ is constructed using a Gaussian base measure $\alpha $, it is straightforward to verify that $ 
\inf_{j=1}^N \alpha(V_j) \geq C \epsilon_n^d \exp(- C' ( \log \epsilon_n^{-1})^{2/\chi}  ) $ for universal constants $C,C' > 0$. By definition of $ \text{DP}_{\alpha}$, $(P(V_1),\dots, P(V_N)) \sim \text{Dir}(N, \alpha(V_1), \dots,  \alpha(V_N))$. As $N = D \{ \log ( \epsilon_n^{-1} )  \}^{\iota}$, an application of \citep[Lemma G.13]{ghosal2017fundamentals} implies  \begin{align*}
    \int_{P : \sum_{j=1}^N \left| P(V_j) - p_j   \right| \leq \epsilon_n    } d  \text{DP}_{\alpha}(P) \geq C \exp \big (  - C'  ( \log \epsilon_n^{-1} )^{\iota + \max \{2/ \chi,1   \}  }   \big)  & = C \exp \big (  - C'  ( \log \epsilon_n^{-1} )^{\lambda  }   \big) \\ & \geq  C \exp \big (  - C''  n \epsilon_n^2   \big) 
\end{align*}
for universal constants $C,C',C'' > 0$.

It remains to bound the outer integral. The law of $G = G_n$ is given by $\Omega / \sigma_n^2$ where $\Omega \sim L$ and $L$ is a probability measure on  $\mathbf{S}_+^d$ that  satisfies Assumption \ref{cov-prior}. By Assumption \ref{cov-prior} and the definition of $\sigma_n^2$, there exists a universal constant $C,C',C'' > 0$ such that \begin{align*}
     \int_{ \Sigma : \| \Sigma - \Sigma_0  \| \leq  \epsilon_n} d G(\Sigma) = \int_{\Sigma : \| \Sigma - \sigma_n^2 \Sigma_0   \| \leq \sigma_n^2 \epsilon_n  } d L(\Sigma) \geq C  \exp \big( - C' \sigma_n^{-2 \kappa}   \big) \geq C \exp ( -C'' n \epsilon_n^2  )
\end{align*}
Assumption \ref{loc-conc} follows. By Theorem \ref{main-contract}, we obtain \begin{align*}
    \nu \bigg ( F : \|  \widehat{\varphi}_{\epsilon} \big( \varphi_X - \varphi_F     \big)    \|_{\mathbb{B}(T_n)}  > D \epsilon_n  \: \bigg| \: \mathcal{D}_n  \bigg ) = o_{\mathbb{P}}(1).
\end{align*}

\item[\textbf{$(ii)$}] We aim to apply Corollary \ref{contract-strong} with the metric $d = \| . \|_{L^2}$. 

The law of  $G = G_n$ is given by $\Omega / \sigma_n^2$ where $\Omega \sim L$  and $L$ is a probability measure on  $\mathbf{S}_+^d$ that  satisfies Assumption \ref{cov-prior}. By Assumption \ref{cov-prior}, it follows that for every $E' > 0$, there exists $E > 0$ such that
\begin{align*}
  \int_{\Sigma : \| \Sigma^{-1}  \| > E \sigma_n^2 (n \epsilon_n^2)^{1/ \kappa}  }   d G(\Sigma) = \int_{\Sigma : \| \Sigma^{-1}  \| > E (n \epsilon_n^2)^{1/ \kappa}  }   d L(\Sigma) \leq \exp \big(  - E' n \epsilon_n^2  \big).
\end{align*}
In particular, we can choose $E > 0$ large enough such that the hypothesis of Corollary \ref{main-contract} holds with the set of distributions $$ \mathcal{H}_n = \{ \phi_{P,\Sigma} :  \| \Sigma^{-1}  \| \leq E \sigma_n^2 (n \epsilon_n^2)^{1/ \kappa}    \}. $$
It remains to compute the modulus  over $\mathcal{H}_n$. Fix any $R = R_n \leq T_n $. For any choice of $(P,\Sigma)$ satisfying  $\| \widehat{\varphi}_{Y}    - \widehat{\varphi}_{\epsilon} {\varphi}_{P,\Sigma}      \|_{\mathbb{B}(T_n)} \leq M \epsilon_n$, an application of Lemma \ref{aux2} yields   \begin{align*}
\| {\varphi}_{Y}    - {\varphi}_{\epsilon} {\varphi}_{P,\Sigma}      \|_{\mathbb{B}(R)}   & \leq \| {\varphi}_{Y} - \widehat{\varphi_{Y}}    \|_{\mathbb{B}(R)} + \|  ({\varphi}_{\epsilon} - \widehat{\varphi_{\epsilon}} ) \varphi_{P,\Sigma}   \|_{\mathbb{B}(R)} + \| \widehat{\varphi}_{Y}    - \widehat{\varphi}_{\epsilon} {\varphi}_{P,\Sigma}      \|_{\mathbb{B}(R)}  \\ & \leq  \| {\varphi}_{Y} - \widehat{\varphi_{Y}}    \|_{\mathbb{B}(R)} + \|  {\varphi}_{\epsilon} - \widehat{\varphi_{\epsilon}}   \|_{\mathbb{B}(R)} + \| \widehat{\varphi}_{Y}    - \widehat{\varphi}_{\epsilon} {\varphi}_{P,\Sigma}      \|_{\mathbb{B}(R)} \\ & \leq D \frac{\sqrt{ T_n^d \log T_n}}{\sqrt{n}}   +  \| \widehat{\varphi}_{Y}    - \widehat{\varphi}_{\epsilon} {\varphi}_{P,\Sigma}      \|_{\mathbb{B}(T_n)} \\ & \leq D \epsilon_n +  \| \widehat{\varphi}_{Y}    - \widehat{\varphi}_{\epsilon} {\varphi}_{P,\Sigma}      \|_{\mathbb{B}(T_n)} \\ & \leq D \epsilon_n.
\end{align*}
Let $\uptau_{R}  = \sup_{\| t \|_{\infty} \leq R } \left|  \varphi_{\epsilon} (t)  \right|^{-1} $. Since $ \varphi_Y = \varphi_X \varphi_{\epsilon} $, the preceding bound implies that
\begin{align*} 
   D \epsilon_n \geq  \| {\varphi}_{Y}    - {\varphi}_{\epsilon} {\varphi}_{P,\Sigma}      \|_{\mathbb{B}(R)} & = \| \big( {\varphi}_{X}   -   {\varphi}_{P,\Sigma} \big) \varphi_{\epsilon}      \|_{\mathbb{B}(R)} \\ & \geq   \uptau_{R}^{-1} \| \big( {\varphi}_{X}   -   {\varphi}_{P,\Sigma} \big) \mathbbm{1} \{  t \in \mathbb{B}(R)  \}   \|_{L^2}.
\end{align*}
Next, we examine the bias from truncating the $L^2$ norm to the set $\mathbb{B}(R)$. Suppose $ \|  \Sigma^{-1}  \| \leq M^2  \sigma_n^2 (n \epsilon_n^2)^{1/\kappa}  $ holds. Since $\sigma_n^{2} \asymp (n \epsilon_n^2)^{-1/\kappa}$, it follows that there exists a $c > 0$ for which $  \lambda_1(\Sigma ) \geq c  $ holds. It follows that there exists  universal constants $C,D > 0$ such that
\begin{align*}
    \| \big( {\varphi}_{X}   -   {\varphi}_{P,\Sigma} \big) \mathbbm{1} \{ \| t \|_{\infty} > R  \}   \|_{L^2}^2 & \leq  2 \|  {\varphi}_{X}    \mathbbm{1} \{ \| t \|_{\infty} > R  \}   \|_{L^2}^2 +  2 \|     {\varphi}_{P,\Sigma}  \mathbbm{1} \{ \| t \|_{\infty} > R  \}   \|_{L^2}^2 \\ & \leq 2 \int_{\| t \|_{\infty} > R}  e^{- t' \Sigma_0 t} dt + 2 \int_{\| t \|_{\infty} > R} e^{- t' \Sigma t} dt \\ & \leq 2 \int_{\| t \|_{\infty} > R}  e^{- t' \Sigma_0 t} dt + 2 \int_{\| t \|_{\infty} > R} e^{- c \| t \|^2     } dt \\ & \leq D   e^{- C R^2} R^{d-2}  .
\end{align*}
It follows that for any sequence $R_n \leq T_n$, the modulus satisfies \begin{align*}
    \omega^2(d,\mathcal{H}_n,\epsilon_n) \leq D \big[ \uptau_{R_n}^2 \epsilon_n^2 + e^{- C R_n^2} R_n^{d-2}  ]
\end{align*}
where $C,D > 0$ are universal constants. In the mildly ill-posed case, we take $R_n = T_n$. In this case, $\uptau_{R_n} \asymp T_n^{\zeta}, R_n^2 \asymp (\log n)(\log \log n)$ and the modulus reduces to \begin{align*}
    \omega(d,\mathcal{H}_n,\epsilon_n) \leq D T_n^{\zeta} \epsilon_n \leq D \frac{(\log n)^{(\lambda + \zeta)/2} (\log \log n)^{\zeta/2} }{\sqrt{n}}. 
\end{align*}
In the severely ill-posed case, we have $T_n \asymp (\log n)^{1/\zeta}$ and  $\uptau_{R} \leq D \exp(B R^{\zeta})$ for some universal constants $D,B > 0$ and $\zeta \in (0,2]$. If $\zeta \in (0,2)$, we let $R_n = v_n \sqrt{\log n}  $ for a sequence $v_n \uparrow \infty$ that grows as slow as desired. In this case, we obtain \begin{align*}
     \omega(d,\mathcal{H}_n,\epsilon_n) \leq D \uptau_{R_n} \epsilon_n \leq D \exp(B (\log n)^{\zeta/2} v_n^{\zeta}  ) (\log n)^{\lambda/2} n^{-1/2}.
\end{align*}
Since $\zeta \in (0,2)$ and $v_n \uparrow \infty$ grows as slow as desired, the term $\exp(B (\log n)^{\zeta/2} v_n^{\zeta}  ) (\log n)^{\lambda/2} $ can be chosen to grow slower than any polynomial $n^{\epsilon}$ for any $\epsilon > 0$. Now suppose $\zeta =2$. Fix any $\xi < 1/2  $ and let $R_n = c_0 T_n  $ for some constant $c_0 \leq 1 $ sufficiently small such that $B c_0^2 T_n^2 \leq (\log n) \xi$. It follows that \begin{align*}
    \omega(d,\mathcal{H}_n,\epsilon_n) \leq D \big[ n^{-1/2 + \xi} (\log n)^{\lambda/2} + e^{- 0.5C c_0^2 T_n^2 } (\log n)^{\frac{d-2}{4}}   ].
\end{align*}
Since $T_n \asymp \sqrt{\log n}$, this is upper bounded by a polynomial $n^{-V}$ for some $V \in (0,1/2)$. 

\end{enumerate}

\end{proof}

\begin{proof}[Proof of Theorem \ref{deconv-3}]
The proof proceeds through several steps which we outline below. Define the sequences \begin{align*}
 \alpha_n = \epsilon_n^{1/(p + \zeta)} \; \; , \; \;  \epsilon_n^2 = \frac{(\log n)^{\lambda + d/2 } }{n^{ \frac{2(p + \zeta)}{2(p + \zeta) + d}}  }          \; \; \; \;   ,  \; \; \; \;  \lambda = \begin{cases}
      \chi^{-1}(d +2)  & \chi < 2 \\ d/ \chi + 1 & \chi \geq 2. 
  \end{cases} 
\end{align*}
\begin{enumerate}
    \item[\textbf{$(i)$}] We aim to apply Theorem \ref{main-contract}. We proceed by verifying that Assumptions \ref{sampling-uncert} - \ref{loc-conc} hold. Assumption \ref{sampling-uncert} follows from an analogous argument to that of Theorem \ref{deconv-1}. For Assumption \ref{weak-bias}, we proceed through several steps. By Condition \ref{deconv-bias}, there exists $\chi > 0 $, universal constant $  C < \infty $, a mixing distribution $G_{\alpha_n}$ supported on the cube $ I_{n} = [-C  (\log \epsilon_n^{-1})^{1/\chi},C  (\log \epsilon_n^{-1})^{1/\chi} ]^d$  and a covariance matrix $\Sigma_{\alpha_n} $ with $\lambda_{\min}(\Sigma_{\alpha_n}) \geq  \alpha_n^2$ such that $\varphi_{G_{\alpha_n},\Sigma_{\alpha_n}}$  satisfies  \begin{align*}  \| \mathcal{G}(\mathbb{P}_Y, \varphi_X) - \mathcal{G}(\mathbb{P}_Y,\varphi_{G_{\alpha_n},\Sigma_{\alpha_n}})   \|     =   \| \varphi_{\epsilon}(\varphi_X - \varphi_{G_{\alpha_n},\Sigma_{\alpha_n}})  \|_{\mathbb{B}(T_n)} &  \leq  \| \varphi_{\epsilon}(\varphi_X - \varphi_{G_{\alpha_n},\Sigma_{\alpha_n}})  \|_{L^2} \\ & \leq  D \| f_X \star f_{\epsilon} - \phi_{G_{\alpha_n}, \Sigma_{\alpha_n}} \star f_{\epsilon}  \|_{L^2}     \\ & \leq D \alpha_n^{p + \zeta} \\ & = D \epsilon_n.  
    \end{align*}
The second inequality is because the Fourier transform preserves $L^2$ distance (up to a constant). By an application of Lemma \ref{aux10}, there exists a discrete probability measure $ G_{\alpha_n}^* = \sum_{i=1}^N p_i \delta_{\mu_i}  $ with $N \leq D T_n^d \{\log( \epsilon_n^{-1}) \}^{ d/ \chi }$  and  $\mu_i \in I_{n}$  that satisfies
$$  \| \mathcal{G}(\mathbb{P}_Y,\varphi_{G_{\alpha_n},\Sigma_{\alpha_n}}) - \mathcal{G}(\mathbb{P}_Y,\varphi_{G_{\alpha_n}^*,\Sigma_{\alpha_n}})  \| \leq \|   \varphi_{G_{\alpha_n}, \Sigma_{\alpha_n}}   - \varphi_{G_{\alpha_n}^*, \Sigma_{\alpha_n}}       \|_{\mathbb{B}(T_n)} \leq D \epsilon_n .$$
From the second claim of Lemma \ref{aux10}, we can also assume without loss of generality that the support points have separation satisfying $\inf_{k \neq j} \| \mu_k  - \mu_j \| \geq   \epsilon_n T_n^{-(d+2)/2} $. From combining the preceding estimates, it follows that Assumption \ref{weak-bias} holds with  $F_n =  \phi_{G_{\alpha_n}^* , \Sigma_{\alpha_n}}$.

It remains to verify Assumption \ref{loc-conc}. Define \begin{align*}
    \Omega_n =  \bigg \{ \Sigma \in \mathbf{S}_{+}^d  : \lambda_j(\Sigma) \in \bigg[  \frac{ \lambda_j(\Sigma_{\alpha_n})  }{1 + \epsilon_n T_n^{-(d+2)/2}  }  ,      \lambda_j(\Sigma_{\alpha_n})   \bigg ]  \; \; \forall \; j=1,\dots,d.   \bigg \}
\end{align*}
Observe that for any distribution $P$ and  $\Sigma \in \Omega_n$, an application of Lemma \ref{aux11} yields \begin{align*}
    \|  \varphi_{P, \Sigma_{\alpha_n}}     -  {\varphi}_{P, \Sigma}   \|_{\mathbb{B}(T_n)} \leq D T_n^{(d+2)/2} \|  \Sigma -  \Sigma_{\alpha_n} \|  & = D T_n^{(d+2)/2} \max_{j=1,\dots,d} \left|   \lambda_j(\Sigma) - \lambda_j(\Sigma_{\alpha_n})  \right| \\ & \leq D \epsilon_n.
\end{align*}
Define $V_i = \{ t \in \R^d : \| t - \mu_i \| \leq  \epsilon_n^2 T_n^{-(d+2)/2} \} $ for $i=1,\dots,N$ and $V_0 = \R^d \setminus \bigcup_{i=1}^N V_i $. From the definition of the $\{\mu_i \}_{i=1}^N$, it follows that $\{V_0, V_1, \dots , V_N  \}$ is a disjoint partition of $\R^d$. For any fixed distribution $P$ an application of  Lemma \ref{aux4} yields \begin{align*}
     \|  \varphi_{P, \alpha_n^2 I} - \varphi_{ F_{\alpha_n} , \alpha_n^2 I }  \|_{\mathbb{B}(T_n)} & \leq D \bigg[ \epsilon_n  + T_n^{d/2} \sum_{j=1}^N \left| P(V_j) - p_j \right|  \bigg] .
 \end{align*}
Define the set \begin{equation} \label{rn-3} \mathcal{R}_n =  \bigg \{  \phi_{P , \Sigma} :  \Sigma \in \Omega_n \:,\: \sum_{j=1}^{N} \left| P(V_j) - p_j     \right| \leq    \epsilon_n T_n^{-d/2}     \bigg   \} . \end{equation}
From the preceding bounds, it follows that any $F \in \mathcal{R}_n$ satisfies \begin{align*}
    \| \mathcal{G}(\mathbb{P}_{Y}, \varphi_{F})  - \mathcal{G}(\mathbb{P}_{Y}, \varphi_{F_n})    \|_{\mathbb{B}(T_n)} \leq \| \varphi_{F} - \varphi_{F_n}   \|_{\mathbb{B}(T_n)} \leq D \epsilon_n.
\end{align*}
Since the prior is a product measure $\nu = \text{DP}_{\alpha} \otimes G$, we have \begin{align*}  \nu(F \in \mathcal{R}_n) = 
     \int_{ \Sigma \in \Omega_n }\int_{P : \sum_{j=1}^N \left| P(V_j) - p_j   \right| \leq \epsilon_n T_n^{-d/2}    } d  \text{DP}_{\alpha}(P)  d G(\Sigma).
\end{align*}
As $ \text{DP}_{\alpha} $ is constructed using a Gaussian base measure $\alpha $, it is straightforward to verify that $ 
\inf_{j=1}^N \alpha(V_j) \geq C  \epsilon_n^{2d} T_n^{-d(d+2)/2}   \exp(- C' ( \log \epsilon_n^{-1})^{2/\chi}  ) $ for universal constants $C,C' > 0$.  By definition of $ \text{DP}_{\alpha}$, $(P(V_1),\dots, P(V_N)) \sim \text{Dir}(N, \alpha(V_1), \dots,  \alpha(V_N))$. As $N = D T_n^d \{\log( \epsilon_n^{-1}) \}^{ d/ \chi }$, an application of \citep[Lemma G.13]{ghosal2017fundamentals} and the definition of $(T_n,\epsilon_n^2)$ implies   \begin{align*}
    \int_{P : \sum_{j=1}^N \left| P(V_j) - p_j   \right| \leq \epsilon_n T_n^{-d/2}    } d  \text{DP}_{\alpha}(P) & \geq   C \exp \big (  - C'  T_n^d  \{   \log \epsilon_n^{-1} \}^{ d/\chi +   \max \{2/ \chi,1   \}  }   \big) \\ &          \geq C \exp(-C'' n \epsilon_n^2)
\end{align*}
for universal constants $C,C',C'' > 0$.

It remains to bound the outer integral. The law of $G= G_n$ is given by $\Omega / \sigma_n^2$ where $\Omega \sim L$ and $L$ is a probability measure on  $\mathbf{S}_+^d$ that  satisfies Assumption \ref{cov-prior}. By Assumption \ref{cov-prior} and the definition of $(\alpha_n^2,\sigma_n^2, \epsilon_n^2)$, there exists a universal constant $C,C',C'' > 0$ such that \begin{align*}
     \int_{  \Sigma \in \Omega_n } d G(\Sigma) = \int_{  \Sigma \in \sigma_n^2 \Omega_n } d L(\Sigma) \geq C \exp \big( - C' \sigma_n^{-2 \kappa} \alpha_n^{-2 \kappa}   \big) \geq C \exp ( -C'' n \epsilon_n^2  ).
\end{align*}
Assumption \ref{loc-conc} follows. By Theorem \ref{main-contract}, we obtain \begin{align*}
    \nu \bigg ( F : \|  \widehat{\varphi}_{\epsilon} \big( \varphi_X - \varphi_F     \big)    \|_{\mathbb{B}(T_n)}  > D \epsilon_n  \: \bigg| \: \mathcal{D}_n  \bigg ) = o_{\mathbb{P}}(1).
\end{align*}

\item[\textbf{$(ii)$}] We aim to apply Corollary \ref{contract-strong} with the metric $d = \| . \|_{L^2}$. 

The law of  $G = G_n$ is given by $\Omega / \sigma_n^2$ where $\Omega \sim L$  and $L$ is a probability measure on  $\mathbf{S}_+^d$ that  satisfies Assumption \ref{cov-prior}. By Assumption \ref{cov-prior}, it follows that for every $E' > 0$, there exists $E > 0$ such that
\begin{align*}
  \int_{\Sigma : \| \Sigma^{-1}  \| > E \sigma_n^2 (n \epsilon_n^2)^{1/ \kappa}  }   d G(\Sigma) = \int_{\Sigma : \| \Sigma^{-1}  \| > E (n \epsilon_n^2)^{1/ \kappa}  }   d L(\Sigma) \leq \exp \big(  - E' n \epsilon_n^2  \big).
\end{align*}
In particular, we can choose $E > 0$ large enough such that the hypothesis of Corollary \ref{main-contract} holds with the set of distributions $$ \mathcal{H}_n = \{ \phi_{P,\Sigma} :  \| \Sigma^{-1}  \| > E \sigma_n^2 (n \epsilon_n^2)^{1/ \kappa}    \}. $$
It remains to compute the modulus  over $\mathcal{H}_n$. Let $\uptau_{T}  = \sup_{\| t \|_{\infty} \leq T } \left|  \varphi_{\epsilon} (t)  \right|^{-1} $. By an analogous argument to the proof of Theorem \ref{deconv-2}, we have \begin{align*} 
   D \epsilon_n \geq  \| {\varphi}_{Y}    - {\varphi}_{\epsilon} {\varphi}_{P,\Sigma}      \|_{\mathbb{B}(T_n)} & = \| \big( {\varphi}_{X}   -   {\varphi}_{P,\Sigma} \big) \varphi_{\epsilon}      \|_{\mathbb{B}(T_n)} \\ & \geq   \uptau_{T_n}^{-1} \| \big( {\varphi}_{X}   -   {\varphi}_{P,\Sigma} \big) \mathbbm{1} \{  t \in \mathbb{B}(T_n)  \}   \|_{L^2}
\end{align*}
for every $ \phi_{P,\Sigma} \in \mathcal{H}_n$.
Next, we examine the bias from truncating the $L^2$ norm to the set $\mathbb{B}(T_n)$. Suppose $ \|  \Sigma^{-1}  \| \leq M^2  \sigma_n^2 (n \epsilon_n^2)^{1/\kappa}  $ holds. It follows that there exists a $c > 0$ for which $  \lambda_1(\Sigma ) \geq c (n \epsilon_n^2)^{-1/\kappa} \sigma_n^{-2} $ holds. Furthermore, since $ f_X \in \mathbf{H}^p  $, we have  $
    \int_{\|t \|_{\infty} > T_n} \left|  \varphi_X(t)  \right|^2 dt \leq D T_n^{-2p}$.
It follows that there exists universal constants $C,D > 0$ such that
\begin{align*}
    \| \big( {\varphi}_{X}   -   {\varphi}_{P,\Sigma} \big) \mathbbm{1} \{ \| t \|_{\infty} > T_n  \}   \|_{L^2}^2 & \leq  2 \|  {\varphi}_{X}    \mathbbm{1} \{ \| t \|_{\infty} > T_n  \}   \|_{L^2}^2 +  2 \|     {\varphi}_{P,\Sigma}  \mathbbm{1} \{ \| t \|_{\infty} > T_n  \}   \|_{L^2}^2 \\ & \leq 2 \int_{\|t \|_{\infty} > T_n} \left|  \varphi_X(t)  \right|^2 dt + 2 \int_{\| t \|_{\infty} > T_n} e^{- t' \Sigma t} dt \\ & \leq 2 \int_{\|t \|_{\infty} > T_n} \left|  \varphi_X(t)  \right|^2 dt + 2 \int_{\| t \|_{\infty} > T_n} e^{- c \| t \|^2  \sigma_n^{-2} (n \epsilon_n^2)^{-1/\kappa}   } dt \\ & \leq D \bigg[  T_n^{-2p} +  \sigma_n^{-2} (n \epsilon_n^2)^{-1/\kappa}   e^{- C T_n^2 \sigma_n^{-2} (n \epsilon_n^2)^{-1/\kappa} } T_n^{d-2}   \bigg ] .
\end{align*}
From the definition of $\sigma_n^2$, we have  $T_n^2 (n \epsilon_n^2)^{-1/\kappa} \sigma_n^{-2} \asymp  (\log n) ( \log \log n )  $. Hence, the preceding bound reduces to $DT_n^{-2p}$. From combining the preceding bounds (and noting that  $ T_n^{-p} \lessapprox \uptau_{T_n} \epsilon_n $), it follows that the modulus satisfies
$$ \omega_n(d, \mathcal{H}_n, \epsilon_n) \leq D  \uptau_{T_n} \epsilon_n.  $$
The claim follows from observing that \begin{align*}
    \uptau_{T_n} \epsilon_n \asymp T_n^{\zeta} \epsilon_n \asymp n^{-p/[2(p + \zeta) + d]} (\log n)^{(\lambda + \zeta)/2 + d/4}.
\end{align*}

    \end{enumerate}

\end{proof}

\begin{proof}[Proof of Corollary \ref{deconv-3-2}]
The proof is a simple modification to the argument in Theorem \ref{deconv-3}. We aim to compute the modulus with $d = \| . \|_{L^{\infty}}$ and  $ \mathcal{H}_n = \{ \phi_{P,\Sigma} :  \| \Sigma^{-1}  \| \leq E \sigma_n^2 (n \epsilon_n^2)^{1/ \kappa}    \}$. Let $\uptau_{T}  = \sup_{\| t \|_{\infty} \leq T } \left|  \varphi_{\epsilon} (t)  \right|^{-1} $. By an analogous argument to the proof of Theorem  \ref{deconv-3}, we have
\begin{align*} 
   D \epsilon_n \geq  \| {\varphi}_{Y}    - {\varphi}_{\epsilon} {\varphi}_{P,\Sigma}      \|_{\mathbb{B}(T_n)} & = \| \big( {\varphi}_{X}   -   {\varphi}_{P,\Sigma} \big) \varphi_{\epsilon}      \|_{\mathbb{B}(T_n)} \\ & \geq   \uptau_{T_n}^{-1} \| \big( {\varphi}_{X}   -   {\varphi}_{P,\Sigma} \big)  \}   \|_{\mathbb{B}(T_n)}.
\end{align*}
Cauchy-Schwarz implies that \begin{align*}
    \int_{\mathbb{B}(T_n)} \left|  \varphi_{X} (t) - \varphi_{P,\Sigma} (t)    \right| dt \leq   T_n^{d/2}   \| \varphi_{X}    -  \varphi_{P,\Sigma}     \|_{\mathbb{B}(T_n)}.
\end{align*}
Next, we examine the bias from truncating the $L^1$ norm to the set $\mathbb{B}(T_n)$. Note that there exists a finite constant $c > 0$ such that the function $t \rightarrow t^{p-d/2} (\log t)^{-1}   $ is monotonic over the set $\{ t \in \R : t   \geq c  \}$. Since $T_n \geq c$ for all sufficiently large $n$, the inequality $\|t \|_{\infty} \geq T_n $ implies that $ \| t \|_{\infty}^{p-d/2} (\log \| t \|_{\infty})^{-1} \geq T_n^{p-d/2} (\log T_n)^{-1}  $. It follows that
\begin{align*}
    \int_{\|t \|_{\infty} \geq T_n} \left| \varphi_X (t) \right| dt \leq \frac{\log(T_n)}{T_n^{p-d/2}} \int_{\|t \|_{\infty} \geq T_n} \frac{\|t \|_{\infty}^{p-d/2}}{\log(\| t \|_{\infty})} \left| \varphi_X(t) \right|.
\end{align*}
Since $f_X \in \mathbf{H}^p$, Cauchy-Schwarz yields \begin{align*}
    \int_{\|t \|_{\infty} \geq T_n} \frac{\|t \|_{\infty}^{p-d/2}}{\log(\| t \|_{\infty})} \left| \varphi_X(t) \right| \leq  \bigg( \int_{\R^d} \|t \|_{\infty}^{2p} \left|  \varphi_X(t)  \right|^2 dt  \bigg)^{1/2} \bigg( \int_{\|t \|_{\infty} \geq T_n}  \frac{1}{\|t \|_{\infty}^d ( \log(\| t \|_{\infty}) )^2 } dt   \bigg)^{1/2} \leq D.
\end{align*}
Next, we examine the bias for $\varphi_{P,\Sigma}$. Suppose $ \|  \Sigma^{-1}  \| \leq M^2  \sigma_n^2 (n \epsilon_n^2)^{1/\kappa}  $ holds. It follows that there exists a $c > 0$ for which $  \lambda_1(\Sigma ) \geq c (n \epsilon_n^2)^{-1/\kappa} \sigma_n^{-2} $ holds. It follows that there exists universal constants $C > 0$ such that
\begin{align*}
    \|   {\varphi}_{P,\Sigma} \mathbbm{1} \{ \| t \|_{\infty} > T_n  \}   \|_{L^1} & \leq  \int_{\| t \|_{\infty} > T_n} e^{- t' \Sigma t/2} dt \\ & \leq \int_{\| t \|_{\infty} > T_n} e^{- c \| t \|^2  \sigma_n^{-2} (n \epsilon_n^2)^{-1/\kappa}   } dt \\ & \leq D   \sigma_n^{-2} (n \epsilon_n^2)^{-1/\kappa}   e^{- C T_n^2 \sigma_n^{-2} (n \epsilon_n^2)^{-1/\kappa} } T_n^{d-2}    .
\end{align*}
From the definition of $\sigma_n^2$, we have  $T_n^2 (n \epsilon_n^2)^{-1/\kappa} \sigma_n^{-2} \asymp  (\log n) ( \log \log n )  $. Hence, the preceding bound reduces to $o(n^{-1/2})$. From combining the preceding bounds (and noting that  $ T_n^{-p} \lessapprox \uptau_{T_n} \epsilon_n $), it follows that the modulus satisfies
$$ \omega(d, \mathcal{H}_n, \epsilon_n) \leq D  \uptau_{T_n} T_n^{d/2} \epsilon_n.  $$
The claim follows from observing that \begin{align*}
   T_n^{d/2} \uptau_{T_n} \epsilon_n \asymp T_n^{\zeta+d/2} \epsilon_n \asymp n^{\frac{-p+d/2}{2(p + \zeta) + d}}  (\log n)^{(\lambda + \zeta + d)/2 }.
\end{align*}

\end{proof}

\begin{proof}[Proof of Theorem \ref{repmeas-1}] The proof proceeds through several steps which we outline below. Define the sequences \begin{align*}
    \sigma_n = \bigg( \frac{1}{n T_n^{3}}  \bigg)^{\frac{1}{6}} \; \; , \; \; \epsilon_n^2 = n^{-1}\sigma_n^{-2  }.
\end{align*}
\begin{enumerate}
    \item[\textbf{$(i)$}] First, we aim to apply Theorem \ref{main-contract}. We proceed by verifying that Assumptions \ref{sampling-uncert} - \ref{loc-conc} hold. \begin{align*}
    & \mathcal{G}(\mathbb{P}_{n,Y}, \varphi_{F})  - \mathcal{G}(\mathbb{P}_{Y}, \varphi_{F})  = \partial_t \log\varphi_F(t)(\E_n - \E)[e^{\mathbf{i}t' Y_2}]  -  (\E_n - \E)[\mathbf{i} Y_1 e^{\mathbf{i} t Y_2}].
\end{align*}

Fix  $\epsilon > 0$ sufficiently small. Let $C,C' > 0$ be such that $F_X( t \in \R :   \left| t \right|  > z) \leq C \exp(- C' z^{\chi})$. There exists a universal constant $R >0$ such that $\exp(- C' z^{\chi}) \leq D \epsilon $ for every $ z  \geq R (\log \epsilon^{-1})^{1/\chi}$. In particular, the interval $I =[-R  (\log \epsilon^{-1})^{1/\chi},R  (\log \epsilon^{-1})^{1/\chi} ] $ satisfies $ 1- F_X(I) \leq D \epsilon$. Denote the probability measure induced from the restriction of $F_X$ to $I$ by $ \overline{F}_X $. Note that \begin{align*}
 \sup_{t \in \R} \left| \varphi_{F_X} (t) -  \varphi_{\overline{F}_X}(t) \right| =  \sup_{t \in \R} \left|  \int_{\R} e^{\mathbf{i} t' x} d(F_X - \overline{F}_X ) (x)     \right|   \leq \| F_X - \overline{F}_X   \|_{TV}  \leq 1- F_X(I) \leq D \epsilon.
\end{align*}
 For all sufficiently large $ M > 0$, the tail bound on $F_X$ implies
\begin{align*}
    \E[ \left| X  \right| \mathbbm{1} \{ \left| X\right| > M \}  ]  &  = \int_{0}^{\infty} \mathbb{P} \big(  \left| X  \right| > M, \left| X \right| > t      \big) dt \\ & \leq   M \mathbb{P} \big( \left| X \right| > M     \big)  + \int_{M}^{\infty} \mathbb{P} \big(  \left| X \right| > t \big) dt        \\ & \leq D M  \exp(- C' M^{\chi}) .
\end{align*}
With $M = R (\log \epsilon^{-1})^{1/\chi} $, the preceding bound implies
 \begin{align*}  \sup_{t \in \R} \left| \partial_t \varphi_{F_X}(t) -  \partial_t \varphi_{\overline{F}_X}(t)  \right| &  =  \sup_{t \in \R} \left|  \int_{\R} x e^{\mathbf{i} t' x} d(F_X - \overline{F}_X ) (x)  \right|    \\ & \leq \sup_{t \in \R}  \left|  \int_{ x \in I} x e^{\mathbf{i} t' x} d(F_X- \overline{F}_X)  \right| + \sup_{t \in \R} \left| \int_{x \notin I} x e^{\mathbf{i} t' x} d F_X   \right|  \\ & \leq D \bigg[ (\log \epsilon^{-1})^{1/\chi} \| F_X - \overline{F}_X   \|_{TV} +  (\log \epsilon^{-1})^{1/\chi} \epsilon      \bigg] \\  & \leq D (\log \epsilon^{-1})^{1/\chi} \epsilon.
\end{align*}
Let $K > 0$ be such that $ \sup_{\| t \|_{\infty} \leq T_n} \left| \partial_t \log \varphi_{F_X} (t) \right| \leq D T_n^K$. Since $ \inf_{\| t \|_{\infty} \leq T_n }  \left| \varphi_{F_X}(t) \right| \geq c \exp(-c' T_n^2 ) $ for some $c,c' > 0$ and $T_n^2 \lessapprox \log(n) $, it follows that there exists a $C_1 > 0$ such that  $ \inf_{\| t \|_{\infty} \leq T_n }  \left| \varphi_{F_X}(t) \right| \geq n^{-C_1}$. Let $L > 1$ be large enough such that the choice $ \epsilon  = \epsilon_n^L $ implies $ \overline{F}_X $  has support contained in the cube $I_n =[-E  (\log \epsilon_n^{-1})^{1/\chi},E (\log \epsilon_n^{-1})^{1/\chi} ] $ for some universal constant $E > 0$ 
and $  T_n^K n^{C_1} (\log \epsilon_n^{-L})^{1/\chi}  \epsilon_n^L \leq \epsilon_n^2  $. It follows that
\begin{align*}
   &  \sup_{ \| t \|_{\infty} \leq T_n} \left| \partial_t \log \varphi_{F_X}  (t) \right| \frac{\left|  \varphi_{F_X} (t) -  \varphi_{\overline{F}_X}(t)  \right|}{\left| \varphi_{F_X} (t) \right|}     \leq   D \epsilon_n^2   \\ &  \sup_{ \| t \|_{\infty}  \leq T_n} \frac{ \left| \partial_t \varphi_{F_X}(t) -  \partial_t \varphi_{\overline{F}_X}(t) \right| }{ \left|  \varphi_{F_X}(t)  \right|}  \leq D \epsilon_n^2 \\ & \sup_{ \| t \|_{\infty}  \leq T_n} \frac{ \left|  \varphi_{F_X}(t) -   \varphi_{\overline{F}_X}(t) \right| }{ \left|  \varphi_{F_X}(t)  \right|}  \leq D \epsilon_n^2
\end{align*}
Since $\epsilon_n \downarrow 0$, the preceding bound implies $  \left|   \varphi_{\overline{F}_X}(t)  \right|  \geq   \left| \varphi_{F_X}(t)    \right| - D \epsilon_n^2 \left| \varphi_{F_X}(t)    \right|  \geq 0.5 \left| \varphi_{F_X}(t)    \right|   $
for all sufficiently large $n$ and $\|t \|_{\infty} \leq T_n$. Hence \begin{align*}
   & \sup_{\| t \|_{\infty} \leq T_n}  \left| \partial_t \log  \varphi_{F_X}  (t) - \partial_t \log \varphi_{\overline{F}_X}(t)     \right|  \\ & \leq \sup_{\| t \|_{\infty} \leq T_n} \bigg(  \left|  \frac{\varphi_{\overline{F}_X}(t) - \varphi_{F_X}(t) }{ \varphi_{\overline{F}_X} (t) }    \right| \left|  \partial_t \log  \varphi_{F_X} (t)  \right| \: +  \: \left|  \frac{ \partial_t \varphi_{F_X} (t)  - \partial_t \varphi_{\overline{F}_X} (t)  }{\varphi_{\overline{F}_X}(t)  }     \right|        \bigg)  \\ & \leq  2 \sup_{\| t \|_{\infty} \leq T_n} \bigg(  \left|  \frac{\varphi_{\overline{F}_X}(t) - \varphi_{F_X}(t) }{ \varphi_{{F}_X} (t) }    \right| \left|  \partial_t \log  \varphi_{F_X} (t)  \right| \: +  \: \left|  \frac{ \partial_t  \varphi_{F_X} (t)  - \partial_t  \varphi_{\overline{F}_X} (t)  }{\varphi_{{F}_X}(t)  }     \right|     \bigg) \\ & \leq D \epsilon_n^2.
\end{align*}
Next, we show that $ \overline{F}_X  $ can be suitably approximated by a discrete measure. Let $\iota = \max \{ 1, 1/\chi + 1/2   \} $.  By Lemma \ref{aux5}, there exists a discrete measure $F' = \sum_{i=1}^N p_i \delta_{\mu_i} $ with at most $ N = D  (\log \epsilon_n^{-1})^{\iota} $ support points on $I_n$ such that
\begin{align*}
     \sup_{\| t \|_{\infty} \leq T_n} \left| \partial_t^r \varphi_{\overline{F}_X} (t) - \partial_t^r \varphi_{F'}(t)   \right| \leq T_n^{-K} n^{-C_1} \epsilon_n^2 \; \; \; \; \; \; \; \; r=0,1
\end{align*}
where $\partial_t^r$ denotes the $r^{th}$ derivative. From the final claim of Lemma \ref{aux5}, we can also assume without loss of generality that the support points satisfy $  \inf_{k \neq j} \| \mu_k - \mu_j   \| \geq \epsilon_n^{L_2}  $ for some $L_2 > 1$. Note that
\begin{align*}
      \sup_{\| t \|_{\infty} \leq T_n }  \left| \partial_t \log \varphi_{ \overline{F}_X}  (t)  \right|  & \leq \sup_{\| t \|_{\infty} \leq T_n}  \left| \partial_t \log  \varphi_{F_X}  (t) - \partial_t \log \varphi_{\overline{F}_X}(t)     \right| +  \sup_{\| t \|_{\infty} \leq T_n }  \left| \partial_t \log \varphi_{ F_X }  (t)  \right| \\ & \leq D \big( \epsilon_n^2 + T_n^K \big) \\ & \leq D T_n^K.
\end{align*}
The preceding bounds imply \begin{align*}
     &  \sup_{\| t \|_{\infty} \leq T_n}  \left|  \partial_t \log  \varphi_{ \overline{F}_X  } (t)  \right| \frac{ \left|  \varphi_{\overline{F}_X}(t) - \varphi_{F'}(t)  \right|}{ \left|  \varphi_{\overline{F}_X}(t)  \right|} \leq  D \epsilon_n^2 \; , \; \sup_{\| t \|_{\infty} \leq T_n} \frac{ \left| \partial_t \varphi_{\overline{F}_X} (t) - \partial_t \varphi_{F'}(t)   \right|}{\left| \varphi_{\overline{F}_X}(t)  \right|} \leq D \epsilon_n^2 .
\end{align*}
From an analogous argument to the bound for $F_X$ and $\overline{F}_X$, it follows that
 \begin{align*}
    \sup_{\| t \|_{\infty} \leq T_n}  \left| \partial_t \log  \varphi_{ \overline{F}_X }  (t) - \partial_t \log \varphi_{F'}(t)     \right| \leq D \epsilon_n^2.
\end{align*} 
Fix any $L_3 > L_2$  sufficiently large such that \begin{align*} T_n^{K+1} n^{C_1} \epsilon_n^{L_3} \leq \epsilon_n^2   \; \; , \; \; n^{C_1} \sqrt{n}  \epsilon_n \sqrt{\log n} \epsilon_n^{L_3/2}  \leq \epsilon_n^2  \; , \;   T_n n^{C_1} (\log \epsilon_n^{-1})^{1/\chi} \epsilon_n^{L_3} \leq \epsilon_n^2 .
\end{align*}
Define $V_i = \{ t \in I_{n} :  \| t - \mu_i  \| \leq  \epsilon_n^{L_3}  \} $ for $i=1,\dots,N$ and set $V_0 = \R \setminus \bigcup_{i=1}^N V_i$. From the definition of the $\{\mu_i \}_{i=1}^N$, it follows that $\{V_0, V_1, \dots , V_N  \}$ is a disjoint partition of $\R$. By Lemma \ref{aux3}, for any distribution $P$ that satisfies $  \int_{\R} \| x \|^2 dP(x) \leq n \epsilon_n^2 \log n  $ and $ \sum_{j=1}^N \left| P(V_j) - p_j   \right| \leq \epsilon_n^{L_3} $, we have (using a similar expansion as above)  \begin{align*}
    \sup_{\| t \|_{\infty} \leq T_n}  \left| \partial_t \log  \varphi_{ P }  (t) - \partial_t \log \varphi_{F'}(t)     \right| \leq D \epsilon_n^2 .
\end{align*} 
From combining all the preceding bounds, observe that all such $P$  also satisfy \begin{align*} &   
\sup_{\| t \|_{\infty} \leq T_n}  \left| \partial_t \log  \varphi_{ P }  (t) - \partial_t\log \varphi_{F_X}(t)     \right| \leq D \epsilon_n^2 
\\ & 
\sup_{\| t \|_{\infty} \leq T_n } \left|  \partial_t \log \varphi_P(t)   \right|   \leq D \epsilon_n^2  + \sup_{\| t \|_{\infty} \leq T_n } \left|  \partial_t \log \varphi_{F_X}(t)  \right| \leq D T_n^K.  \end{align*}
Given any $\sigma^2 > 0$, the preceding bound implies \begin{align*}
    \sup_{\| t \|_{\infty} \leq T_n } \left|  \partial_t \log \varphi_{P,\sigma^2}(t)  \right| & \leq \sup_{ \| t \|_{\infty} \leq T_n } \bigg[ \left|  \partial_t \log \varphi_{P}(t)   \right| + \left|  \partial_t \log \varphi_{\sigma^2}(t)   \right|   \bigg] \\ & \leq D  \bigg( T_n^K + T_n \sigma^2 \bigg).
\end{align*}
This is bounded by $D T_n^K $ if $\sigma^2 \leq T_n^{-1}$. By Lemma \ref{aux2}, we have   \begin{align*}
  &    \int_{\mathbb{B}(T_n)}   \left|  \widehat{\varphi}_{Y_2}(t) - \varphi_{Y_2}(t)    \right|^2 dt \leq D \frac{ T_n  \log(T_n)  }{n}  \; ,  \\ &  \int_{\mathbb{B}(T_n)}   \|  \E_n[\mathbf{i} Y_1 e^{\mathbf{i} t Y_2}] - \E[\mathbf{i} Y_1 e^{\mathbf{i} t Y_2}]   \|^2 dt \leq D \frac{  T_n \log(T_n)  }{n} 
\end{align*}
with $\mathbb{P}$ probability approaching $1$. From the preceding bounds, it follows that Assumption \ref{sampling-uncert} holds with the set $$ \mathcal{S}_n = \bigg \{ \phi_{P,\sigma^2} : \|  \sigma^2 - \sigma_n^2 \| \leq \epsilon_n^2 , \int_{\R} \| x \|^2 dP(x) \leq n \epsilon_n^2 \log(n) \: ,   \sum_{j=1}^N \left| P(V_j) - p_j   \right| \leq \epsilon_n^{L_3}   \bigg  \}.  $$

Next, we verify that Assumption \ref{weak-bias} holds. Let $F_n = \phi_{F', \sigma_n^2}$ where $F'$ is as specified above. Since $\log \varphi_{F_n} = \log \varphi_{F'} + \log \varphi_{\sigma_n^2}$, it follows that
\begin{align*}
    \mathcal{G}(\mathbb{P}_{Y}, \varphi_{F_n})  - \mathcal{G}(\mathbb{P}_{Y}, \varphi_{F_X})  & =   \varphi_{Y_2} \big( \partial_t \log \varphi_{F_n} - \partial_t \log \varphi_{F_X}    \big) \\ & =  \varphi_{Y_2} \big( \partial_t \log \varphi_{F'} - \partial_t \log \varphi_{F_X}    \big) + \varphi_{Y_2} \partial_t \log \varphi_{\sigma_n^2}.
\end{align*}
From the preceding bounds for $F'$ and the definition of $(\sigma_n^2, \epsilon_n , T_n)$, we obtain
\begin{align*}
   \|  \mathcal{G}(\mathbb{P}_{Y}, \varphi_{F_n})  - \mathcal{G}(\mathbb{P}_{Y}, \varphi_{F_X}) \|_{\mathbb{B}(T_n)} & \leq  \| \partial_t \log \varphi_{F'} - \partial_t \log \varphi_{F_X}  \|_{\mathbb{B}(T_n)} + \| \partial_t \log \varphi_{\sigma_n^2}  \|_{\mathbb{B}(T_n)} \\ & \leq D \big[\epsilon_n^2 T_n^{1/2}  + \sigma_n^2 T_n^{3/2}  \big] \\ & \leq D \epsilon_n.
\end{align*}
 It follows that Assumption \ref{weak-bias} holds with $F_n = \phi_{F',\sigma_n^2}$. For Assumption \ref{loc-conc}, we define $\mathcal{R}_n$ to be the set   \begin{align*}
   \mathcal{R}_n = \bigg \{  \phi_{P,\sigma^2} : \|  \sigma^2 - \sigma_n^2 \| \leq \epsilon_n^2  \: ,   \sum_{j=1}^N \left| P(V_j) - p_j   \right| \leq \epsilon_n^{L_3}   \bigg  \} .
\end{align*}
Note that for any $ \phi_{P,\sigma^2} \in \mathcal{R}_n$, we have \begin{align*}
    \|  \mathcal{G}(\mathbb{P}_{Y}, \varphi_{F_n})  - \mathcal{G}(\mathbb{P}_{Y}, \varphi_{P,\sigma^2}) \|_{\mathbb{B}(T_n)} & \leq  \| \partial_t \log \varphi_{F'} - \partial_t \log \varphi_{P}  \|_{\mathbb{B}(T_n)} + \| \partial_t \log \varphi_{\sigma^2} - \partial_t \log \varphi_{\sigma_n^2}  \|_{\mathbb{B}(T_n)}   \\ & \leq D \big[\epsilon_n^2 T_n^{1/2}  + \epsilon_n^2 T_n^{3/2}  \big] \\ & \leq D \epsilon_n.
\end{align*}
Furthermore, we have that \begin{align*}
    \int_{ \mathcal{R}_n \setminus \mathcal{S}_n }  d \nu_{\alpha,G} (P,\sigma^2) &  \leq \int_{ P : \int_{\R} \| x \|^2 dP(x) > n \epsilon_n^2 \log(n)  }  d \nu_{\alpha,G} (P,\sigma^2) \\ & \leq  \int_{ P : \int_{\R} \| x \|^2 dP(x) > n \epsilon_n^2 \log(n)  }  d \text{DP}_{\alpha}(P) \\ &  \leq \exp(-D n \epsilon_n^2 \log n) \; ,
\end{align*}
where the second inequality is due to $\nu_{\alpha,G}$ being a product measure $\nu_{\alpha,G} = \text{DP}_{\alpha} \otimes G$ and the third inequality follows from an application of Lemma \ref{aux12}. Assumption \ref{loc-conc} $(ii)$ follows.

To finish verifying Assumption \ref{loc-conc} $(i)$, note that \begin{align*}
    \int_{\mathcal{R}_n} d \nu_{\alpha,G} (P,\sigma^2) = \int_{ \sigma^2 : \| \sigma^2 - \sigma_0^2  \| \leq  \epsilon_n^2} \int_{P : \sum_{j=1}^N \left| P(V_j) - p_j   \right| \leq \epsilon_n^{L_3}    } d  \text{DP}_{\alpha}(P)  d G(\sigma^2).
\end{align*}
As $ \text{DP}_{\alpha} $ is constructed using a Gaussian base measure $\alpha $, it is straightforward to verify that $ 
\inf_{j=1}^N \alpha(V_j) \geq C \epsilon_n^{L_3 } \exp(- C' ( \log \epsilon_n^{-1})^{2/\chi}  ) $ for universal constants $C,C' > 0$. By definition of $ \text{DP}_{\alpha}$, $(P(V_1),\dots, P(V_N)) \sim \text{Dir}(N, \alpha(V_1), \dots,  \alpha(V_N))$. As $N = D \{ \log ( \epsilon_n^{-1} )  \}^{\iota}$, an application of \citep[Lemma G.13]{ghosal2017fundamentals} implies  \begin{align*}
    \int_{P : \sum_{j=1}^N \left| P(V_j) - p_j   \right| \leq \epsilon_n    } d  \text{DP}_{\alpha}(P) \geq C \exp \big (  - C'  ( \log \epsilon_n^{-1} )^{\iota + \max \{2/ \chi,1   \}  }   \big)  & = C \exp \big (  - C'  ( \log \epsilon_n^{-1} )^{\lambda  }   \big) \\ & \geq  C \exp \big (  - C''  n \epsilon_n^2   \big) .
\end{align*}
It remains to bound the outer integral. By Assumption \ref{cov-prior}, there exists universal constant $C,C',C'' > 0$ such that \begin{align*}
     \int_{ \sigma^2 : \| \sigma^2 - \sigma_n^2  \| \leq \epsilon_n^2  } d G(\sigma^2)  \geq  C \exp \big( - C' \sigma_n^{-2 \kappa}    \big) \geq C \exp(-C' \sigma_n^{-2})    \geq  C \exp ( -C'' n \epsilon_n^2  ).
\end{align*}
Assumption \ref{loc-conc} follows. By Theorem \ref{main-contract}, we obtain \begin{align*}
    \nu \bigg ( F : \|  \partial_t \log\varphi_F(t)\E_n[e^{\mathbf{i}t' Y_2}]  -  \E_n[\mathbf{i} Y_1 e^{\mathbf{i} t Y_2}]  \|_{\mathbb{B}(T_n)}  > D \epsilon_n  \: \bigg| \: \mathcal{D}_n  \bigg ) = o_{\mathbb{P}}(1).
\end{align*}

\item[\textbf{$(ii)$}] 
Observe that $ \epsilon_n^2 \asymp n^{-2/3} T_n $. We aim to apply Corollary \ref{contract-strong} with the metric $d = \mathbf{W}_1$. Given any $\delta > 0$ and $s > 1$, define the set \begin{align*}
    \mathcal{H}_n = \bigg \{ F = \phi_{P, \sigma^2} :  \sigma^{\frac{(s-1)}{1+2s} + 1} \leq \epsilon_n^{\frac{-2(s-1)}{1+2s}} n^{- \delta}  \: , \: 
 \bigg( \int \| x \|^s dP(x) \bigg)^{\frac{(s-1)}{s(1+2s)} + 1/s} \leq  \epsilon_n^{\frac{-2(s-1)}{1+2s}} n^{-\delta}     \bigg \}.
\end{align*}
By an analogous argument to Theorem \ref{deconv-1}, we can pick $s > 0$ large enough and $\delta > 0$ small enough such that $\nu(\mathcal{H}_n^c)  \lessapprox \exp(-B_n)$ for some $n \epsilon_n^2 = o(B_n) $, so that the conditions of  Corollary \ref{contract-strong} are satisfied. Let $K: \R \rightarrow (0, \infty)$ denote any symmetric density function on $\R$ with Fourier transform $\mathcal{F}(K)$  that is bounded, has support contained in $[-1,1]$ and has finite moments of order $s$. Define $K_R(x) = R K(Rx)$ for every $R > 0$. By an analogous argument to Theorem \ref{deconv-1}, we obtain
\begin{align*}
    \mathbf{W}_1(F,F_X) \leq D  \bigg[ \epsilon_n^{\frac{-2(s-1)}{1+2s}} n^{- \delta}  \| \varphi_{F \star K_R} - \varphi_{F_X \star K_R}    \|_{L^2}^{\frac{2(s-1)}{1+2s}} + R^{-1}   \bigg].
\end{align*}
for every $F \in \mathcal{H}_n$. The fundemental theorem of calculus, Cauchy-Schwarz and the initial value condition $ \partial_t \log \varphi_X (0) = \partial_t \log \varphi_{F}(0) = 0 $  imply that   \begin{align*}
    \left|\log \varphi_{F_X} (t) - \log \varphi_{F}(t)           \right| & = \left| \int_0^t  \big[  \partial_s \log \varphi_{X} (s) - \partial_s \log \varphi_{F}(s)     ]  ds  \right| \\ & \leq \sqrt{R_n} \| \partial \log  \varphi_{X}  - \partial \log \varphi_{F}      \|_{\mathbb{B}(R_n)} 
\end{align*}
holds for every $t \in \mathbb{B}(R_n)$. Furthermore, for every fixed $t \in \R$, the mean value theorem implies $
   \left|  \varphi_{X}(t) -   \varphi_{F}(t)   \right|  \leq  \left|  \log \varphi_{X} (t) - \log \varphi_{F}(t)    \right|$. In particular $$
    \|   \varphi_{X} -  \varphi_{F}  \|_{\mathbb{B}(R_n)} \leq D R_n \| \partial_t \log  \varphi_{X}  - \partial_t \log \varphi_{F}      \|_{\mathbb{B}(R_n)} $$
for any distribution $F$. By assumption, we have $\inf_{\| t \|_{\infty} \leq R} \left| \varphi_{Y_2}   \right| \geq b \exp(-B R^{2})   $ for some constants $b,B> 0$. Let $R = R_n =  \min \{ c_0 (\log n)^{1/2} , T_n \} $ for a constant $c_0$ sufficiently small such that $ 2 B c_0^{2} \max \{ (s-1) (1+2s)^{-1},1 \} = \xi < \delta < 1/2 $. An application of Lemma \ref{aux2} implies that \begin{align*}
   \sup_{\| t \|_{\infty} \leq R_n}   \frac{  \left| \widehat{\varphi}_{Y_2}(t)  - \varphi_{Y_2}(t)    \right| }{\left| \varphi_{Y_2}(t) \right|} \leq D n^{-1/2 + \xi} \sqrt{\log \log  n}
\end{align*}
with $\mathbb{P}$ probability approaching $1$. As the quantity on the right converges to zero, it follows that $
     0.5 \left| \varphi_{Y_2}(t) \right| \leq \left| \widehat{\varphi}_{Y_2}(t)   \right| \leq 1.5\left| \varphi_{Y_2}(t) \right|$
for every $t \in \mathbb{B}(R_n)$. By combining the preceding bounds, it follows that \begin{align*}
    \| \varphi_{F \star K_{R_n}} - \varphi_{F_X \star K_{R_n}}    \|_{L^2}^2 & = \int_{\mathbb{B}(R_n)} \left| \varphi_F(t) - \varphi_X(t)   \right|^2 \left| \mathcal{F}(K)(t)   \right|^2 dt \\ & \leq D \int_{\mathbb{B}(R_n)} \left| \varphi_F(t) - \varphi_X(t)   \right|^2  dt \\ & \leq D R_n^2 \| \partial_t \log  \varphi_{X}  - \partial_t \log \varphi_{F}      \|_{\mathbb{B}(R_n)}^2 \\ & \leq D R_n^2 \sup_{\| t \|_{\infty} \leq R_n} \left| \varphi_{Y_2}(t)  \right|^{-2} \| \widehat{\varphi}_{Y_2} (\partial_t \log  \varphi_{X}  - \partial_t \log \varphi_{F} )     \|_{\mathbb{B}(R_n)}^2 \\ & \leq D R_n^2  \exp(2B R_n^2)   \| \widehat{\varphi}_{Y_2} (\partial_t \log  \varphi_{X}  - \partial_t \log \varphi_{F} )     \|_{\mathbb{B}(R_n)}^2 
\end{align*}
From the definition of $R_n$, it follows that  \begin{align*}
    W_1(F,F_X) \leq D \big[ R_n^{\frac{2(s-1)}{1+2s}}  \epsilon_n^{\frac{-2(s-1)}{1+2s}} n^{- \delta + \xi} \| \widehat{\varphi}_{Y_2} (\partial_t \log  \varphi_{X}  - \partial_t \log \varphi_{F} )     \|_{\mathbb{B}(T_n)}^{\frac{2(s-1)}{1+2s}}  + R_n^{-1}    ] .
\end{align*}
 Suppose $F$ is such that $\|  \partial_t \log\varphi_F(t)\E_n[e^{\mathbf{i}t' Y_2}]  -  \E_n[\mathbf{i} Y_1 e^{\mathbf{i} t Y_2}]  \|_{\mathbb{B}(T_n)} \leq D \epsilon_n$. Since $\E[\mathbf{i} Y_1 e^{\mathbf{i} t Y_2}] = \varphi_{Y_2}(t) \partial_t \log \varphi_X$ and $ \sup_{t \in \mathbb{B}(T_n)} \left| \partial_t \log \varphi_X (t) \right| \lessapprox T_n^K$, a straightforward application of Lemma \ref{aux2} implies that $\| \widehat{\varphi}_{Y_2} (\partial_t \log  \varphi_{X}  - \partial_t \log \varphi_{F} )     \|_{\mathbb{B}(T_n)} \leq D \epsilon_n$. From the preceding bounds, it follows that the modulus satisfies $$ \omega(d,\mathcal{H}_n,\epsilon_n) \leq D \bigg[ R_n^{\frac{2(s-1)}{1+2s}} n^{- \delta + \xi} + R_n^{-1}  \bigg] \leq D R_n^{-1}. $$
Since $ R_n =  \min \{ c_0 (\log n)^{1/2} , T_n \} $, the claim follows.
 \end{enumerate}
\end{proof}

\begin{proof}[Proof of Theorem \ref{deconv-prod-t2}]
The proof proceeds through several steps. Let $\epsilon_n^2 = n^{-1} (\log n)^{\lambda}$.
\begin{enumerate}
    \item[\textbf{$(i)$}] We aim to apply Theorem \ref{main-contract}. The verification of Assumptions \ref{sampling-uncert} - \ref{loc-conc}  is completely analogous to the argument in Theorem \ref{repmeas-1} and omitted. Specifically, Assumptions \ref{sampling-uncert} and \ref{loc-conc} hold with the sets \begin{align*} & \mathcal{S}_n = \bigg \{ \phi_{P,\sigma^2} : \|  \sigma^2 - \sigma_0^2 \| \leq \epsilon_n^{L_3} , \int_{\R} \| x \|^2 dP(x) \leq n \epsilon_n^2 \log(n) \: ,   \sum_{j=1}^N \left| P(V_j) - p_j   \right| \leq \epsilon_n^{L_3}   \bigg  \} , \\ & \mathcal{R}_n = \bigg \{  \phi_{P,\sigma^2} : \|  \sigma^2 - \sigma_0^2 \| \leq \epsilon_n^{L_3}  \: ,   \sum_{j=1}^N \left| P(V_j) - p_j   \right| \leq \epsilon_n^{L_3}   \bigg  \} ,
    \end{align*}
for some fixed constant $ L_3 > 0$ sufficiently large, $N = D(\log \epsilon_n^{-1})^{\max \{ 1, \chi^{-1} +1/2  \}}$, and sets $\{V_0,\dots,V_N \}$  that form a disjoint partition of $\R$. By Theorem \ref{main-contract}, we obtain \begin{align*}
    \nu \bigg ( F : \|  \partial_t \log\varphi_F(t)\E_n[e^{\mathbf{i}t' Y_2}]  -  \E_n[\mathbf{i} Y_1 e^{\mathbf{i} t Y_2}]  \|_{\mathbb{B}(T_n)}  > D \epsilon_n  \: \bigg| \: \mathcal{D}_n  \bigg ) = o_{\mathbb{P}}(1).
\end{align*}
\item[\textbf{$(ii)$}] 
We aim to apply Corollary \ref{contract-strong} with the metric $d = \| . \|_{L^2}$. The law of  $G = G_n$ is given by $\Omega / \sigma_n^2$ where $\Omega \sim L$  and $L$ is a probability measure on  $\mathbf{S}_+^d$ that  satisfies Assumption \ref{cov-prior}. By Assumption \ref{cov-prior}, it follows that for every $E' > 0$, there exists $E > 0$ such that
\begin{align*}
  \int_{\sigma^2 : \sigma^{-2} > E \sigma_n^2 (n \epsilon_n^2)^{1/ \kappa}  }   d G(\sigma^2) = \int_{\sigma^2 : \sigma^{-2} > E (n \epsilon_n^2)^{1/ \kappa}  }   d L(\sigma^2) \leq \exp \big(  - E' n \epsilon_n^2  \big).
\end{align*}
In particular, we can choose $E > 0$ large enough such that the hypothesis of Corollary \ref{main-contract} holds with the set of distributions $ \mathcal{H}_n = \{ \phi_{P,\sigma^2} :  \sigma^{-2} \leq E \sigma_n^2 (n \epsilon_n^2)^{1/ \kappa}    \}. $ It remains to compute the modulus  over $\mathcal{H}_n$. By assumption, we have $\inf_{\| t \|_{\infty} \leq R} \left| \varphi_{Y_2}   \right| \geq b \exp(-B R^{2})   $ for some constants $b,B> 0$. Fix any $R = R_n \leq T_n $. As in the proof of Theorem \ref{repmeas-1}, we can deduce  \begin{align*}
   \int_{\mathbb{B}(R_n)} \left| \varphi_F(t) - \varphi_X(t)   \right|^2  dt  & \leq D R_n^2 \| \partial_t \log  \varphi_{X}  - \partial_t \log \varphi_{F}      \|_{\mathbb{B}(R_n)}^2 \\ & \leq D R_n^2 \sup_{\| t \|_{\infty} \leq R_n} \left| \varphi_{Y_2}(t)  \right|^{-2} \| \widehat{\varphi}_{Y_2} (\partial_t \log  \varphi_{X}  - \partial_t \log \varphi_{F} )     \|_{\mathbb{B}(R_n)}^2  \\ & \leq D R_n^2  \exp(2B R_n^2)   \| \widehat{\varphi}_{Y_2} (\partial_t \log  \varphi_{X}  - \partial_t \log \varphi_{F} )     \|_{\mathbb{B}(T_n)}^2.
\end{align*}
for any $F$. Suppose $F$ is such that $\|  \partial_t \log\varphi_F(t)\E_n[e^{\mathbf{i}t' Y_2}]  -  \E_n[\mathbf{i} Y_1 e^{\mathbf{i} t Y_2}]  \|_{\mathbb{B}(T_n)} \leq D \epsilon_n$. Since $\E[\mathbf{i} Y_1 e^{\mathbf{i} t Y_2}] = \varphi_{Y_2}(t) \partial_t \log \varphi_X$ and $ \sup_{t \in \mathbb{B}(T_n)} \left| \partial_t \log \varphi_X (t) \right| \lessapprox T_n$, a straightforward application of Lemma \ref{aux2} implies that $\| \widehat{\varphi}_{Y_2} (\partial_t \log  \varphi_{X}  - \partial_t \log \varphi_{F} )     \|_{\mathbb{B}(T_n)} \leq D \epsilon_n$. Next, we examine the bias from truncating the $L^2$ norm to the set $\mathbb{B}(R)$. Suppose $ \sigma^{-2} \leq E^2  \sigma_n^2 (n \epsilon_n^2)^{1/\kappa}  $ holds. Since $\sigma_n^{2} \asymp (n \epsilon_n^2)^{-1/\kappa}$, it follows that there exists a $c > 0$ for which $  \sigma^2 \geq c  $ holds. It follows that there exists  universal constants $C,D > 0$ such that
\begin{align*}
    \| \big( {\varphi}_{X}   -   {\varphi}_{P,\sigma^2} \big) \mathbbm{1} \{ \| t \|_{\infty} > R  \}   \|_{L^2}^2 & \leq  2 \|  {\varphi}_{X}    \mathbbm{1} \{ \| t \|_{\infty} > R  \}   \|_{L^2}^2 +  2 \|     {\varphi}_{P,\sigma^2}  \mathbbm{1} \{ \| t \|_{\infty} > R  \}   \|_{L^2}^2 \\ & \leq 2 \int_{\| t \|_{\infty} > R}  e^{- t^2 \sigma_0^2} dt + 2 \int_{\| t \|_{\infty} > R} e^{- t^2 \sigma^2} dt \\ & \leq 2 \int_{\| t \|_{\infty} > R}  e^{- t^2 \sigma_0^2} dt + 2 \int_{\| t \|_{\infty} > R} e^{- c t^2     } dt \\ & \leq D   e^{- C R^2} R^{-1}  .
\end{align*}
It follows that for any sequence $R_n \leq T_n$, the modulus satisfies \begin{align*}
    \omega^2(d,\mathcal{H}_n,\epsilon_n) \leq D \big[ R_n^2  \exp(2B R_n^2) \epsilon_n^2 +  e^{- C R^2} R^{-1}  ].
\end{align*}
The claim follows by picking $ R_n =  \min \{ c_0 (\log n)^{1/2} , T_n \} $ for a sufficiently small $c_0 > 0$.

    \end{enumerate}

\end{proof}

\begin{proof}[Proof of Theorem \ref{multfac-1}] The proof proceeds through several steps which we outline below. For simplicity, we refer to $F_{X_k}$ as $F_X$. Define the sequences \begin{align*}
    \sigma_n = \bigg( \frac{1}{n T_n}  \bigg)^{\frac{1}{6}} \; \; , \; \; \epsilon_n^2 = n^{-1}\sigma_n^{-2  }.
\end{align*}
\begin{enumerate}
    \item[\textbf{$(i)$}] First, we aim to apply Theorem \ref{main-contract}. We proceed by verifying that Assumptions \ref{sampling-uncert} - \ref{loc-conc} hold. For every $F \in \mathcal{F}_n$, we have  \begin{align*}
    &  \G(\mathbb{P}_{n,Y},\varphi_{F}) -\G(\mathbb{P}_Y,\varphi_{F}) =  \mathbf{Q}_k^*\big[\widehat{\varphi}_{\mathbf{Y}}^2     \widehat{\mathcal{V}}_{\mathbf{Y}} - {\varphi}_{\mathbf{Y}}^2     {\mathcal{V}}_{\mathbf{Y}}  \big] +  \big[{\varphi}_{\mathbf{Y}}^2 - \widehat{\varphi}_{\mathbf{Y}}^2  \big](\log \varphi_{F_{X}})''.
\end{align*} 

Note that $\left|t' \mathbf{A}_k \right| \leq \sqrt{d} \| t \|_{\infty} \| \mathbf{A}_k \| \leq \sqrt{d}T_n \| \mathbf{A}_k \|   $ uniformly over $t \in \mathbb{B}(T_n)$. Since $\| \mathbf{A}_k \| < \infty$, it suffices to work under the setting where $ \sup_{t \in \mathbb{B}(T_n)} \left| t' \mathbf{A}_k \right| \leq D T_n$ . 

Fix  $\epsilon > 0$ sufficiently small. Let $C,C' > 0$ be such that $F_X( t \in \R :   \left| t \right|  > z) \leq C \exp(- C' z^{\chi})$. There exists a universal constant $R >0$ such that $\exp(- C' z^{\chi}) \leq D \epsilon $ for every $ z  \geq R (\log \epsilon^{-1})^{1/\chi}$. In particular, the interval $I =[-R  (\log \epsilon^{-1})^{1/\chi},R  (\log \epsilon^{-1})^{1/\chi} ] $ satisfies $ 1- F_X(I) \leq D \epsilon$. Denote the probability measure induced from the restriction of $F_X$ to $I$ by $ \overline{F}_X $. Note that \begin{align*}
 \sup_{t \in \R} \left| \varphi_{F_X} (t) -  \varphi_{\overline{F}_X}(t) \right| =  \sup_{t \in \R} \left|  \int_{\R} e^{\mathbf{i} t' x} d(F_X - \overline{F}_X ) (x)     \right|   \leq \| F_X - \overline{F}_X   \|_{TV}  \leq 1- F_X(I) \leq D \epsilon.
\end{align*}
For all sufficiently large $M$ and $  r=(1,2) $, the tail bound on $F_X$ implies
\begin{align*}
    \E[ \left|X \right|^r \mathbbm{1} \{  \left|X \right| > M  \}  ]  &  = \int_{0}^{\infty} \mathbb{P} \big(  \left|X \right|^r > M^r, \left|X \right|^r > t      \big) dt \\ & \leq   M^r \mathbb{P} \big( \left|X \right| > M     \big)  + \int_{M^r}^{\infty} \mathbb{P} \big(  X^r > t \big) dt       \\  & \leq D  M^r \exp(- C' M^{\chi})   .
\end{align*}
With $M = R (\log \epsilon^{-1})^{1/\chi} $ and $r=(1,2)$, the preceding bound implies
 \begin{align*}  \sup_{t \in \R} \left|  \partial_{t}^r  \varphi_{F_X}(t) -  \partial_{t}^r \varphi_{\overline{F}_X}(t) \right| &  =  \sup_{t \in \R}  \left|  \int_{\R} x^r e^{\mathbf{i} t' x} d(F_X - \overline{F}_X ) (x)     \right|  \\ & \leq \sup_{t \in \R}  \left|  \int_{ x \in I} x^r e^{\mathbf{i} t' x} d(F_X- \overline{F}_X)  \right| + \sup_{t \in \R}  \left| \int_{x \notin I} x^r e^{\mathbf{i} t' x} d F_X    \right|  \\ & \leq D \bigg[ (\log \epsilon^{-1})^{r/\chi} \| F_X - \overline{F}_X   \|_{TV} +  (\log \epsilon^{-1})^{r/\chi} \epsilon      \bigg] \\  & \leq D (\log \epsilon^{-1})^{r/\chi} \epsilon.
\end{align*}
We can write $(\log \varphi_{F_X})''  -  (\log \varphi_{ \overline{F}_X})'' $ as \begin{align*}
   &  (\log \varphi_{F_X})'' -  (\log \varphi_{ \overline{F}_X})''  \\ & = (\log \varphi_{F_X})'' \frac{( \varphi_{\overline{F}_X}^2 - \varphi_{F_X}^2    )}{ \varphi_{\overline{F}_X}^2  } +  \frac{  \varphi_{F_X} (  \varphi_{F_X}''  - \varphi_{\overline{F}_X}''    )      }{ \varphi_{\overline{F}_X}^2 }  +   \frac{  \varphi_{\overline{F}_X}'' (  \varphi_{F_X} - \varphi_{\overline{F}_X}    )        }{ \varphi_{\overline{F}_X}^2 }        + \frac{ (\varphi_{\overline{F}_X}')^2  - (\varphi_{{F}_X}')^2   }{\varphi_{\overline{F}_X}^2}.
\end{align*}
Let $K > 0$ be such that $ \sup_{\| t \|_{\infty} \leq T_n} \left| ( \log \varphi_{F_X})'' (t) \right| \leq D T_n^K$. Since $ \inf_{ \|t \| \leq D T_n }  \left| \varphi_{F_X}(t) \right| \geq c \exp(-c' T_n^2 ) $ for some $c,c' > 0$ and $T_n^2 \lessapprox \log(n) $, there exists a $C_1 > 0$ such that  $ \inf_{ \left| t \right| \leq DT_n }  \left| \varphi_{F_X}(t) \right| \geq n^{-C_1}$. Let $L > 1$ be large enough such that the choice $ \epsilon  = \epsilon_n^L $ implies $ \overline{F}_X $  has support contained in the cube $I_n =[-E  (\log \epsilon_n^{-1})^{1/\chi},E (\log \epsilon_n^{-1})^{1/\chi} ] $ for some universal constant $E > 0$ 
and $ T_n^K n^{2 C_1} (\log \epsilon_{n}^{-L})^{2/\chi} \epsilon_n^L \leq \epsilon_n^2$. The first consequence of this is  \begin{align*}
   \sup_{ \| t \|_{\infty}  \leq DT_n} \frac{ \left|  \varphi_{F_X}(t) -   \varphi_{\overline{F}_X}(t) \right| }{ \left|  \varphi_{F_X}(t)  \right|}  \leq D \epsilon_n^2
\end{align*}
Since $\epsilon_n \downarrow 0$, the preceding bound implies $  \left|   \varphi_{\overline{F}_X}(t)  \right|  \geq   \left| \varphi_{F_X}(t)    \right| - D \epsilon_n^2 \left| \varphi_{F_X}(t)    \right|  \geq 0.5 \left| \varphi_{F_X}(t)    \right|   $
for all sufficiently large $n$ and $\|t \|_{\infty} \leq D T_n$. Since $ \sup_{\| t \|_{\infty} \leq T_n
}\left| (\log \varphi_{F_X})'' (t)  \right|  \leq D T_n^K $ and $|  \varphi_{F_X}'(t)  | , |\varphi_{F_X}''(t)  | $ are bounded by a universal constant $C$ (which only depends on the first two moments of $F_X$), the preceding expression for $(\log \varphi_{F_X})''  -  (\log \varphi_{ \overline{F}_X})''$ implies that \begin{align*}
    \sup_{ \| t  \|_{\infty} \leq D T_n}  \left| (\log \varphi_{F_X})'' (t)  -  (\log \varphi_{ \overline{F}_X})'' (t) \right| \leq D T_n^K n^{2 C_1} (\log \epsilon_{n}^{-L})^{2/\chi} \epsilon_n^L \leq D\epsilon_n^2.
\end{align*}
Next, we show that $ \overline{F}_X  $ can be suitably approximated by a discrete measure. Let $\iota = \max \{ 1, 1/\chi + 1/2   \} $.  By Lemma \ref{aux5}, there exists a discrete measure $F' = \sum_{i=1}^N p_i \delta_{\mu_i} $ with at most $ N = D  (\log \epsilon_n^{-1})^{\iota} $ support points on $I_n$ such that \begin{align*}
        \sup_{ \left| t \right| \leq DT_n } \left| \partial_{t}^r  \varphi_{\overline{F}_X}(t) - \partial_{t}^r \varphi_{F'}(t)  \right| \leq T_n^{-K}    n^{-2C_1} \epsilon_n^2   \; \; \; \; \; \; \; r=0,1,2.
\end{align*}
where $\partial_t^r$ denotes the $r^{th}$ derivative. From the final claim of Lemma \ref{aux5}, we can also assume without loss of generality that the support points satisfy $  \inf_{k \neq j} \| \mu_k - \mu_j   \| \geq \epsilon_n^{L_2}  $ for some $L_2 > 1$. Note that
\begin{align*}
      \sup_{\left| t \right| \leq DT_n }  \left|  (\log \varphi_{ \overline{F}_X})''  (t)  \right|  & \leq \sup_{\left| t \right| \leq D T_n}  \left| (\log  \varphi_{F_X} )''  (t) -  (\log \varphi_{\overline{F}_X} )''(t)     \right| +  \sup_{\left| t \right| \leq DT_n }  \left| ( \log \varphi_{ F_X }  )''(t)  \right| \\ & \leq D \big( \epsilon_n^2 + T_n^K \big) \\ & \leq D T_n^K .
\end{align*}
From the preceding bound and a similar expansion to the expression used for $(\log \varphi_{F_X})''  -  (\log \varphi_{ \overline{F}_X})'' $, it follows that \begin{align*}
    \sup_{\left| t \right| \leq  DT_n}  \left| ( \log  \varphi_{ \overline{F}_X } )''  (t) - ( \log \varphi_{F'})''(t)     \right| \leq D \epsilon_n^2.
\end{align*}
Fix any $L_3 > L_2$  sufficiently large such that $$   T_n^{K+1} n^{2 C_1}  \epsilon_n^{L_3}  \leq \epsilon_n^2 \;, \;
n^{2C_1} \sqrt{n}  \epsilon_n \sqrt{\log n} \epsilon_n^{L_3/2}  \leq \epsilon_n^2  \:,\: n^{2 C_1} (\log \epsilon_n^{-1})^{2/ \chi} T_n \epsilon_n^{L_3 }  \leq \epsilon_n^2 .  $$
Define $V_i = \{ t \in I_{n} :  \| t - \mu_i  \| \leq  \epsilon_n^{L_3}  \} $ for $i=1,\dots,N$ and set $V_0 = \R \setminus \bigcup_{i=1}^N V_i$. From the definition of the $\{\mu_i \}_{i=1}^N$, it follows that $\{V_0, V_1, \dots , V_N  \}$ is a disjoint partition of $\R$. By Lemma \ref{aux3}, for any distribution $P$ that satisfies $  (\int_{\R} \left| x \right|^4 dP(x))^{1/4} \leq \sqrt{n} \epsilon_n \sqrt{\log n}  $ and $ \sum_{j=1}^N \left| P(V_j) - p_j   \right| \leq \epsilon_n^{L_3} $, we have (using a similar expansion as above)   \begin{align*}
    \sup_{ \left| t \right| \leq  D T_n}  \left| ( \log  \varphi_{ F' } )''  (t) - ( \log \varphi_{P})''(t)     \right| \leq D \epsilon_n^2.
\end{align*}
From combining all the preceding bounds, observe that all such $P$  also satisfy \begin{align*} &   
\sup_{ \left| t \right| \leq DT_n}  \left| ( \log  \varphi_{ P } )'' (t) -  (\log \varphi_{F_X} )''(t)     \right| \leq D \epsilon_n^2
\\ & 
\sup_{\left| t \right|  \leq D T_n } \left|  ( \log \varphi_P )'' (t)   \right|   \leq D \epsilon_n^2 + \sup_{\left| t \right|  \leq DT_n } \left| ( \log \varphi_{F_X})''(t)  \right| \leq D T_n^K .
\end{align*}
Given any $\sigma^2  > 0$, the preceding bound implies \begin{align*}
    \sup_{\left| t \right| \leq D T_n } \left|  (\log \varphi_{P,\sigma^2})''(t) \right|  &  \leq \sup_{\left| t \right| \leq D T_n } \bigg[ \left| ( \log \varphi_{P})''(t)   \right| + \left|  (\log \varphi_{\sigma^2} )''(t)   \right|   \bigg]  \\ &  \leq D  \big( T_n^K +  \sigma^2 \big).
\end{align*}
This is bounded by $D T_n^K$ if $\sigma^2 \leq D$. From combining the preceding bounds, it follows that Assumption \ref{sampling-uncert} holds with the set $$ \mathcal{S}_n = \bigg \{ \phi_{P,\sigma^2} : \|  \sigma^2 - \sigma_n^2 \| \leq \epsilon_n^2 , \int_{\R} \| x \|^4 dP(x) \leq n^2 \epsilon_n^4 (\log n)^2 \: ,   \sum_{j=1}^N \left| P(V_j) - p_j   \right| \leq \epsilon_n^{L_3}   \bigg  \}.  $$
Next, we verify that Assumption \ref{weak-bias} holds. Let $F_n = \phi_{F', \sigma_n^2}$ where $F'$ is as specified above. Since $\log \varphi_{F_n} = \log \varphi_{F'} + \log \varphi_{\sigma_n^2}$, it follows that \begin{align*}
    \mathcal{G}(\mathbb{P}_{Y}, \varphi_{F_n})  - \mathcal{G}(\mathbb{P}_{Y}, \varphi_{F_X})  & =   \varphi_{\mathbf{Y}}^2 \big[(\log \varphi_{F_{n}})'' - (\log \varphi_{F_{X}})''] \\ & = \varphi_{\mathbf{Y}}^2 \big[ (\log \varphi_{F'})'' - (\log \varphi_{F_{X}})''   \big] + \varphi_{\mathbf{Y}}^2  (\log \varphi_{\sigma_n^2})''.
\end{align*}
From the preceding bounds for $F'$ and the definition of $(\sigma_n^2, \epsilon_n , T_n)$, we obtain
\begin{align*}
   \|  \mathcal{G}(\mathbb{P}_{Y}, \varphi_{F_n})  - \mathcal{G}(\mathbb{P}_{Y}, \varphi_{F_X}) \|_{\mathbb{B}(T_n)} & \leq \| (\log \varphi_{F'})'' - (\log \varphi_{F_{X}})''    \|_{\mathbb{B}(T_n)} + \| (\log \varphi_{\sigma_n^2})''  \|_{\mathbb{B}(T_n)} \\ & \leq D \big[\epsilon_n^2 T_n^{1/2}  + \sigma_n^2 T_n^{1/2}  \big] \\ & \leq D \epsilon_n.
\end{align*}
It follows that Assumption \ref{weak-bias} holds with $F_n = \phi_{F',\sigma_n^2}$. For Assumption \ref{loc-conc}, we define $\mathcal{R}_n$ to be the set   \begin{align*}
   \mathcal{R}_n = \bigg \{  \phi_{P,\sigma^2} : \|  \sigma^2 - \sigma_n^2 \| \leq \epsilon_n^2  \: ,   \sum_{j=1}^N \left| P(V_j) - p_j   \right| \leq \epsilon_n^{L_3}   \bigg  \} .
\end{align*}
Note that for any $ \phi_{P,\sigma^2} \in \mathcal{R}_n$, we have  \begin{align*}
   & \|  \mathcal{G}(\mathbb{P}_{Y}, \varphi_{F_n})  - \mathcal{G}(\mathbb{P}_{Y}, \varphi_{P,\sigma^2}) \|_{\mathbb{B}(T_n)} \\ &  \leq  \| (\log \varphi_{F'})'' - (\log \varphi_{P})''    \|_{\mathbb{B}(T_n)} + \| (\log \varphi_{\sigma_n^2})''- (\log \varphi_{\sigma^2})''  \|_{\mathbb{B}(T_n)}  \\ & \leq D \big[\epsilon_n^2 T_n^{1/2}  + \epsilon_n^2 T_n^{1/2}  \big] \\ & \leq D \epsilon_n.
\end{align*}
Furthermore, we have that \begin{align*}
    \int_{ \mathcal{R}_n \setminus \mathcal{S}_n }  d \nu_{\alpha,G} (P,\sigma^2) &  \leq \int_{ P : \int_{\R} \left| x \right|^4 dP(x) >  n^2 \epsilon_n^4 (\log n)^2 }  d \nu_{\alpha,G} (P,\sigma^2) \\ & \leq  \int_{ P : \int_{\R} \left| x \right|^4 dP(x) > n^2 \epsilon_n^4 (\log n)^2  }  d \text{DP}_{\alpha}(P) \\ &  \leq \exp(-D n \epsilon_n^2 \log n) 
\end{align*}
where the second inequality is due to $\nu_{\alpha,G}$ being a product measure $\nu_{\alpha,G} = \text{DP}_{\alpha} \otimes G$ and the third inequality follows from an application of Lemma \ref{aux12}. Assumption \ref{loc-conc} $(ii)$ follows.

To finish verifying Assumption \ref{loc-conc} $(i)$, note that \begin{align*}
    \int_{\mathcal{R}_n} d \nu_{\alpha,G} (P,\sigma^2) = \int_{ \sigma^2 : \| \sigma^2 - \sigma_n^2  \| \leq  \epsilon_n^2} \int_{P : \sum_{j=1}^N \left| P(V_j) - p_j   \right| \leq \epsilon_n^{L_3}    } d  \text{DP}_{\alpha}(P)  d G(\sigma^2).
\end{align*}
As $ \text{DP}_{\alpha} $ is constructed using a Gaussian base measure $\alpha $, it is straightforward to verify that $ 
\inf_{j=1}^N \alpha(V_j) \geq C \epsilon_n^{L_3 } \exp(- C' ( \log \epsilon_n^{-1})^{2/\chi}  ) $ for universal constants $C,C' > 0$. By definition of $ \text{DP}_{\alpha}$, $(P(V_1),\dots, P(V_N)) \sim \text{Dir}(N, \alpha(V_1), \dots,  \alpha(V_N))$. As $N = D \{ \log ( \epsilon_n^{-1} )  \}^{\iota}$, an application of \citep[Lemma G.13]{ghosal2017fundamentals} implies  \begin{align*}
    \int_{P : \sum_{j=1}^N \left| P(V_j) - p_j   \right| \leq \epsilon_n^{L_3}    } d  \text{DP}_{\alpha}(P) \geq C \exp \big (  - C'  ( \log \epsilon_n^{-1} )^{\iota + \max \{2/ \chi,1   \}  }   \big)  & = C \exp \big (  - C'  ( \log \epsilon_n^{-1} )^{\lambda  }   \big) \\ & \geq  C \exp \big (  - C''  n \epsilon_n^2   \big) 
\end{align*}
It remains to bound the outer integral. By Assumption \ref{cov-prior}, there exists universal constant $C,C',C'' > 0$ such that \begin{align*}
     \int_{ \sigma^2 : \| \sigma^2 - \sigma_n^2  \| \leq \epsilon_n^2  } d G(\sigma^2)  \geq  C \exp \big( - C' \sigma_n^{-2 \kappa}    \big) \geq C \exp(-C' \sigma_n^{-2})    \geq  C \exp ( -C'' n \epsilon_n^2  ).
\end{align*}
Assumption \ref{loc-conc} follows. By Theorem \ref{main-contract}, we obtain \begin{align*}
    \nu \bigg ( F : \|  \widehat{\varphi}_{\mathbf{Y}}^2 \big[\mathbf{Q}_k^*  \widehat{\mathcal{V}}_{\mathbf{Y}}  - (\log \varphi_F)''   ]  \|_{\mathbb{B}(T_n)}  > D \epsilon_n  \: \bigg| \: \mathcal{D}_n  \bigg ) = o_{\mathbb{P}}(1).
\end{align*}
\item[\textbf{$(ii)$}] 
Observe that $ \epsilon_n^2 \asymp n^{-2/3} T_n^{1/3} $. We aim to apply Corollary \ref{contract-strong} with the metric $d = \mathbf{W}_1$. Given any $\delta > 0$ and $s > 1$, define the set \begin{align*}
    \mathcal{H}_n = \bigg \{ F = \phi_{P, \sigma^2} :  \sigma^{\frac{(s-1)}{1+2s} + 1} \leq \epsilon_n^{\frac{-2(s-1)}{1+2s}} n^{- \delta}  \: , \: 
 \bigg( \int \| x \|^s dP(x) \bigg)^{\frac{(s-1)}{s(1+2s)} + 1/s} \leq  \epsilon_n^{\frac{-2(s-1)}{1+2s}} n^{-\delta}     \bigg \}.
\end{align*}
By an analogous argument to Theorem \ref{deconv-1}, we can pick $s > 0$ large enough and $\delta > 0$ small enough such that $\nu(\mathcal{H}_n^c)  \lessapprox \exp(-B_n)$ for some $n \epsilon_n^2 = o(B_n) $, so that the conditions of  Corollary \ref{contract-strong} are satisfied. Let $K: \R \rightarrow (0, \infty)$ denote any symmetric density function on $\R$ with Fourier transform $\mathcal{F}(K)$  that is bounded, has support contained in $[-1,1]$ and has finite moments of order $s$. Define $K_R(x) = R K(Rx)$ for every $R > 0$. By an analogous argument to Theorem \ref{deconv-1}, we obtain
\begin{align*}
    \mathbf{W}_1(F,F_X) \leq D  \bigg[ \epsilon_n^{\frac{-2(s-1)}{1+2s}} n^{- \delta}  \| \varphi_{F \star K_R} - \varphi_{F_X \star K_R}    \|_{L^2}^{\frac{2(s-1)}{1+2s}} + R^{-1}   \bigg].
\end{align*}
for every $F \in \mathcal{H}_n$. Now, it remains to bound the quantity above. Suppose first that the distribution $F$ is demeaned so that $\E_{X \sim F}(X) = 0$ and $(\log \varphi_F)'(0)  = 0$. The fundemental theorem of calculus and Cauchy-Schwarz imply 
 \begin{align*}
    \left| (\log \varphi_{F_X})' (t) - (\log {\varphi}_{F})'(t)           \right| & = \left| \int_0^t  \big[ \big(\log  \varphi_{F_X} \big)''(s) - (\log  {\varphi}_{F} )''(s)     ]  ds  \right| \\ & \leq \sqrt{R} \|  \big(\log  \varphi_{F_X} \big)''(.) - \big(\log  {\varphi}_{F} \big)''(.)         \|_{ \mathbb{B}( R) }
\end{align*}
for every $t \in \mathbb{B}(R)$. Similarly, as all characteristic functions satisfy $\log \varphi_{F} (0) = 0 $, we have \begin{align*}
    \left| \log \varphi_{F_X} (t) - \log {\varphi}_{F}(t)           \right| & = \left| \int_0^t  \big[  ( \log \varphi_{F_X})' (s) - ( \log {\varphi}_{F})' (s)    ]  ds  \right| \\ & \leq  \sqrt{R} \| \big(\log  \varphi_{F_X} \big)'(.) - \big(\log  {\varphi}_{F} \big)'(.)     \|_{\mathbb{B}(R)} 
\end{align*}
for every $t \in \mathbb{B}(R)$. For every fixed $t \in \R$, the mean value theorem implies  $ 
    \left|  \varphi_{X}(t) -   \varphi_{F}(t)   \right|  \leq  \left|  \log \varphi_{X} (t) - \log \varphi_{F}(t)    \right|$. In particular \begin{align*}
    \|   \varphi_{F_X} -  {\varphi}_{F}  \|_{\mathbb{B}(R)} \leq D R^2 \|  \big(\log  \varphi_{F_X} \big)''(.) - \big(\log  {\varphi}_{F} \big)''(.)  \|_{\mathbb{B}(R)} 
\end{align*}
for any mean-zero distribution $F$. By assumption, we have $\inf_{\| t \|_{\infty} \leq R} \left| \varphi_{\mathbf{Y}}   \right| \geq b \exp(-B R^{2})   $ for some constants $b,B> 0$. Denote the elements of $\mathbf{A}_k$ by   $\mathbf{A}_k = (a_1,\dots,a_L)$. Fix any $i$ such that $a_i \neq 0$. Without loss of generality, let $i=1$ and $a_i > 0$.  Let $R = R_n =  \min \{ c_0 (\log n)^{1/2} , a_1 T_n \} $ for a constant $c_0$ sufficiently small such that $ 4 B c_0^{2} \max \{ (s-1) (1+2s)^{-1},1 \} = \xi < \delta < 1/2 $. An application of Lemma \ref{aux2} implies that \begin{align*}
   \sup_{\| t \|_{\infty} \leq R_n}   \frac{  \left| \widehat{\varphi}_{\mathbf{Y}}(t)  - \varphi_{\mathbf{Y}}(t)    \right| }{\left| \varphi_{\mathbf{Y}}(t) \right|} \leq D n^{-1/2 + \xi} \sqrt{\log \log  n}
\end{align*}
with $\mathbb{P}$ probability approaching $1$. As the quantity on the right converges to zero, it follows that
 $
     0.5 \left| \varphi_{\mathbf{Y}}(t) \right| \leq \left| \widehat{\varphi}_{\mathbf{Y}}(t)   \right| \leq 1.5 \left| \varphi_{\mathbf{Y}}(t) \right|
$ for every $t \in \mathbb{B}(R_n)$. By combining the preceding bounds, it follows that
\begin{align*}
    \| \varphi_{F \star K_{R_n}} - \varphi_{F_X \star K_{R_n}}    \|_{L^2}^2 & = \int_{\mathbb{B}(R_n)} \left| \varphi_F(t) - \varphi_X(t)   \right|^2 \left| \mathcal{F}(K)(t)   \right|^2 dt \\ & \leq D \int_{\mathbb{B}(R_n)} \left| \varphi_F(t) - \varphi_X(t)   \right|^2  dt \\ & \leq D R_n^4 \| (\log  \varphi_{X})''  - ( \log \varphi_{F})''      \|_{\mathbb{B}(R_n)}^2 \\ & \leq D R_n^4 \sup_{\| t \|_{\infty} \leq R_n} \left| \varphi_{\mathbf{Y}}(t)  \right|^{-4} \| \widehat{\varphi}_{\mathbf{Y}}^2 [(\log  \varphi_{X})''  - ( \log \varphi_{F})'' ) ]    \|_{\mathbb{B}(R_n)}^2 \\ & \leq D R_n^4  \exp(4B R_n^2)   \| \widehat{\varphi}_{\mathbf{Y}}^2 [(\log  \varphi_{X})''  - ( \log \varphi_{F})'' ) ]    \|_{\mathbb{B}(R_n)}^2 \\ & \leq D R_n^4  \exp(4B R_n^2)   \| \widehat{\varphi}_{\mathbf{Y}}^2 [(\log  \varphi_{X})''  - ( \log \varphi_{F})'' ) ]    \|_{\mathbb{B}(a_{1}T_n)}^2 .
\end{align*}
From the definition of $R_n$, it follows that  \begin{align*}
    \mathbf{W}_1(F,F_X) \leq D \big[ R_n^{\frac{4(s-1)}{1+2s}}  \epsilon_n^{\frac{-2(s-1)}{1+2s}} n^{- \delta + \xi} \| \widehat{\varphi}_{\mathbf{Y}}^2 [(\log  \varphi_{X})''  - ( \log \varphi_{F})'' ) ]    \|_{\mathbb{B}(a_{1}T_n)}^{\frac{2(s-1)}{1+2s}}  + R_n^{-1}    ] .
\end{align*}
It remains to to compute the modulus over $\mathcal{H}_n$. Fix any $F$ that satisfies $\|  \widehat{\varphi}_{\mathbf{Y}}^2 \big[\mathbf{Q}_k^*  \widehat{\mathcal{V}}_{\mathbf{Y}}  - (\log     \varphi_{F})''   ]  \|_{\mathbb{B}(T_n)} \leq D \epsilon_n $. Since $ {\varphi}_{\mathbf{Y}}^2\mathbf{Q}_k^* \mathcal{{V}}_{\mathbf{Y}} = {\varphi}_{\mathbf{Y}}^2 (\log \varphi_X)'' $ and $ \sup_{t \in \mathbb{B}(T_n)} \left| ( \log \varphi_X )''(t) \right|\lessapprox T_n^K$, a straightforward application of Lemma \ref{aux2} implies that \begin{align} \label{weak-ineq} \| \widehat{\varphi}_{\mathbf{Y}}^2 [ ( \log  \varphi_{X})''(t'\mathbf{A}_k)  - ( \log \varphi_{F} )'' (t'\mathbf{A}_k)]     \|_{\mathbb{B}(T_n)} \leq D \epsilon_n .\end{align} The distribution of the demeaned posterior is given by
\begin{align*} \overline{\nu}_{}(. \big|\:\mathcal{D}_n) \sim  Z - \E[Z] \; \; \; \; \; \; \text{where} \; \; \; \; \;  Z  \sim \nu_{}\big(. \big|\:\mathcal{D}_n\big) . \end{align*}
Denote a demeaned version of a Gaussian mixture $\phi_{P,\sigma^2}$ by $\overline{\phi}_{P,\sigma^2}$. We refer to the demeaned characteristic function by $ \overline{\varphi}_{P,\sigma^2} $. For any distribution $Z$, we have $(\log \varphi_{Z} )'' = (\log \varphi_{Z - \E[Z]} )''$. In particular, for any $F = \phi_{P,\sigma^2}$ satisfying (\ref{weak-ineq}),  the same inequality holds with $\overline{\varphi}_{P,\sigma^2} $. Recall, we denote the elements of $\mathbf{A}_k$ by   $\mathbf{A}_k = (a_1,\dots,a_L)$. Consider the change of variables \begin{align*}
    z_1 = t' \mathbf{A}_k \; , \; z_2 = t_2 \; , \dots , z_L = t_L.
\end{align*}
The Jacobian of the change of variables $(t_1,\dots,t_L) \rightarrow (z_1,\dots,z_L)$ is given by $ J(z_1,\dots,z_L) =   a_1^{-1}$. Let  $c_L = \inf_{t \in \mathbb{B}(T_n)} t' \mathbf{A}_k $ and $c_U = \sup_{t \in \mathbb{B}(T_n)}  t' \mathbf{A}_k$. It follows that for any non-negative Borel function $f : \R \rightarrow \R_{+}$, we have that \begin{align*}
    \int_{\mathbb{B}(T_n)} f(t' \mathbf{A}_k) dt = \left| a_1 \right|^{-1} (2 T_n)^{L-1}  \int_{c_L}^{c_U} f(z_1)  dz_1.
\end{align*}
In particular, since $ c_U \geq  a_1 T_n $ and $c_L \leq - a_1 T_n $, we have $ \| f   \|_{\mathbb{B}(  a_1  T_n)}^2 \leq D T_n^{1-L}  \| f(t' \mathbf{A}_k)  \|_{\mathbb{B}(T_n)}^2  $ for some universal constant $D > 0$. From combining the preceding bounds, it follows that the modulus satisfies  $$ \omega(d,\mathcal{H}_n,\epsilon_n) \leq D \bigg[ R_n^{\frac{4(s-1)}{1+2s}} T_n^{\frac{(s-1)(1-L)}{1+2s}} n^{- \delta + \xi} + R_n^{-1}  \bigg] \leq D R_n^{-1}. $$
Since $ R_n =  \min \{ c_0 (\log n)^{1/2} , a_1T_n \} $, the claim follows.

\end{enumerate}

    \end{proof}

\begin{proof}[Proof of Theorem \ref{multfac-2}]

The proof proceeds through several steps. Let $\epsilon_n^2 = n^{-1} (\log n)^{\lambda}$.
\begin{enumerate}
    \item[\textbf{$(i)$}] We aim to apply Theorem \ref{main-contract}. The verification of Assumptions \ref{sampling-uncert} - \ref{loc-conc}  is completely analogous to the argument in Theorem \ref{repmeas-1} and omitted. Specifically, Assumptions \ref{sampling-uncert} and \ref{loc-conc} hold with the sets \begin{align*} & \mathcal{S}_n = \bigg \{ \phi_{P,\sigma^2} : \|  \sigma^2 - \sigma_0^2 \| \leq \epsilon_n^{L_3} , \int_{\R} \| x \|^4 dP(x) \leq n^2 \epsilon_n^4 (\log n)^2 \: ,   \sum_{j=1}^N \left| P(V_j) - p_j   \right| \leq \epsilon_n^{L_3}   \bigg  \} , \\ & \mathcal{R}_n = \bigg \{  \phi_{P,\sigma^2} : \|  \sigma^2 - \sigma_0^2 \| \leq \epsilon_n^{L_3}  \: ,   \sum_{j=1}^N \left| P(V_j) - p_j   \right| \leq \epsilon_n^{L_3}   \bigg  \} ,
    \end{align*}
for some fixed constant $ L_3 > 0$ sufficiently large, $N = D(\log \epsilon_n^{-1})^{\max \{ 1, \chi^{-1} +1/2  \}}$, and sets $\{V_0,\dots,V_N \}$  that form a disjoint partition of $\R$. By Theorem \ref{main-contract}, we obtain \begin{align*}
    \nu \bigg ( F : \|  \widehat{\varphi}_{\mathbf{Y}}^2 \big[\mathbf{Q}_k^*  \widehat{\mathcal{V}}_{\mathbf{Y}}  - (\log \varphi_F)''   ]  \|_{\mathbb{B}(T_n)}  > D \epsilon_n  \: \bigg| \: \mathcal{D}_n  \bigg ) = o_{\mathbb{P}}(1).
\end{align*}
\item[\textbf{$(ii)$}] 
We aim to apply Corollary \ref{contract-strong} with the metric $d = \| . \|_{L^2}$. The law of  $G = G_n$ is given by $\Omega / \sigma_n^2$ where $\Omega \sim L$  and $L$ is a probability measure on  $\mathbf{S}_+^d$ that  satisfies Assumption \ref{cov-prior}. By Assumption \ref{cov-prior}, it follows that for every $E' > 0$, there exists $E > 0$ such that
\begin{align*}
  \int_{\sigma^2 : \sigma^{-2} > E \sigma_n^2 (n \epsilon_n^2)^{1/ \kappa}  }   d G(\sigma^2) = \int_{\sigma^2 : \sigma^{-2} > E (n \epsilon_n^2)^{1/ \kappa}  }   d L(\sigma^2) \leq \exp \big(  - E' n \epsilon_n^2  \big).
\end{align*}
In particular, we can choose $E > 0$ large enough such that the hypothesis of Corollary \ref{main-contract} holds with the set of distributions $ \mathcal{H}_n = \{ \phi_{P,\sigma^2} :  \sigma^{-2} \leq E \sigma_n^2 (n \epsilon_n^2)^{1/ \kappa}    \}. $ It remains to compute the modulus  over $\mathcal{H}_n$. By assumption, we have $\inf_{\| t \|_{\infty} \leq R} \left| \varphi_{\mathbf{Y}}   \right| \geq b \exp(-B R^{2})   $ for some constants $b,B> 0$. Denote the elements of $\mathbf{A}_k$ by   $\mathbf{A}_k = (a_1,\dots,a_L)$. Fix any $i$ such that $a_i \neq 0$. Without loss of generality, let $i=1$ and $a_i > 0$.  Fix any $ R_n \leq a_1 T_n $. As in the proof of Theorem \ref{multfac-1}, we can deduce \begin{align*}
   \int_{\mathbb{B}(R_n)} \left| \varphi_F(t) - \varphi_X(t)   \right|^2  dt  \leq D R_n^4  \exp(4B R_n^2)   \| \widehat{\varphi}_{\mathbf{Y}}^2 [(\log  \varphi_{X})''  - ( \log \varphi_{F})'' ) ]    \|_{\mathbb{B}(a_1 T_n)}^2
\end{align*}
for any $F$. Fix any $F$ that satisfies $\|  \widehat{\varphi}_{\mathbf{Y}}^2 \big[\mathbf{Q}_k^*  \widehat{\mathcal{V}}_{\mathbf{Y}}  - (\log     \varphi_{F})''   ]  \|_{\mathbb{B}(T_n)} \leq D \epsilon_n $. Since $ {\varphi}_{\mathbf{Y}}^2\mathbf{Q}_k^* \mathcal{{V}}_{\mathbf{Y}} = {\varphi}_{\mathbf{Y}}^2 (\log \varphi_X)'' $ and $ \sup_{t \in \mathbb{B}(T_n)} \left| ( \log \varphi_X )''(t) \right|\lessapprox T_n$, a straightforward application of Lemma \ref{aux2} implies that $ \| \widehat{\varphi}_{\mathbf{Y}}^2 [ ( \log  \varphi_{X})''(t'\mathbf{A}_k)  - ( \log \varphi_{F} )'' (t'\mathbf{A}_k)]     \|_{\mathbb{B}(T_n)} \leq D \epsilon_n $. As in the proof of Theorem \ref{multfac-1}, we have $ \| f   \|_{\mathbb{B}(  a_1  T_n)}^2 \leq D T_n^{1-L}  \| f(t' \mathbf{A}_k)  \|_{\mathbb{B}(T_n)}^2  $ for any non-negative Borel function $f : \R \rightarrow \R_{+}$. Therefore, for all such $F$, we also obtain  \begin{align*}
   \int_{\mathbb{B}(R_n)} \left| \varphi_F(t) - \varphi_X(t)   \right|^2  dt  \leq D R_n^4  \exp(4B R_n^2)   \epsilon_n^2.
\end{align*}
Next, we examine the bias from truncating the $L^2$ norm to the set $\mathbb{B}(R)$. Suppose $ \sigma^{-2} \leq E^2  \sigma_n^2 (n \epsilon_n^2)^{1/\kappa}  $ holds. Since $\sigma_n^{2} \asymp (n \epsilon_n^2)^{-1/\kappa}$, it follows that there exists a $c > 0$ for which $  \sigma^2 \geq c  $ holds. It follows that there exists  universal constants $C,D > 0$ such that
\begin{align*}
    \| \big( {\varphi}_{X}   -   {\varphi}_{P,\sigma^2} \big) \mathbbm{1} \{ \| t \|_{\infty} > R  \}   \|_{L^2}^2 & \leq  2 \|  {\varphi}_{X}    \mathbbm{1} \{ \| t \|_{\infty} > R  \}   \|_{L^2}^2 +  2 \|     {\varphi}_{P,\sigma^2}  \mathbbm{1} \{ \| t \|_{\infty} > R  \}   \|_{L^2}^2 \\ & \leq 2 \int_{\| t \|_{\infty} > R}  e^{- t^2 \sigma_0^2} dt + 2 \int_{\| t \|_{\infty} > R} e^{- t^2 \sigma^2} dt \\ & \leq 2 \int_{\| t \|_{\infty} > R}  e^{- t^2 \sigma_0^2} dt + 2 \int_{\| t \|_{\infty} > R} e^{- c t^2     } dt \\ & \leq D   e^{- C R^2} R^{-1}  .
\end{align*}
It follows that for any sequence $R_n \leq a_1T_n$, the modulus satisfies \begin{align*}
    \omega^2(d,\mathcal{H}_n,\epsilon_n) \leq D \big[ R_n^4  \exp(4B R_n^2) \epsilon_n^2 +  e^{- C R^2} R^{-1}  ].
\end{align*}
The claim follows by picking $ R_n =  \min \{ c_0 (\log n)^{1/2} , T_n \} $ for a sufficiently small $c_0 > 0$.

\end{enumerate}

\end{proof}

\end{document}